\documentclass[11pt]{article}
\usepackage{times,fullpage,amssymb,amsmath,amsthm,hyperref,xcolor,graphicx}

\hypersetup{pageanchor=false}

\newcommand{\R}{\mathbb{R}}
\newcommand{\eps}{\varepsilon}

\newcommand{\RRR}{\mathcal{R}}
\newcommand{\FFF}{\mathcal{F}}
\newcommand{\AAA}{\mathcal{A}}
\newcommand{\vdoor}{\text{v-door}}
\newcommand{\hdoor}{\text{h-door}}
\newcommand{\vreach}{\text{v-reach}}
\newcommand{\hreach}{\text{h-reach}}
\newcommand{\vDoors}{\text{v-Doors}}
\newcommand{\hDoors}{\text{h-Doors}}
\newcommand{\vDoorsRep}{\text{v-Doors-ID}}
\newcommand{\hDoorsRep}{\text{h-Doors-ID}}
\newcommand{\vreachrep}{\text{v-reach-ID}}
\newcommand{\hreachrep}{\text{h-reach-ID}}
\newcommand{\vreachinrep}{\text{v-reach-in-ID}}
\newcommand{\hreachinrep}{\text{h-reach-in-ID}}
\newcommand{\vreachoutrep}{\text{v-reach-out-ID}}
\newcommand{\hreachoutrep}{\text{h-reach-out-ID}}
\newcommand{\ID}{\text{ID}}
\newcommand{\RANK}{\text{rank}}
\newcommand{\change}[1]{\widehat{#1}}
\newcommand{\DD}{\change{D}}

\newtheorem{lemma}{Lemma}[section]
\newtheorem{theorem}[lemma]{Theorem}

\newtheorem{operation}{Operation}

\title{Shaving Two Logs for Discrete and Continuous Fr\'echet Distance}

\begin{document}

\author{Timothy M. Chan\thanks{
Siebel School of Computing and Data Science, University of Illinois at Urbana-Champaign, USA (tmc@illinois.edu).
Work supported in part by NSF Grant CCF-2224271.
} \and
Isaac M. Hair\thanks{
Department of Computer Science, University of California, Santa Barbara, USA (hair@ucsb.edu).
}
}

\maketitle
\thispagestyle{empty}
\setcounter{page}{0}

\begin{abstract}
We revisit the classical computational geometry problem of computing
the exact (discrete or continuous) Fr\'echet distance between two given polygonal curves
with $n$ vertices in a constant dimension~$d$.

\begin{enumerate}
\item We present an $O((n^2/\log^2n)\log^2\log n)$-time deterministic algorithm
for the discrete Fr\'echet distance problem, improving the previous
$O((n^2/\log n)\log\log n)$ result by Agarwal, Ben Avraham, Kaplan, and Sharir [SODA'13].

\item We present an $O((n^2/\log^2n)\log^2\log n)$-time deterministic algorithm
for the decision version of the continuous Fr\'echet distance problem,
improving the previous
$O((n^2/\log n)\log^2\log n)$ result by Buchin, Buchin, Meulemans, and Mulzer [SODA'14]
and the recent $O(n^2/\log^{1+\mu}n)$ randomized result by Cheng and Huang [SODA'25] for
a small unspecified constant $\mu>0$ (which converges to 0 as $d$ increases).

\item Finally, we present an $O((n^2/\log^2 n)\log^{O(1)}\log n)$-time Las Vegas randomized algorithm
for the continuous Fr\'echet distance problem, improving the previous 
$O(n^2\log^2\log n)$ randomized result by Buchin, Buchin, Meulemans, and Mulzer [SODA'14]
and the recent $O(n^2/\log^{\mu}n)$-time randomized result claimed by Cheng and Huang [SODA'25].
\end{enumerate}

Since these problems have near-quadratic OV-based conditional lower bounds and shaving of two log factors is the best that is known for similar, simpler problems such as longest common subsequence and edit distance,
our results reach the natural limit of the current state of the art (up to $\log\log$ factors).
\end{abstract}

\newpage
\tableofcontents
\thispagestyle{empty}
\newpage

\setcounter{page}{1}
\hypersetup{pageanchor=true}

\section{Introduction}

\newcommand{\Ohat}{\widehat{O}}

The Fr\'echet distance, a standard measure of similarity between two polygonal curves, is a very popular topic of study in computational geometry.
Given two sequences of points $P=\langle p_0,\ldots,p_n\rangle$ and $Q=\langle q_0,\ldots,q_n\rangle$ in Euclidean space $\R^d$ of constant dimension~$d$,
the \emph{discrete Fr\'echet distance} between $P$ and $Q$ is defined as the minimum
of $\max_{i=0}^N \|p_{\rho(i)}-q_{\sigma(i)}\|$ over all 
functions $\rho,\sigma: \{0,\ldots,N\}\rightarrow\{0,\ldots,n\}$ and all $N$ with $\rho(0)=\sigma(0)=0$,
$\rho(N)=\sigma(N)=n$, and $\rho(i\!+\!1)-\rho(i),\sigma(i\!+\!1)-\sigma(i)\in\{0,1\}$ for all $i\in\{0,\ldots,N-1\}$.
The \emph{continuous Fr\'echet distance} between $P$ and $Q$ is defined as the minimum
of $\max_{t\in [0,1]} \|P(\rho(t))-Q(\sigma(t))\|$ over all nondecreasing, continuous functions $\rho,\sigma: [0,1]\rightarrow [0,n]$ with $\rho(0)=\sigma(0)=0$ and $\rho(1)=\sigma(1)=n$, where $P(i+\alpha) := p_i + \alpha(p_{i+1}-p_i)$
and $Q(i+\alpha) := q_i + \alpha(q_{i+1}-q_i)$ for $i\in\{0,\ldots,n-1\}$ and $\alpha\in [0,1]$.
(The analogy with ``walking the dog'' with the shortest leash is often mentioned in prior papers on this topic.)

The discrete Fr\'echet distance problem has a straightforward $O(n^2)$-time algorithm via dynamic programming (DP), similar to textbook algorithms for the longest common subsequence (LCS) or edit distance problem.  In their seminal paper from SoCG'92 (which won the SoCG'22 test-of-time award),
Alt and Godau~\cite{AltG95} obtained an $O(n^2\log n)$-time algorithm for the continuous Fr\'echet
distance problem, by first solving the decision problem (deciding whether the Fr\'echet distance
is at most a given value $r$) in $O(n^2)$ time via DP, and then applying some form of parametric search.

Although Fr\'echet-distance-related papers have regularly appeared in past SODAs, SoCGs, and other algorithm conferences (for example,
on faster approximation algorithms~\cite{BringmannM16,ChanR18,ColombeF21,HorstKOS23,ChengHZ25}, faster exact algorithms for realistic input~\cite{DriemelHW12,ConradiHHO26} or special cases~\cite{Blank26}, different variants of Fr\'echet distance~\cite{DriemelH13,Har-PeledNSS16,AvrahamFKKS15,Har-PeledR14,BuchinOS17,FiltserGMP24,HorstKOS25},
 more complicated optimization problems under the Fr\'echet distance such as finding a translation to
minimize the distance~\cite{BringmannKN21}, data structures to find a curve close to a query curve under the Fr\'echet distance~\cite{AfshaniD18,BringmannDNP22}, algorithm engineering~\cite{BringmannKN19,Har-PeledRR25,ConradiHR26}, etc.), in this paper we are interested in the \emph{original} problem---finding the best theoretical worst-case algorithm to compute the exact (discrete or continuous) Fr\'echet distance for two given polygonal curves.

\paragraph{Previous Results.}
Throughout this paper, we assume the standard real RAM model, where a word can store either an input real number or a $(\log n)$-bit integer, and standard operations on words take constant time (and operations on reals are limited to constant-degree algebraic predicates). This includes bitwise operations and table lookup on integer words, as in earlier log-shaving algorithms for 3SUM and APSP~\cite{Chan20,Chan10}. The main prior algorithms can be summarized as follows:
\begin{itemize}
\item For discrete Fr\'echet distance:
\begin{itemize}
\item
In SODA'13, Agarwal, Ben Avraham, Kaplan, and Sharir~\cite{AgarwalAKS14} presented a slightly subquadratic algorithm for the decision problem running in $\Ohat(n^2/\log^2n)$ time;%
\footnote{Throughout the introduction, we use the $\Ohat$ notation to hide (tiny) $\log\log n$ factors, to make time bounds look cleaner. We will spell out the number of $\log\log n$ factors in our theorem statements later.  Also, for simplicity, we focus on just the balanced case when $P$ and $Q$ have $O(n)$ vertices.  In the imbalanced case when $P$ has $m$ vertices and $Q$ has $n$ vertices, $n^2$ becomes $mn$ in the time bounds in all previous results as well as our new result (with some of the $\log n$ factors in the denominators replaced by $\log\min\{m,n\}$).} 
they then obtained an algorithm for the original optimization problem 
running in $\Ohat(n^2/\log n)$ time via binary search (and a known distance selection algorithm).
\end{itemize}
\item
For continuous Fr\'echet distance: 
\begin{itemize}
\item In SODA'14, Buchin, Buchin, Meulemans, and Mulzer~\cite{BuchinBMM17}  
gave a slightly subquadratic algorithm for the decision problem running in $\Ohat(n^2/\log n)$ time; they then obtained,
in the $d=2$ case, a (randomized) algorithm for the optimization problem running in $\Ohat(n^2)$ time.
\item 
Recently, in SODA'25, Cheng and Huang~\cite{ChengH25} presented a slightly improved (randomized) decision algorithm running in $O(n^2/\log^{1+\mu}n)$ time for some unspecified small constant
$\mu>0$ (which converges to 0 as a function of $d$); they then claimed that this implies a slightly subquadratic (randomized) algorithm for the optimization problem running in $O(n^2/\log^\mu n)$ time.\footnote{
Cheng and Huang's paper~\cite{ChengH25} did not provide details on how to reduce the optimization problem to the decision problem, however.  They cited Alt and Godau's paper~\cite{AltG95}, but as the authors in~\cite{BuchinBMM17} pointed out earlier, Alt and Godau's reduction has an extra overhead cost of $O(n^2\log n)$ and so cannot be directly used if we aim for a subquadratic algorithm.
See Section~\ref{sec:reduce} for more information.
}
\end{itemize}
\end{itemize}

\paragraph{Why Shaving Logs?}
All the preceding algorithms are about shaving logarithmic factors, for good reasons:
in a landmark paper in fine-grained complexity, Bringmann~\cite{Bringmann14} (see also \cite{BuchinOS19}) proved that
if an $O(n^{2-\eps})$-worst-case-time algorithm were to exist for the discrete or continuous Fr\'echet
distance problem for some constant $\eps>0$, then the Orthogonal Vectors (OV) Hypothesis, and consequently the Strong Exponential Time Hypothesis (SETH), would be false (this result paved
the way for subsequent conditional lower bounds for LCS, edit distance, and other related problems~\cite{BackursI18,AbboudBW15,BringmannK15}).  In fact, Abboud, Hansen, Vassilevska Williams, and Williams~\cite{AbboudHWW16} later showed for a sufficiently
large constant $c$ that the existence of an $O(n^2/\log^c n)$-time algorithm would imply a new lower bound about circuit complexity.
These hardness results give more urgency to the question of exactly how many logs can be shaved
for Fr\'echet distance. In particular, shaving a second log brings this fundamental geometric problem to the same bound, up to $\log\log n$ factors, as LCS, edit distance, and 3SUM.

Ever since the appearance of the ``four Russians'' method~\cite{ArlazarovDKF70}, there has been much interest in shaving logs for fundamental problems in fine-grained
complexity, e.g., \cite{Fredman76,FederM95,Chan10,Chan12,Chan15,BarbaCILOS19,Yu18,GronlundP18,Chan20}.  For LCS and edit distance, $\Ohat(n^2/\log^2n)$-time algorithms
were known~\cite{MasekP80,BilleF08,Grabowski16}.  Current state of the art suggests that two log factors
are the most we could hope to shave.  While the authors in~\cite{AgarwalAKS14} had already succeeded in 
shaving two logs for the decision version of discrete Fr\'echet distance,
Cheng and Huang's work shaved only a $\log^{1+\mu}n$ factor for the decision version of continuous
Fr\'echet distance, and so leaves room for further improvement.
For the original optimization version of discrete or continuous Fr\'echet distance,
current approaches all use some form of parametric or binary search and inherently 
lose one log factor,\footnote{
There are some specialized techniques for reducing optimization to decision problems
that do not lose log factors~\cite{Chan99}, but the Fr\'echet distance problem does not fit in these
settings.
} 
and so $n^2/\log n$ time seems the best we could hope for, one might think\ldots

\paragraph{New Results.}
In this paper, we not only obtain an $\Ohat(n^2/\log^2n)$-time algorithm for the decision version of the
continuous Fr\'echet distance problem, improving Cheng and Huang's work,
but we also obtain $\Ohat(n^2/\log^2n)$-time algorithms for the optimization versions
of both the discrete and continuous Fr\'echet distance problems, surprisingly, without 
losing any logarithmic factor from parametric or binary search!
(The algorithm for continuous Fr\'echet distance is Las Vegas.) 
With time bounds matching the state of the art for LCS and edit distance, our work
essentially closes the decade-old line of research from the authors in~\cite{AgarwalAKS14} and the authors in~\cite{BuchinBMM17} to Cheng and Huang~\cite{ChengH25} (ignoring tiny $\log\log$ factors).

\paragraph{Techniques.}
Not surprisingly, the previous algorithms~\cite{AgarwalAKS14,BuchinBMM17,ChengH25} and our algorithms all share some common ingredients (division into slabs and blocks and/or mini-blocks, constructing geometric arrangements, encoding certain order information to enable bit packing and
table lookups, etc.), but in addition we use several new
interesting ideas. The main new contribution is a conversion from decision to optimization with only a $\log^{O(1)}\log n$ overhead, based on a simple decomposition of the inputs. Our general approach may be useful for other geometric optimization problems where the usual decision to optimization conversions lose a log factor.
\begin{itemize}
\item For discrete Fr\'echet distance (see Section~\ref{sec:discrete}), the key new idea is conceptually simple: We first invoke the decision algorithm only $O(\log\log n)$ times (instead of $O(\log n)$),
which yields an interval that contains the answer and has at most $n^2/\log^C n$
critical values for an arbitrarily large constant $C$ (for the discrete problem, the critical values are pairwise distances).
We then classify blocks that contain such critical values as ``bad'', and blocks that avoid such values as ``good''. On one hand, good blocks can be handled more efficiently; on the other hand,
there are not too many bad blocks (at most $n^2/\log^C n$), so they can be handled naively.
The entire algorithm is describable in around 2 pages (assuming standard geometric data structures and the decision algorithm in~\cite{AgarwalAKS14} as a black box).
\item For the decision version of continuous Fr\'echet distance (see Section~\ref{sec:cont:decis}),
we not only improve but simplify
Cheng and Huang's previous solution~\cite{ChengH25}.  For example, they used point location data structures for hyperplane arrangements in nonconstant dimensions as a subroutine~\cite{Meiser93,EzraHKS20} (which account for the occurrence of $\mu$ in their time bound).
Our solution also considers higher-dimensional arrangements of surfaces but only in the running-time analysis; the algorithm itself avoids explicitly building such higher-dimensional arrangements, and so is simpler.  (In addition, we do not need to ``linearize'' and turn surfaces into hyperplanes.)
\item For the optimization version of continuous Fr\'echet distance (see Section~\ref{sec:cont}),
we again rely on the ``good'' vs. ``bad'' block idea, but the problem gets much tougher (for one thing, critical values are defined by triples instead of pairs), and we need an intricate combination of a slew of new technical ideas.  Still, the heart of the solution is based on some simple geometric observations: as is well-known~\cite{AltG95}, Fr\'echet distance computation can be restated in terms of finding $xy$-monotone paths in a certain space, and we exploit the fact that 
two such paths that cross yield more $xy$-monotone paths
(see Lemmas~\ref{lem:intervalspan},~\ref{lem:sourcedominance}, and~\ref{lem:alivesuccessor}). Such crossing-path properties have been used in earlier Fr\'echet distance algorithms~\cite{AltG95,BenAvrahamKS15,BringmannKN21}; here we use them to preprocess good blocks for fast reachability queries.
\end{itemize}

\section{Discrete Fr\'echet Distance Algorithm}
\label{sec:discrete}
\label{sec:discreteoptimize}

We are given points $p_0,\ldots,p_n,q_0,\ldots,q_n\in \R^d$, with $d=O(1)$.
The discrete Fr\'echet distance problem is equivalent to computing $r^*=D_{n,n}$, where
\[ D_{i,j} = \max\{ \min\{ D_{i-1,j}, D_{i,j-1}, D_{i-1,j-1}\},\ \|p_i-q_j\| \}\] for all $i,j \geq 0$ such that $(i,j)\neq(0,0)$,
with the base cases $D_{0,0}=\|p_0-q_0\|$ and $D_{-1,j}=D_{i,-1}=\infty$.
Naively evaluating the above recursive formula by
dynamic programming (DP) yields an $O(n^2)$-time algorithm, which we would like to improve.

We assume the availability of a decision algorithm which can determine whether $r^*\le r$ for 
any given value $r$ in $T_D(n)$ time. By Agarwal, Ben Avraham, Kaplan, and Sharir's decision algorithm~\cite{AgarwalAKS14},
$T_D(n)=O((n^2/\log^2n)\log\log n)$.

Let $U=\{\|p_i-q_j\|: 0\le i,j\le n\}$ be the multiset of all pairwise distances.
We know that $r^*\in U$.
Although we cannot afford to explicitly enumerate all distances in $U$, we can select the $k$th smallest distance in $U$ in $n^{2-1/O(d)}$ time by known algorithms (using standard
range searching techniques)~\cite{AgarwalASS93,AgarwalE98,ChanZ24}.
As a first step, we either find $r^*$ or find an open interval $(r^-,r^+)$ containing $r^*$, which contains at most $n^2/\log^{3c} n$ distances in $U$ for some sufficiently large constant $c$.  This can be done via a binary search, which requires $O(\log\log n)$ calls to the decision algorithm and selection algorithm, and takes
$O((n^{2-1/O(d)} + T_D(n))\log\log n) = O((n^2/\log^2n)\log^2\log n)$ time. Our plan from here is to solve the optimization problem directly, with no further calls to a decision oracle.\footnote{To detect equality in this and subsequent searches, we also query infinitesimally below a tested value. This costs only constant overhead and also applies to the continuous decision algorithm in the next section.}

Values outside the interval $(r^-,r^+)$ are no longer important.
So, define $\change{x}$ to be $x$ if $x\in (r^-,r^+)$,
or $\infty$ if $x\ge r^+$, or $-\infty$ if $x\le r^-$.
We will compute $\DD_{i,j}$ via DP, sped up by blocking and bit packing techniques. 
Because $r^* \in (r^-, r^+)$, we have $D_{n, n} = \DD_{n, n}$, so this is sufficient to compute the discrete Fr\'echet distance.

We will use two block sizes:
let $B\approx \log^c n$ and $b\approx \eps \log n/\log\log n$ for 
a sufficiently small constant $\eps>0$.  
A \emph{block} refers to a square $[I_0,I_0+B]\times [J_0,J_0+B]$ for some $I_0,J_0$ divisible by $B$.  A \emph{mini-block} refers to a square $[i_0,i_0+b]\times [j_0,j_0+b]$ for some $i_0,j_0$ divisible by $b$.
A \emph{vertical slab} refers to a rectangle $[I_0,I_0+B]\times [0,n]$ for some $I_0$ divisible by $B$. 
(See Figure~\ref{fig:blocks}(a).)
Assume without loss of generality that $n$, $B$, and $b$ are all powers of 2. This implies that each block contains exactly $B^2/b^2$ mini-blocks, and the entire DP table contains $n^2/B^2$ blocks. Observe that each pair of adjacent blocks overlaps along a boundary, and similarly for each pair of adjacent mini-blocks.

\begin{figure}
    \centering
    \includegraphics[scale=0.35]{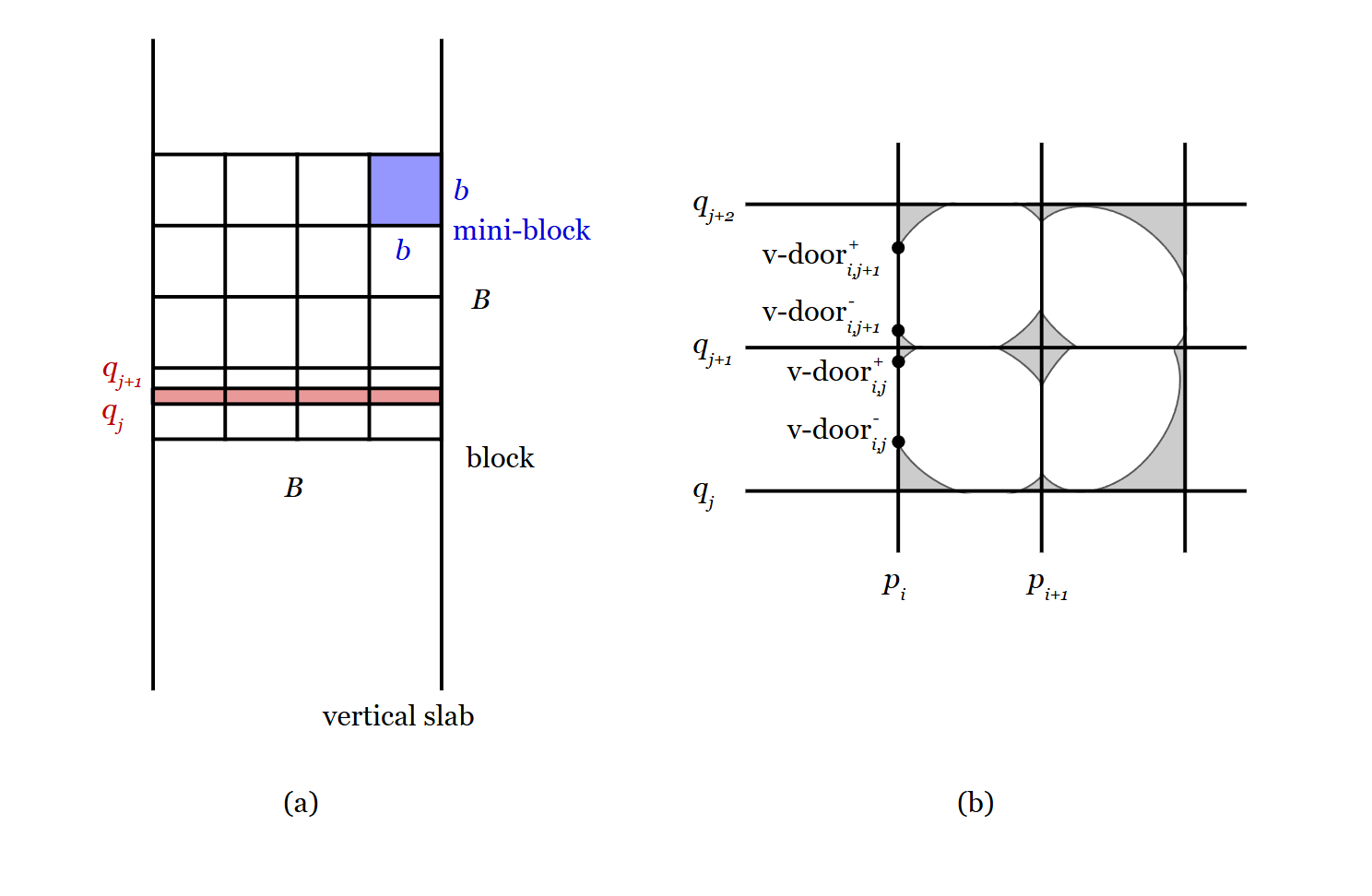}
    \caption{(a) Vertical slabs, blocks, and mini-blocks.  Each row in the slab corresponds to a pair of consecutive points $(q_j, q_{j+1})$.  (b)~Doors at grid edges for the continuous problem.}\label{fig:blocks}
\end{figure}

For each vertical slab $\sigma=[I_0,I_0+B]\times [0,n]$, 
consider the arrangement $\AAA^{(\sigma)}$ of the $B+1$ balls centered
at $p_{I_0},\ldots,p_{I_0+B}$ of radius $r^-$.
Since $d = O(1)$, we can construct in time $B^{O(d)}$ a data structure that allows us to locate any point with respect to the arrangement in $O(\log B)$ time~\cite{AgarwalS00a,ChazelleEGS91}.
For each $j\in\{0,\ldots,n\}$, query this data structure to
find the cell of $\AAA^{(\sigma)}$ containing $q_j$, and define
$\ID^{(\sigma)}_j$ to be a label of this cell. The arrangement has $B^{O(d)}$ cells, so the label can be stored in $O(\log (B^{O(d)}))=O(\log B)$ bits. 
The total time over all vertical slabs $\sigma$ is $O((n/B) (B^{O(d)} + n\log B))= o(n^2/\log^2n)$.
Note that given $\sigma$ and $\ID^{(\sigma)}_j$, we have enough information to deduce whether $\|p_i-q_j\|\le r^-$ or $\|p_i-q_j\|> r^-$ for every $i\in [I_0,I_0+B]$.

For each block $\beta=[I_0,I_0+B]\times [J_0,J_0+B]$, let $U_{\beta} = \{\|p_i - q_j\| : (i, j) \in \beta\}$. We classify $\beta$ as \emph{bad} if $U_{\beta} \cap (r^-, r^+) \neq \emptyset$, or \emph{good} otherwise.
The number of bad blocks is at most $O(n^2/\log^{3c}n)$, because every pair $(i, j)$ is contained within at most 4 blocks.
As with the set $U=\{\|p_i-q_j\|: 0\le i,j\le n\}$, we can select the $k$th smallest distance in $U_{\beta}$ for a given block $\beta$ in time $B^{2-1/O(d)}$. Thus we can test whether a block is bad in $B^{2-1/O(d)}$ time, via binary search.
The total time over all blocks is $O((n/B)^2 B^{2-1/O(d)})=o(n^2/\log^2n)$, assuming $c$ is a sufficiently large constant.

\paragraph{DP over Blocks and Mini-Blocks.}
Consider a block $\beta=[I_0,I_0+B]\times [J_0,J_0+B]$ contained in a vertical slab $\sigma$.
Assume that we have already computed $\DD_{I_0,J_0},\ldots,\DD_{I_0,J_0+B}$ and
$\DD_{I_0,J_0},\ldots,\DD_{I_0+B,J_0}$. (We assume that the values $\DD_{0,0}, \ldots, \DD_{0,n}$ and $\DD_{0,0}, \ldots, \DD_{n,0}$ were already computed naively by DP in time $O(n)$.) There are two cases:

\paragraph{Case 1: $\beta$ is bad.}
We compute $\DD_{I_0+B,J_0},\ldots,\DD_{I_0+B,J_0+B}$
and $\DD_{I_0,J_0+B},\ldots,\DD_{I_0+B,J_0+B}$ naively by DP in $O(B^2)$ time.

\paragraph{Case 2: $\beta$ is good.}
We first sort the list $L^{(\beta)}=\{\DD_{I_0,J_0},\ldots,\DD_{I_0,J_0+B}\}\cup \{\DD_{I_0,J_0},\ldots,\DD_{I_0+B,J_0}\}$ in $O(B\log B)$ time.
Because $\beta$ is good, $\|p_i-q_j\|$ is either at most $r^-$ or at least $r^+$ for every $(i,j)\in\beta$ (and so $\change{\|p_i-q_j\|}$ is either $-\infty$ or $\infty$)---and we can deduce which is the case solely using $\ID_j^{(\sigma)}$ and $\sigma$.
Consequently, for every $(i,j)\in\beta$, we know that $\DD_{i,j}$ must be in $L^{(\beta)}\cup\{-\infty,+\infty\}$.
Let $\RANK^{(\beta)}(x)$ denote the rank of $x$ with respect to $L^{(\beta)}$ (we include two special symbols for when $x = -\infty$ or $x = \infty$). Observe that $\RANK^{(\beta)}(x)$ can be stored using $O(\log B)$ bits.

We build a table to implement the following operation:

\begin{operation}\label{op:discrete}\rm (Processing a mini-block $\mu$)
\begin{tabbing}
\ \ \ \ \ \ \ \= Output:\ \ \=\+\kill
Input: 
\> vertical slab $\sigma=[I_0,I_0+B]\times[0,n]$, \\
\> index $i_0\in [I_0,I_0+B)$ divisible by $b$, \\
\> $\langle\ID^{(\sigma)}_{j_0},\ldots,\ID^{(\sigma)}_{j_0+b}\rangle$,\\  
\> $\langle \RANK^{(\beta)}(\DD_{i_0,j_0}),\ldots,\RANK^{(\beta)}(\DD_{i_0,j_0+b})\rangle$, and\\
\>  $\langle \RANK^{(\beta)}(\DD_{i_0,j_0}),\ldots,\RANK^{(\beta)}(\DD_{i_0+b,j_0})\rangle$\\
\> \quad for a mini-block $\mu=[i_0,i_0+b]\times [j_0,j_0+b]$ contained in a good block $\beta\subset\sigma$.\\
Output: 
\> $\langle \RANK^{(\beta)}(\DD_{i_0+b,j_0}),\ldots,\RANK^{(\beta)}(\DD_{i_0+b,j_0+b})\rangle$, and\\
\> $\langle \RANK^{(\beta)}(\DD_{i_0,j_0+b}),\ldots,\RANK^{(\beta)}(\DD_{i_0+b,j_0+b})\rangle$.
\end{tabbing}
\end{operation}



Note that the input contains enough information to deduce the output, by simulating the standard DP\@.
It is not necessary to be given the index $j_0$, and we do not need to know the specific block $\beta$ for which the ranks were computed.
Each input vector has $O(b\log B)=O(\eps\log n)$ bits and thus fits in a word.
The preprocessing time across all vertical slabs $\sigma$ and choices of $i_0$ is $O((n/B)\cdot (B/b)\cdot  B^{O(1)}\cdot 2^{O(b\log B)})=n^{1+O(\eps)}$, which is strongly subquadratic assuming $\varepsilon > 0$ is a sufficiently small constant. Afterwards, the operation takes $O(1)$ time by table lookup.

As initialization for the block $\beta$, for each $j_0\in [J_0,J_0+B)$ divisible by $b$, we form
the vectors $\langle\ID^{(\sigma)}_{j_0},\ldots,\ID^{(\sigma)}_{j_0+b}\rangle$ and 
$\langle \RANK^{(\beta)}(\DD_{I_0,j_0}),\ldots,\RANK^{(\beta)}(\DD_{I_0,j_0+b})\rangle$
in $O(b)$ time, and pack each vector into a word.
Similarly, for each $i_0\in [I_0,I_0+B)$ divisible by $b$,
we form the vector $\langle \RANK^{(\beta)}(\DD_{i_0,J_0}),\ldots,$ $\RANK^{(\beta)}(\DD_{i_0+b,J_0})\rangle$ and pack it into a word.  The time per block is $O((B/b)\cdot b)=O(B)$.

Processing all mini-blocks $\mu=[i_0,i_0+b]\times [j_0,j_0+b]\subset\beta$ 
in increasing lexicographical order of $(i_0,j_0)$ using Operation~\ref{op:discrete}, 
we obtain $\RANK^{(\beta)}(\DD_{I_0+B,J_0}),\ldots,\RANK^{(\beta)}(\DD_{I_0+B,J_0+B})$ and
$\RANK^{(\beta)}(\DD_{I_0,J_0+B}),\ldots,$ $\RANK^{(\beta)}(\DD_{I_0+B,J_0+B})$. The time per mini-block is $O(1)$, so the time over the entire block is $O(B^2/b^2)$. The resulting ranks can be mapped back to real values in $O(B)$ additional time. 

Repeating this process for all blocks $\beta$ in increasing lexicographical order of $(I_0,J_0)$,
we obtain $\DD_{n,n}=D_{n,n}$.
The total cost of handling Case 1 is $O((n^2/\log^{3c} n)\cdot B^2) = o(n^2/\log^2n)$.
The total cost for handling Case 2 is $O((n/B)^2 (B\log B + B^2/b^2)) = O((n^2/\log^2n)\log^2\log n)$.

\begin{theorem}
The discrete Fr\'echet distance problem in any constant dimension can be solved in\\ $O((n^2/\log^2n)\log^2\log n)$
deterministic time.
\end{theorem}

\section{Continuous Fr\'echet Distance Decision Algorithm}
\label{sec:cont:decis}
\label{sec:continuousdecision}


We are given points $p_0,\ldots,p_n,q_0,\ldots,q_n\in \R^d$, with $d=O(1)$, and a value~$r$.
Let \[ F_{i,j} =\{(i+s,j+t): s,t\in [0,1],\ \|(p_i+s(p_{i+1}-p_i))-(q_j+t(q_{j+1}-q_j))\|\le r\}. \] Each $F_{i,j}$ is a convex region defined by a quadratic inequality, clipped to the grid cell $[i,i+1]\times [j,j+1]$.
The \emph{free region} is $\FFF = \bigcup_{i,j\in [0,n)}F_{i,j}$.
The \emph{reachability region} $\RRR$ consists of all $(x,y)\in [0,n]\times[0,n]$ for which there exists an $xy$-monotone path from $(0,0)$ to $(x,y)$ in $\FFF$. (We say that a path is $xy$-monotone if it consists of a sequence of points whose $x$ coordinates are nondecreasing, and whose $y$ coordinates are nondecreasing.)
The continuous Fr\'echet distance decision problem is equivalent to deciding whether $(n,n)\in\RRR$.

The \emph{door at a grid edge $e$} refers to the intersection
of $\FFF$ with $e$; if nonempty, it is a vertical or horizontal line segment.
Let $\vdoor^-_{i,j}$ and $\vdoor^+_{i,j}$ denote the min and max $y$-coordinates respectively
of the door at the vertical grid edge $\{i\}\times [j,j+1]$, or set $\vdoor^-_{i,j} = \vdoor^+_{i,j} =$ ``-'' if the door is empty.
(See Figure~\ref{fig:blocks}(b).)
Let $\hdoor^-_{i,j}$ and $\hdoor^+_{i,j}$ denote the min and max $x$-coordinates respectively
of the door at the horizontal grid edge $[i,i+1]\times \{j\}$, again setting $\hdoor^-_{i,j} = \hdoor^+_{i,j} =$ ``-'' if the door is empty.
There are constant-complexity piecewise-algebraic functions $f^-,f^+,g^-,g^+$ giving the endpoint coordinates relative to the grid edge, so that
$\vdoor^-_{i,j}=j+f^-(p_i,q_j,q_{j+1},r)$,
$\vdoor^+_{i,j}=j+f^+(p_i,q_j,q_{j+1},r)$,
$\hdoor^-_{i,j}=i+g^-(p_i,p_{i+1},q_j,r)$, and
$\hdoor^+_{i,j}=i+g^+(p_i,p_{i+1},q_j,r)$ when the corresponding doors are nonempty.

The \emph{reach at a grid edge $e$} refers to the intersection
of $\RRR$ with $e$; if nonempty, it is also a vertical or horizontal line segment, contained inside the door at $e$.
Let $\vreach^-_{i,j}$ and $\vreach^+_{i,j}$ denote the min and max $y$-coordinates respectively
of the reach at the vertical grid edge $\{i\}\times [j,j+1]$ (or use ``-'' to indicate that the reach is empty).
Assuming the reach is nonempty, we have 
$\vreach^+_{i,j}=\vdoor^+_{i,j}$.
Let $\hreach^-_{i,j}$ and $\hreach^+_{i,j}$ denote the min and max $x$-coordinates respectively
of the reach at the horizontal grid edge $[i,i+1]\times \{j\}$ (or, as before, use ``-'' to indicate that the reach is empty). Assuming the reach is nonempty, we have
$\hreach^+_{i,j}=\hdoor^+_{i,j}$.  We will compute $\vreach^-_{i,j}$ and $\hreach^-_{i,j}$ via dynamic programming (DP), as these values (along with $\vdoor^+_{i,j}$ and $\hdoor^+_{i,j}$) determine $\vreach^+_{i,j}$ and $\hreach^+_{i,j}$.

To this end, there are simple recursive formulas relating $\vreach^-_{i+1,\allowbreak{}j}$ and $\hreach^-_{i,\allowbreak{}j+1}$ in terms of
$\vreach^-_{i,\allowbreak{}j}$ and $\hreach^-_{i,\allowbreak{}j}$, 
given $\vdoor^-_{i,\allowbreak{}j},\allowbreak{}\vdoor^+_{i,\allowbreak{}j},\allowbreak{}\vdoor^-_{i+1,\allowbreak{}j},\allowbreak{}\vdoor^+_{i+1,\allowbreak{}j},\allowbreak{}\hdoor^-_{i,\allowbreak{}j},\allowbreak{}\hdoor^+_{i,\allowbreak{}j},\allowbreak{}\hdoor^-_{i,\allowbreak{}j+1},\allowbreak{}$ $\hdoor^+_{i,\allowbreak{}j+1}$ (the precise formulas are not important, but can be found e.g.\ in \cite[pp.\ 188]{BuchinBMM17}).
The exact coordinate values of the doors are not important; what matters is how they are ordered.\footnote{All orderings below retain equality classes and the empty-door flags; labels break ties only within an equality class. Ranks distinguish equality with an endpoint from a position strictly between consecutive distinct endpoints. These rules do not change any encoding-size bound.}  In the next two subsections, we will describe how to compactly encode these orderings, to enable bit packing for speedup.

We will use two block sizes:
let $B\approx \log^4 n$ and $b\approx \eps \log n/\log\log n$ for 
a sufficiently small constant $\eps>0$. Assume without loss of generality that $n$, $B$, and $b$ are all
powers of 2. We inherit the definitions of blocks, mini-blocks, and slabs from Section~\ref{sec:discreteoptimize}.
A \emph{block} refers to a square $[I_0,I_0+B]\times [J_0,J_0+B]$ for some $I_0,J_0$ divisible by $B$.  A \emph{mini-block} refers to a square $[i_0,i_0+b]\times [j_0,j_0+b]$ for some $i_0,j_0$ divisible by $b$.
A \emph{vertical slab} refers to a rectangle $[I_0,I_0+B]\times [0,n]$ for some $I_0$ divisible by $B$.  A \emph{horizontal slab} refers to a rectangle $[0,n]\times [J_0,J_0+B]$ for some $J_0$ divisible by $B$.  A block is then the intersection of a vertical and a horizontal slab.

\subsection{Encoding Doors in a Slab}\label{sec:encodeslab}

For a vertical slab $\sigma=[I_0,I_0+B]\times [0,n]$ and each $j\in [0,n)$, define $\vDoors^{(\sigma)}_j$ to be the sorted ordering of $\vdoor^-_{I_0,j},\vdoor^+_{I_0,j},\ldots,\vdoor^-_{I_0+B,j},\vdoor^+_{I_0+B,j}$.

For a fixed $\sigma$,
the number of possible orderings for $\vDoors^{(\sigma)}_j$ is bounded by the number
of cells in the arrangement $\AAA^{(\sigma)}$ of 
surfaces $\{(\eta,\eta',\rho)\in\R^{2d+1}: f^-(p_{i},\eta,\eta',\rho)=f^-(p_{i'},\eta,\eta',\rho)\}$,
$\{(\eta,\eta',\rho)\in\R^{2d+1}: f^+(p_{i},\eta,\eta',\rho)=f^+(p_{i'},\eta,\eta',\rho)\}$,
and $\{(\eta,\eta',\rho)\in\R^{2d+1}: f^-(p_{i},\eta,\eta',\rho)=f^+(p_{i'},\eta,\eta',\rho)\}$,
over all $i,i'\in [I_0,I_0+B]$.\footnote{Notice that these surfaces represent when a door appears/disappears, or when the ordering of two door endpoints changes. Here and below, arrangements of endpoint equalities can also be augmented by the constant-degree polynomial sign tests defining the endpoint branches, nonemptiness, and grid-edge endpoint incidences, although these are not strictly required.}  This is an arrangement of $O(B^2)$ surfaces in $2d+1$ dimensions, which has $B^{O(d)}$ cells, and we can generate all cells and their
corresponding orderings in $B^{O(d)}$ time~\cite{AgarwalS00a}.
For each $j$, we find the cell of $\AAA^{(\sigma)}$ containing $(q_j,q_{j+1},r)$
by point location~\cite{AgarwalS00a,ChazelleEGS91} in $O(\log B)$ time, and define
$\vDoorsRep^{(\sigma)}_j$ to be a label of this cell, which requires $O(\log (B^{O(d)}))=O(\log B)$ bits.
The total time over all vertical slabs $\sigma$ is $O((n/B) (B^{O(d)} + n\log B))= o(n^2/\log^2n)$.
Note that $\vDoors^{(\sigma)}_j$ can be recovered from $\vDoorsRep^{(\sigma)}_j$,
assuming that $\sigma$ is given (the IDs are dependent on $\sigma$).

Similarly, for a horizontal slab $\tau=[0,n]\times [J_0,J_0+B]$,
define $\hDoors^{(\tau)}_i$ to be the sorted ordering
of $\hdoor^-_{i,J_0},\hdoor^+_{i,J_0},\ldots,\hdoor^-_{i,J_0+B},\hdoor^+_{i,J_0+B}$,
and define $\hDoorsRep^{(\tau)}_i$ analogously.
As above, we find $\hDoorsRep^{(\tau)}_i$ for each choice of horizontal slab $\tau$ and each choice of $i$ in time $O((n/B) (B^{O(d)} + n\log B))= o(n^2/\log^2n)$.

\subsection{Encoding Doors in a Mini-Block}
\label{sec:decisionencodingdoors}

For a mini-block $\mu=[i_0,i_0+b]\times [j_0,j_0+b]$ and each $j\in [j_0,j_0+b)$,
define $\vDoors^{(\mu)}_j$ to be the sorted ordering
of $\vdoor^-_{i_0,j},\vdoor^+_{i_0,j},\ldots,\vdoor^-_{i_0+b,j},\vdoor^+_{i_0+b,j}$.
Define $\vDoors^{(\mu)}$ to be the $b$-tuple of orderings $\langle\vDoors^{(\mu)}_{j_0},\ldots,$ $\vDoors^{(\mu)}_{j_0+b-1}\rangle$.

The number of possible $b$-tuples of orderings for $\vDoors^{(\mu)}$
is bounded by the number
of cells in the arrangement $\AAA$ of 
surfaces $\{(\xi_0,\ldots,\xi_b,\eta_0,\ldots,\eta_b)\in\R^{2(b+1)d}: f^-(\xi_i,\eta_j,\eta_{j+1},r)=f^-(\xi_{i'},\eta_j,\eta_{j+1},r)\}$,
$\{(\xi_0,\ldots,\xi_b,\eta_0,\ldots,\eta_b)\in\R^{2(b+1)d}: f^+(\xi_i,\eta_j,\eta_{j+1},r)=f^+(\xi_{i'},\eta_j,\eta_{j+1},r)\}$,
and $\{(\xi_0,\ldots,\xi_b,\eta_0,\ldots,\eta_b)\in\R^{2(b+1)d}: f^-(\xi_i,\eta_j,\eta_{j+1},r)=f^+(\xi_{i'},\eta_j,\eta_{j+1},r)\}$,
over all $i,i'\in [0,b]$ and $j\in [0,b)$.\footnote{As in the previous subsection, arrangements of endpoint equalities can also be augmented by the constant-degree polynomial sign tests defining the endpoint branches, nonemptiness, and grid-edge endpoint incidences, although these are not strictly required.}
This is an arrangement of $O(b^3)$ surfaces in $O(db)$ dimensions (in contrast to the arrangement $\AAA^{(\sigma)}$ we considered above, which resides in a lower dimension $O(d)$), and has $b^{O(db)}$ cells by the Milnor--Thom theorem~\cite{AgarwalS00a,Milnor64,Thom65}.  It follows that there are at most
$b^{O(db)}$ distinct $b$-tuples of orderings. The same bound holds for the candidates generated below, which need not be jointly realizable: allow independent copies of the two endpoints and the radius in each row. There are still $O(db)$ variables and $O(b^3)$ constant-degree predicates, so the same Milnor--Thom bound applies.

We cannot afford to explicitly construct this high-dimensional arrangement $\AAA$ and all its cells.  Instead, to generate all candidates for $\vDoors^{(\mu)}$,
we examine each vertical slab $\sigma=[I_0,I_0+B]\times[0,n]$, and each index $i_0\in [I_0,I_0+B)$ divisible by $b$, and enumerate all choices of the vector $\langle\vDoorsRep_{j_0}^{(\sigma)},\ldots,$ $\vDoorsRep_{j_0+b-1}^{(\sigma)}\rangle$, which has $O(b\log B)$ bits; each such choice
determines $\vDoors^{(\mu)}$ for a mini-block $\mu=[i_0,i_0+b]\times [j_0,j_0+b]$.  The number of candidates produced is $N=O((n/B)\cdot (B/b)\cdot 2^{O(b\log B)})=n^{1+O(\eps)}$, and
the time to generate them is $O(N\cdot B^{O(1)})= n^{1+O(\eps)}$.
This is strongly subquadratic assuming $\varepsilon > 0$ is a sufficiently small constant.
From the resulting list of candidates, we eliminate duplicates, e.g., by applying a sorting algorithm
in time $O(B^{O(1)}N\log N)=n^{1+O(\eps)}$.
Now define $\vDoorsRep^{(\mu)}$ to be the index of
$\vDoors^{(\mu)}$ in the list, which requires $O(\log (b^{O(db)}))=
O(b\log b)=O(\eps\log n)$ bits and thus fits in a word.  Note that $\vDoors^{(\mu)}$ can be recovered from $\vDoorsRep^{(\mu)}$, even without being given $\mu$ (the IDs for mini-blocks are ``universal'')---this fact will be important.

We build a table to implement the following operation:
\begin{operation}\label{op:doors}\rm
(Converting $\vDoorsRep$ for slab $\sigma$ to $\vDoorsRep$ for mini-block $\mu$)
\begin{tabbing}
\ \ \ \ \ \ \ \= Output:\ \ \=\+\kill
Input: 
\> vertical slab $\sigma=[I_0,I_0+B]\times[0,n]$, \\
\> index $i_0\in [I_0,I_0+B)$ divisible by $b$, and \\  
\> $\langle\vDoorsRep_{j_0}^{(\sigma)},\ldots,\vDoorsRep_{j_0+b-1}^{(\sigma)}\rangle$.\\
Output: \> $\vDoorsRep^{(\mu)}$ for mini-block $\mu=[i_0,i_0+b]\times [j_0,j_0+b]$.
\end{tabbing}
\end{operation}

Note that the input contains enough information to deduce the output (it is not necessary to be given the actual index $j_0$). 
The input vector above has $O(b\log B)=O(\eps\log n)$ bits and thus fits in a word.
The preprocessing time is $O((n/B)\cdot (B/b)\cdot 2^{O(b\log B)}\cdot B^{O(1)})=n^{1+O(\eps)}$; afterwards, the operation takes $O(1)$ time by table lookup. We describe later in this section how to use Operation~\ref{op:doors} to efficiently find $\vDoorsRep^{(\mu)}$ for each mini-block $\mu$.

Similarly, we define $\hDoors^{(\mu)}_i$, $\hDoors^{(\mu)}$, and $\hDoorsRep^{(\mu)}$, and 
build another table for an analog to Operation~\ref{op:doors}.

\subsection{Encoding Reaches}
\label{sec:encodingreaches}

We want to compute the reaches at the boundary of all mini-blocks, but as their coordinates are reals, we need to encode them as well, in order to enable bit packing. 

Consider a mini-block $\mu=[i_0,i_0+b]\times [j_0,j_0+b]$ contained in 
a vertical slab $\sigma=[I_0,I_0+B]\times [0,n]$.  Let $j\in [j_0,j_0+b)$.
Below are three specific encoding schemes, where the first is relative to $\sigma$, and the last two are relative to $\mu$:
\begin{itemize}
\item
Define $\vreachrep^{-(\sigma)}_{i_0,j}$ to be $z-I_0$ if $\vreach^-_{i_0,j}$ is equal to $\text{v-door}^-_{z, j}$ for some $z\in [I_0,I_0+B]$, or $(\mbox{``!''}, k)$ if 
 $\vreach^-_{i_0,j}$ is equal to $\vreach^-_{I_0,j}$ and $k$ is the
 rank of $\vreach^-_{I_0,j}$ with respect to $\vDoors^{(\sigma)}_j$, or ``-'' if the reach at the grid edge $\{i_0\}\times [j,j+1]$ is empty.  (One of these cases must hold.)
This requires $O(\log B)$ bits.
\item
Define $\vreachinrep^{-(\mu)}_{i_0,j}$ to be the rank
of $\vreach^-_{i_0,j}$ with respect to $\vDoors^{(\mu)}_j$ (or ``-'' if the reach at the grid edge $\{i_0\}\times [j,j+1]$ is empty).  This requires $O(\log b)$ bits.
\item
Define $\vreachoutrep^{-(\mu)}_{i_0+b,j}$ to be $z-i_0$ if 
$\vreach^-_{i_0+b,j}$ is equal to $\text{v-door}^-_{z, j}$ for some $z\in [i_0,i_0+b]$, or
``!'' if $\vreach^-_{i_0+b,j}$ is equal to $\vreach^-_{i_0,j}$, or ``-'' if the reach at grid edge $\{i_0+b\}\times [j,j+1]$ is empty.
This requires $O(\log b)$ bits.
\end{itemize}

The following operations allow us to convert between these different encodings:

\begin{operation}\label{op:reach1}\rm
(Converting $\vreachrep$ for slab $\sigma$ to $\vreachinrep$ for mini-block $\mu$)
\begin{tabbing}
\ \ \ \ \ \ \ \= Output:\ \ \=\+\kill
Input: 
\> vertical slab $\sigma=[I_0,I_0+B]\times[0,n]$,\\
\> index $i_0\in [I_0,I_0+B)$ divisible by $b$,\\
\> $\langle\vDoorsRep_{j_0}^{(\sigma)},\ldots,\vDoorsRep_{j_0+b-1}^{(\sigma)}\rangle$, and\\
\> $\langle\vreachrep^{-(\sigma)}_{i_0,j_0},\ldots,\vreachrep^{-(\sigma)}_{i_0,j_0+b-1}\rangle$.\\
Output: \> $\langle\vreachinrep^{-(\mu)}_{i_0,j_0},\ldots,\vreachinrep^{-(\mu)}_{i_0,j_0+b-1}\rangle$ for $\mu=[i_0,i_0+b]\times [j_0,j_0+b]$.
\end{tabbing}
\end{operation}

\begin{operation}\label{op:reach2}\rm
(Converting $\vreachoutrep$ for mini-block $\mu$ to $\vreachrep$ for slab $\sigma$)
\begin{tabbing}
\ \ \ \ \ \ \ \= Output:\ \ \=\+\kill
Input: 
\> vertical slab $\sigma=[I_0,I_0+B]\times[0,n]$,\\
\> index $i_0\in [I_0,I_0+B)$ divisible by $b$,\\
\> $\langle\vreachrep^{-(\sigma)}_{i_0,j_0},\ldots,\vreachrep^{-(\sigma)}_{i_0,j_0+b-1}\rangle$, and\\
\> $\langle\vreachoutrep^{-(\mu)}_{i_0+b,j_0},\ldots,\vreachoutrep^{-(\mu)}_{i_0+b,j_0+b-1}\rangle$ for $\mu=[i_0,i_0+b]\times [j_0,j_0+b]$.\\
Output: \> $\langle\vreachrep^{-(\sigma)}_{i_0+b,j_0},\ldots,\vreachrep^{-(\sigma)}_{i_0+b,j_0+b-1}\rangle$.
\end{tabbing}
\end{operation}

Note that the input contains enough information to deduce the output.
(In Operation~\ref{op:reach2}, the reason for including $\vreachrep^{-(\sigma)}_{i_0,j}$ as part of the input is that it lets us find a replacement value in the case when $\vreachoutrep^{-(\mu)}_{i_0+b,j}$ is ``!''.)
Each vector above has $O(b\log B)=O(\eps\log n)$ bits and thus fits in a word.
The preprocessing time is $O((n/B)\cdot (B/b)\cdot 2^{O(b\log B)}\cdot B^{O(1)})= n^{1+O(\eps)}$; afterwards, the operation takes $O(1)$ time by table lookup.

Similarly, we define $\hreachrep^{-(\tau)}$, $\hreachinrep^{-(\mu)}$, and $\hreachoutrep^{-(\mu)}$, and build tables for analogs to Operations~\ref{op:reach1}--\ref{op:reach2}.

\subsection{DP over Mini-Blocks}
\label{sec:DPfordecision}

The following is the key operation to implement DP over mini-blocks:

\begin{operation}\label{op:dp}\rm
(Computing $\vreachoutrep$ and $\hreachoutrep$ from $\vreachinrep$ and $\hreachinrep$ for mini-block $\mu$)
\begin{tabbing}
\ \ \ \ \ \ \ \= Output:\ \ \=\+\kill
Input: 
\> $\vDoorsRep^{(\mu)}$, $\hDoorsRep^{(\mu)}$, \\
\> $\langle\vreachinrep^{-(\mu)}_{i_0,j_0},\ldots,\vreachinrep^{-(\mu)}_{i_0,j_0+b-1}\rangle$, and\\
\> $\langle\hreachinrep^{-(\mu)}_{i_0,j_0},\ldots,\hreachinrep^{-(\mu)}_{i_0+b-1,j_0}\rangle$ for $\mu=[i_0,i_0+b]\times [j_0,j_0+b]$.\\
Output:
\> $\langle\vreachoutrep^{-(\mu)}_{i_0+b,j_0},\ldots,\vreachoutrep^{-(\mu)}_{i_0+b,j_0+b-1}\rangle$, and\\
\> $\langle\hreachoutrep^{-(\mu)}_{i_0,j_0+b},\ldots,\hreachoutrep^{-(\mu)}_{i_0+b-1,j_0+b}\rangle$.
\end{tabbing}
\end{operation}

Note that the input contains enough information to deduce the output, by simulating the Alt--Godau DP for continuous Fr\'echet~\cite{AltG95}.
(Recall that
$\vDoorsRep^{(\mu)}$ and $\hDoorsRep^{(\mu)}$ are decodable without being given
$\mu$ or the vertical slab $\sigma$ and horizontal slab $\tau$ containing $\mu$.)  
Each vector above has $O(b\log b)=O(\eps\log n)$ bits and thus fits in a word.
The preprocessing time is $O(2^{O(b\log B)}\cdot B^{O(1)})\le n^{O(\eps)}$; afterwards, the operation takes $O(1)$ time by table lookup.
(Luckily, $\sigma$ and $\tau$ are not part of the input, because otherwise, the preprocessing time would exceed $(n/B)^2 2^{\Omega(b\log B)}$, which is superquadratic!)

Consider a block $\beta=[I_0,I_0+B]\times [J_0,J_0+B]$ that is the intersection 
of a vertical slab $\sigma$ and a horizontal slab~$\tau$.
Assume that we have already computed $\vreach^-_{I_0,J_0},\ldots,\vreach^-_{I_0,J_0+B-1}$,
and $\hreach^-_{I_0,J_0},\ldots,$ $\hreach^-_{I_0+B-1,J_0}$. (We assume that $\vreach^-_{0, 0}, \ldots, \vreach^-_{0, n-1}$ and $\hreach^-_{0,0}, \ldots,$ $\hreach^-_{n-1,0}$ were already computed naively by DP in time $O(n)$.)

As initialization for the block $\beta$,
for each $j_0\in [J_0,J_0+B)$ divisible by $b$, we first form the vector $\langle\vDoorsRep_{j_0}^{(\sigma)},\ldots,\vDoorsRep_{j_0+b-1}^{(\sigma)}\rangle$ in $O(b)$ time, and pack it in a word.
Furthermore, we form the vector 
$\langle\vreachrep^{-(\sigma)}_{I_0,j_0},$ $\ldots,\vreachrep^{-(\sigma)}_{I_0,j_0+b-1}\rangle$
in $O(b\log B)$ time, by using binary search to compute ranks, and pack it in a word.
Similarly, for each $i_0\in [I_0,I_0+B)$ divisible by $b$, we form
the vectors $\langle\hDoorsRep_{i_0}^{(\tau)},\ldots,\hDoorsRep_{i_0+b-1}^{(\tau)}\rangle$ 
and $\langle\hreachrep^{-(\tau)}_{i_0,J_0},\ldots,\hreachrep^{-(\tau)}_{i_0+b-1,J_0}\rangle$.
The time per block is $O((B/b)\cdot b\log B) = O(B\log B)$.

Now consider a mini-block $\mu=[i_0,\allowbreak{}i_0+b]\times [j_0,\allowbreak{}j_0+b]\subset \beta$.
Assume that we have already computed 
$\langle\vreachrep^{-(\sigma)}_{i_0,\allowbreak{}j_0},\allowbreak{}\ldots,\allowbreak{}\vreachrep^{-(\sigma)}_{i_0,\allowbreak{}j_0+b-1}\rangle$
and $\langle\hreachrep^{-(\tau)}_{i_0,\allowbreak{}j_0},\allowbreak{}\ldots,\allowbreak{}\hreachrep^{-(\tau)}_{i_0+b-1,\allowbreak{}j_0}\rangle$.
We first compute $\vDoorsRep^{(\mu)}$ in $O(1)$ time by Operation~\ref{op:doors}, and similarly
$\hDoorsRep^{(\mu)}$.
Next we compute $\langle\vreachinrep^{-(\mu)}_{i_0,\allowbreak{}j_0},\allowbreak{}\ldots,\allowbreak{}\vreachinrep^{-(\mu)}_{i_0,\allowbreak{}j_0+b-1}\rangle$ in $O(1)$ time by Operation~\ref{op:reach1}, and similarly
$\langle\hreachinrep^{-(\mu)}_{i_0,\allowbreak{}j_0},\allowbreak{}$ $\ldots,\allowbreak{}\hreachinrep^{-(\mu)}_{i_0+b-1,\allowbreak{}j_0}\rangle$.
Then we compute $\langle\vreachoutrep^{-(\mu)}_{i_0+b,\allowbreak{}j_0},\allowbreak{}\ldots,\allowbreak{}\vreachoutrep^{-(\mu)}_{i_0+b,\allowbreak{}j_0+b-1}\rangle$ and $\langle\hreachoutrep^{-(\mu)}_{i_0,\allowbreak{}j_0+b},\allowbreak{}\ldots,\allowbreak{}\hreachoutrep^{-(\mu)}_{i_0+b-1,\allowbreak{}j_0+b}\rangle$ in $O(1)$ time by Operation~\ref{op:dp}.
Finally, we compute $\langle\vreachrep^{-(\sigma)}_{i_0+b,\allowbreak{}j_0},\allowbreak{}\ldots,\allowbreak{}\vreachrep^{-(\sigma)}_{i_0+b,\allowbreak{}j_0+b-1}\rangle$ in $O(1)$ time by Operation~\ref{op:reach2}, and
similarly $\langle\hreachrep^{-(\tau)}_{i_0,\allowbreak{}j_0+b},\allowbreak{}$ $\ldots,\allowbreak{}\hreachrep^{-(\tau)}_{i_0+b-1,\allowbreak{}j_0+b}\rangle$.  

Repeating this process for all mini-blocks $\mu\subset \beta$ in increasing lexicographical order of $(i_0,j_0)$,
we obtain $\text{v-reach-ID}^{-(\sigma)}_{I_0+B, J_0},$ $\ldots, \text{v-reach-ID}^{-(\sigma)}_{I_0+B, J_0+B-1}$ and $\text{h-reach-ID}^{-(\tau)}_{I_0, J_0+B}, \ldots, \text{h-reach-ID}^{-(\tau)}_{I_0+B-1,J_0+B}$ in time $O(B^2/b^2)$. As a final step, we map back from these IDs to real values in time $O(B)$, to obtain $\vreach^-_{I_0+B,J_0},$ $\ldots,$ $\vreach^-_{I_0+B,J_0+B-1}$,
and $\hreach^-_{I_0,J_0+B},\ldots,\hreach^-_{I_0+B-1,J_0+B}$.

Repeating this process for all blocks $\beta$ in increasing lexicographical order of $(I_0,J_0)$,
we obtain the final boundary reaches, including $\vreach^-_{n,n-1}$ and $\hreach^-_{n-1,n}$, and can decide whether $(n,n)\in\RRR$.

The total time is $O((n/B)^2B\log B + (n/b)^2) = O((n^2/\log^2n)\log^2\log n)$.

\begin{theorem} \label{thm:continuousdecision}
The continuous Fr\'echet distance decision problem in any constant dimension can be solved in $O((n^2/\log^2n)\log^2\log n)$
deterministic time.
\end{theorem}

\section{A Reduction to the Decision Problem}\label{sec:reduce}

Alt and Godau~\cite{AltG95} gave a reduction from the continuous Fr\'echet distance problem to its corresponding decision problem, based on parametric search~\cite{Megiddo83}, showing that if the latter can be solved in $T_D$ time, then the former can be solved in $O(n^2\log n + T_D\log n)$ time.
The authors in~\cite{BuchinBMM17} improved the running time to $O(n^2 2^{\alpha(n)} + T_D\log n)$, with randomization, specifically for $d=2$.
Unfortunately, because of the first term, these reductions cannot be used if we are aiming for subquadratic algorithms.  

In this section, we describe a reduction that uses random sampling to further lower the overhead term to subquadratic, at the expense of increasing the second term. 
We state our reduction in a general setting below.  Our reduction shares some similarities
with a variant of parametric search called the \emph{shrink-and-bifurcate} technique~\cite{AvrahamFKKS15,ChanH25a}, but our setting is different (part of the challenge for the continuous Fr\'echet distance problem is that critical values may be defined by \emph{triples} instead of pairs):

\begin{lemma}\label{lem:reduce}
Let $H$ be a set of $n$ 2-dimensional surfaces in $\R^{d+1}$ and $S$
be a set of $n$ $d$-dimensional surfaces in $\R^{d+1}$, each with constant description complexity, for a constant $d$. Assume that these surfaces are in general position.\footnote{On each $h\in H$, the restrictions of distinct surfaces of $S$ should share no curve arc, and their intersections and local $z$-extrema should be isolated. We use a constant-size planar decomposition of each $h$ on which these restrictions are constant-complexity curves. Only triples with distinct surfaces from $S$ are considered.}
Let ``$z$-coordinate'' refer to the last coordinate.
Also let $P$ be a set of $n$ points in $\R^d$, and let $Q$ be a set of $n$ points in $\R^d$.
Let $U$ denote the multiset of all $O(n^3)$ \emph{critical values}, defined as the
$z$-coordinates of the intersection points of triples in $H\times S\times S$,
together with the local $z$-minima/maxima of intersections of pairs in $H\times S$, and together with all pairwise distances between the points in $P$ and the points in $Q$.
Suppose there is an unknown value $z^*\in U$, and suppose there is a decision algorithm that can decide whether $z^*\le\zeta$ for a given value $\zeta$ in $T_D$ time;
this algorithm only accesses $\zeta$ by comparisons with critical values in $U$.

Then there is a Las Vegas randomized algorithm that always computes $z^*$ correctly.  Its running time is
$O(n^{2 - 1/O(d)} + n^{3/2}\log(1/\delta)\log n + T_D\log n + n\sqrt{T_D\log(1/\delta)}\log n)$
with probability at least $1-O(\delta)$, for any given $0<\delta\le 1/2$.
\end{lemma}
\begin{proof}

Let $t \geq 1$ be a parameter to be set later.\footnote{In each binary search, we also evaluate the decision algorithm infinitesimally below a tested value. This requires only resolving equality in each critical-value comparison toward the lower side, and costs $O(T_D)$ time. The answers at a value and just below it differ exactly when that value is $z^*$; otherwise we obtain a strict bracket. We allow $-\infty$ and $+\infty$ as outer bracket endpoints.}

\begin{enumerate}
\item
First, we either find $z^*$ or find an interval $(z^-, z^+)$ containing $z^*$ that does not contain any pairwise distances between the points in $P$ and the points in $Q$. We can select the $k$th smallest pairwise distance in $n^{2 - 1/O(d)}$ time by known algorithms (using standard range searching techniques)~~\cite{AgarwalASS93,AgarwalE98,ChanZ24}, so via binary search, this requires $O(\log n)$ calls to the decision algorithm and $n^{2 - 1/O(d)}$ additional time. Now every remaining critical value in $(z^-, z^+)$ is either defined by a triple in $H \times S\times S$ or a pair in $H \times S$.

\item
Pick $\lceil n^{3/2}\rceil$ independent uniform samples from the union of triples in $H\times S\times S$ with distinct $S$-factors and pairs in $H\times S$, let $R$ be the resulting collection, and then either find $z^*$ or find
a sub-interval $(z^-,z^+)$ of the interval found in Step 1 that contains $z^*$ but does not contain any critical values defined by 
$R$. By binary search on $O(n^{3/2})$ values, this requires $O(\log n)$ calls to the decision algorithm.
Let $K$ be the number of critical events of $U$ in this interval, counted with their defining-tuple multiplicities.  With
probability $1-O(\delta)$, we have
$K\le O(n^{3/2}\log(1/\delta))$. (To see this, note that on either side of the fixed value $z^*$ (including equality), an interval containing this many events involves $\Omega(n^{3/2}\log(1/\delta))$ defining tuples, since each tuple defines only $O(1)$ events. The probability of missing all of them is $O(\delta)$.)

\item
Pick a random subset $H'$ of $H$ of size $\lceil n/t\rceil$. 
For each 2-dimensional surface $h'\in H'$, compute all critical values that involve $h'$.  These critical values can be obtained
from the arrangement $\AAA_{h',(z^-,z^+)}$ formed by the intersection
of the surfaces of $S$ with $h' \cap (\R^{d}\times (z^-,z^+))$.
Now, $\AAA_{h',(z^-,z^+)}$ is an arrangement of $O(n)$ curves in 2D, and can be computed
by a standard plane sweep in $O((n+K_{h'})\log n)$ time~\cite{BentleyO79}, where
$K_{h'}$ counts its intersection and local-extremum events, with multiplicity.
On the event from Step 2,
$\sum_{h'} K_{h'} = O(K) = O(n^{3/2}\log(1/\delta))$, and therefore
the total time for this step is $O((n^2/t)\log n + n^{3/2}\log(1/\delta)\log n)$.

\item
By binary search over these critical values using $O(\log n)$ calls
to the decision algorithm, we can either find $z^*$ or shrink the interval $(z^-,z^+)$ 
so that it still contains $z^*$ but avoids critical values that involve some surface of $H'$. 
With probability $1-O(\delta)$, the critical values in
the shrunk interval $(z^-,z^+)$ are all defined by at most $O(t\log (1/\delta))$ surfaces of $H$. (To see this, order the surfaces on each side of $z^*$ (including equality) by their nearest critical value in the current interval. A sample of $\lceil n/t\rceil$ surfaces misses the first $k$ surfaces with probability at most $e^{-k/t}$; apply this with $k=O(t\log(1/\delta))$ on both sides.)

\item
Finally, we simulate the decision algorithm on the unknown $z^*$.  
Whenever the decision algorithm compares $z^*$ with a critical value
that is outside $(z^-,z^+)$, the comparison can be resolved.  Suppose it compares $z^*$ with
a critical value inside $(z^-,z^+)$.  If the critical value is defined by a 2-dimensional surface $h\in H$ that has not already been processed, we enumerate all critical values that involve $h$,
by constructing $\AAA_{h,(z^-,z^+)}$ in $O((n+K_h)\log n)$ time.
We can then do binary search over these critical values to resolve all comparisons with them using $O(\log n)$ calls to the decision algorithm. 
Otherwise, its defining surface has already been processed and the comparison was resolved earlier.

On the event from the preceding step, the total time for this step is
\[O((n\log n)(t\log (1/\delta)) + n^{3/2}\log(1/\delta)\log n + (T_D\log n)(t\log (1/\delta))).\]
\end{enumerate}
In all cases, the simulated decision algorithm is completed with every
comparison resolved, and the value returned is the unique critical value
consistent with these comparisons. Combining the high-probability events above
with $t=\max\{1,n/\sqrt{(n+T_D)\log(1/\delta)}\}$ gives the stated running-time bound.
\end{proof}

Lemma~\ref{lem:reduce} is indeed applicable to the continuous Fr\'echet distance problem.
For each edge $p_ip_{i+1}$, define the 2-dimensional flat $\{(\xi,r)\in\R^{d+1}: \mbox{$\xi$ lies on the line through $p_i$ and $p_{i+1}$}\}$.
For each point $p_i$, define the cone $\{(\xi,r)\in\R^{d+1}: \mbox{$\xi$ has distance $r$ to $p_i$}\}$.
Similarly, for each edge $q_jq_{j+1}$, define the 2-dimensional flat $\{(\xi,r)\in\R^{d+1}: \mbox{$\xi$ lies on the line
through $q_j$ and $q_{j+1}$}\}$.
For each point $q_j$, define the cone $\{(\xi,r)\in\R^{d+1}: \mbox{$\xi$ has distance $r$ to $q_j$}\}$.
Finally, we take $P$ and $Q$ to simply be the points defining the input polygonal chains.
We assume (for example, by an infinitesimal perturbation) that the input chains have no zero-length edges or repeated vertices, and that no edge-supporting line lies in the perpendicular bisector of two distinct input vertices.
It is not difficult to see that the Fr\'echet distance is a critical value for the resulting
set $H$ of $O(n)$ 2-dimensional flats, the set $S$ of $O(n)$ cones in $\R^{d+1}$, and the set of all pairwise distances between $P$ and $Q$.
Indeed, on an edge-supporting line, two distinct distance cones have a common curve only if the line lies in the perpendicular bisector of their vertices, which is excluded above. Their remaining intersections and extrema are isolated. The ordinary free-space DP compares the parameter only with these critical values; the faster oracle in the next section will also have this property.

Even with a critical-comparison decision algorithm taking $T_D=O((n^2/\log^2n)\log^2\log n)$ time, the bound in Lemma~\ref{lem:reduce} would be
$O(n^2\log\log n\sqrt{\log(1/\delta)})$, which is
still superquadratic.  
(For $d=2$ specifically, it is possible to improve the
time bound in the general reduction slightly, by replacing sorting with 2D arrangement construction, following the approach in~\cite{BuchinBMM17}, but we would still lose at least one log factor due to the binary searches.)

In the next section, we present a way that not only lets us obtain a
time bound below $n^2$ for any constant dimension~$d$, but also avoids losing a log factor!  The plan is to do some preprocessing in $O((n^2/\log^2n)\log^{O(1)}\log n)$ time, so that subsequently we can decrease the decision time $T_D$ by an arbitrarily large polylogarithmic factor---this way, both the $T_D\log n$ term and the $n\sqrt{T_D\log(1/\delta)}\log n$ term can be kept below $n^2/\log^2n$.

\section{Continuous Fr\'echet Distance Algorithm}
\label{sec:cont}

As mentioned, the plan is to do some preprocessing (which does not
require knowing the precise value of $r$) so that
the decision problem can be solved more efficiently later (when
the actual value of $r$ is given).
Here is a rough overview of the preprocessing and decision algorithms:

We first borrow the ``good'' vs.\ ``bad'' block idea from our
discrete Fr\'echet distance algorithm in Section~\ref{sec:discrete}.
Unfortunately, there could be $O(n^3)$ critical values (instead of $O(n^2)$)
for the continuous problem, since they may now be defined by triples.
To get around this, we will only work with critical values that are \emph{local} to some block---the number
of local critical values is $O(B^3)$ per block, and so $O((n/B)^2 B^3)= O(n^2 B)$ in total.
We sample a sufficiently large polylogarithmic number of blocks, sort all of their
local critical values, and binary search this list using
$O(\log B)=O(\log\log n)$ calls to the decision algorithm.  This either finds
$r^*$ immediately, or gives an interval $(r^-,r^+)$ containing $r^*$ such
that, with high probability, only $O(n^2/B^3)$ blocks contain a local critical
value in this interval.  We call these blocks \emph{bad}, and call all other
blocks \emph{good}.

At the time of preprocessing, we don't know the value $r\in(r^-,r^+)$ at
which the later decision algorithm will be queried.  However, inside a good
block, the empty/nonempty status of every door and the ordering of all door
endpoints stay fixed throughout $(r^-,r^+)$.  We can therefore fix any
$\tilde r\in(r^-,r^+)$, compute the needed information at $\tilde r$, and
store it using door-endpoint IDs that remain valid for every
$r\in(r^-,r^+)$.

At preprocessing time, we also don't have information about which points on
the left and bottom sides of a block $\beta$ are reachable from $(0,0)$.
Instead, we will compute local ``all-pairs" information about paths contained
inside $\beta$.  For a parameter $r$, call a point $p\in\beta$
\emph{$(\beta,r)$-locally reachable} if it is reachable through
$\mathcal F_r\cap\beta$ from the bottom or left side of $\beta$, and can also
reach the top or right side of $\beta$ through $\mathcal F_r\cap\beta$. (Here $\mathcal F_r$ is the free region with parameter $r$.) We
denote the set of these points by $\mathcal L_r^{(\beta)}$.  It will also be
useful to distinguish points that can reach the top side from points that can
reach the right side; we call these points $T$-alive and $R$-alive,
respectively.  The corresponding points on each grid edge form intervals.
All of these intervals can be computed by the same forward and reverse
mini-block scans used in the decision algorithm from
Section~\ref{sec:cont:decis}.

Suppose now that $X$ is a line segment contained in one grid edge on the
bottom or left side of $\beta$.  The \emph{top span} of $X$ will be the
smallest interval on the top side containing every point reachable from $X$,
and the \emph{right span} is defined similarly.  We show in Lemma \ref{lem:intervalspan} that, quite serendipitously, the points on the top side
that are actually reachable from $X$ are exactly the intersection of the top
span with $\mathcal L_r^{(\beta)}$, and the same statement holds on the right
side.  Thus, after computing the local reaches, it is enough to report the
two endpoints of each span when processing $X$.

The first geometric observation is that the reachable points on the top have
no gaps after we restrict to locally reachable points.  More precisely,
suppose $q_1,q_2,q_3$ occur in this order on the top side, $p$ reaches $q_1$
and $q_3$, and $q_2$ is locally reachable.  A path witnessing the local
reachability of $q_2$ must cross one of the two paths from $p$, and switching
paths at the crossing shows that $p$ also reaches $q_2$.  The same observation
holds on the right side. (See Figure \ref{fig:span} (a).)

The second observation is that, as $p$ moves forward in the order obtained by
walking down the left side and then rightward along the bottom side, both
endpoints of its top span move monotonically rightward, while both endpoints
of its right span move monotonically downward.  For example, if $q$ is the
rightmost top point reachable from $p$, then the corresponding point for a
later source $p'$ cannot lie to the left of $q$: the two witness paths would
cross, and this would make $q$ reachable from $p'$ as well. (See Figure \ref{fig:span} (b).)

\begin{figure}
\centering
\includegraphics[scale=1]{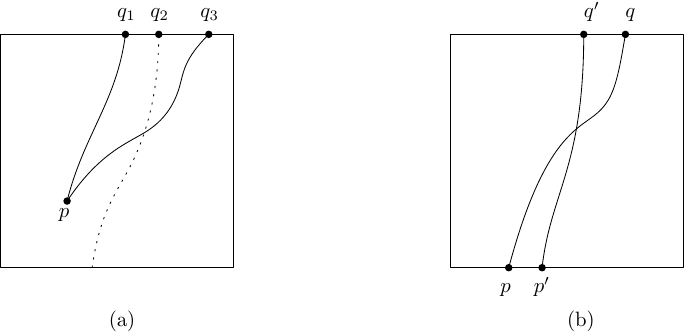}
\caption{(a) If there is an $xy$-monotone path from $p$ to $q_1$ and from $p$ to $q_3$ in the free region, and $q_2$ is alive, then there is an $xy$-monotone path from $p$ to $q_2$ in the free region.  (b) If $q$ is the rightmost point on the top side reachable from $p$ by an $xy$-monotone path in the free region, then $q'$ can't be the rightmost point on the top side reachable from $p'$ (since $q$ is reachable from $p'$).}\label{fig:span} 
\end{figure}

It is still too expensive to store the span of every possible point
separately.  We instead divide a block into $m=B/b$ \emph{vertical mini-slabs},
each being a $b \times B$ rectangle, and
use packed sweeps between their vertical boundaries to store a compressed representation.
Suppose for the time being that $X$ lies
on the bottom of a vertical mini-slab~$\upsilon$.

We split its top span into
the portion directly above $\upsilon$ and the portion weakly to the right of~$\upsilon$.  The first portion is completely determined by the two door
orderings inside $\upsilon$ and the two endpoints of $X$.  As in our earlier
algorithms, there are only $B^{O(b)}=n^{O(\eps)}$ possible combinatorial
descriptions, so its two endpoints can be found by one pre-computed table
lookup.

The remaining four span endpoints are computed using two kinds of packed
sweeps.  Consider first the rightmost reachable point on the top side and
the bottommost reachable point on the right side.  If $X$ intersects the
$T$-alive interval on its input grid edge, then the right endpoint of this
interval gives the rightmost top answer; similarly, the right endpoint of
the $R$-alive interval gives the bottommost right answer.  We propagate all
of these endpoints simultaneously from left to right through the vertical
mini-slabs.  At each separator, the crossing argument above lets us replace
a reached point by the lowest relevant $T$-alive or $R$-alive point through
which the sweep can continue.  There are $m$ possible starting mini-slabs
and $m$ stages, while the $b$ sources belonging to one starting mini-slab
fit in one word.  So all sweeps take only $\Ohat(m^2)$ time.

The leftmost top endpoint and the uppermost right endpoint need a slightly
different representation.  For either side, consider the leftmost point
where $X$ meets the corresponding alive interval on its input edge.  We
first propagate this point across its starting mini-slab, and ask for the
highest $T$-alive or $R$-alive interval on the right separator that can be
reached.  As the source moves from left to right along the bottom of the
mini-slab, this interval moves monotonically from top to bottom.  More
importantly, it changes only $O(B)$ times, once for each alive interval on
the separator (plus the empty gaps).  We can therefore describe this map by
$O(B)$ breakpoints.

To every breakpoint range, we attach a \emph{suffix answer}: the leftmost
point on the top side, or the uppermost point on the right side, reachable
from the corresponding alive interval on the separator.  These suffix
answers are computed for all separators by one right-to-left sweep.  When
the sweep crosses the next mini-slab, we retain the uppermost reached point
on the next separator and replace it by the highest relevant alive interval
through which a witness path can continue.  Again, this replacement preserves
the desired extremum by the path-crossing observation above.

There are $O(B)$ alive intervals and suffix answers on one separator, and
$O(B)$ breakpoint records for one starting mini-slab.  Each record has
$O(\log B)$ bits, so $b$ records fit in one word and either list occupies
only $O(B/b)=O(m)$ words.  There are $m$ separators, and the breakpoint
representation is constructed for each of the $m$ possible starting
mini-slabs.  Packed sorting and universal lookup
tables therefore let us construct all of these structures in
$\Ohat(m^2)=\Ohat(B^2/b^2)$ time per block.  The low-level implementation is
somewhat unpleasant, but this is the main reason that we gain the second log
factor.

The symmetric information for line segments $X$ on the left side of
a block is constructed analogously.

At the end, we have a compact representation of the top and right spans for
every line segment contained in one bottom or left grid edge of a block.
Over all blocks, the preprocessing time is
\[
    \Ohat\left((n/B)^2\frac{B^2}{b^2}\right)
    =\Ohat(n^2/\log^2 n).
\]
For the final decision algorithm, given $r$, we simulate DP over the blocks.
Each bad block is handled naively in $O(B^2)$ time.
For a good block, we decompose the reachable portions of its bottom and left
sides into $O(B)$ grid-edge intervals, query their span representations, and
combine the resulting intervals in $O(B\log B)$ time.
The total time for the bad blocks is $O((n^2/B^3)B^2)=O(n^2/B)$, while the
total time for the good blocks is $\Ohat((n/B)^2B)=\Ohat(n^2/B)$.
Setting $B=\log^C n$ for an arbitrarily large constant $C$ gives a decision
algorithm whose running time is smaller than $n^2$ by an arbitrarily large
polylogarithmic factor, after preprocessing in
$\Ohat(n^2/\log^2 n)$ time.

In the following subsections, we will make all these ideas precise.

\subsection{Preliminaries}

We use the notation of Section~\ref{sec:continuousdecision}, with the parameter
$r$ made explicit.  This means that $\FFF_r$ and $\RRR_r$ denote the free space and
reachable region at distance $r$. The door, reach, and ID definitions are the
same as before with an additional subscript $r$, and the univariate door endpoint
functions are denoted by $f_{i,j}^-(r),f_{i,j}^+(r),g_{i,j}^-(r),g_{i,j}^+(r)$.
The optimization problem is to find
$r^*=\min\{r:(n,n)\in\RRR_r\}$.

As in the previous sections, we will use two block sizes: $B \approx \log^c n$ and $b \approx \varepsilon \log n/\log\log n$, where $c$ is a sufficiently large constant, and $\varepsilon > 0$ is a sufficiently small constant. (Both constants depend on the dimension $d$.) We use the same definitions for blocks, mini-blocks, and vertical and horizontal slabs as in Section~\ref{sec:continuousdecision}. We will also use \emph{vertical mini-slabs}, which are rectangles of the form $[i_0, i_0 + b] \times [J_0, J_0 + B]$ for some $i_0$ divisible by $b$ and some $J_0$ divisible by $B$. \emph{Horizontal mini-slabs} are rectangles of the form $[I_0, I_0 + B] \times [j_0, j_0 + b]$ for some $I_0$ divisible by $B$ and some $j_0$ divisible by $b$.

\paragraph{An Initial Pruning Step.} The first step in our algorithm is similar to the one given in Section~\ref{sec:discreteoptimize}: we will invoke the decision algorithm from Section~\ref{sec:continuousdecision} $O(\log\log n)$ times, either finding $r^*$ immediately or restricting ourselves to an open interval $(r^-, r^+)$ such that $r^* \in (r^-, r^+)$.

For a block $\beta = [I_0, I_0 + B] \times [J_0, J_0 + B]$, let $\mathcal{C}_\beta$ be the set consisting of the following values of $r$. First, for every $j \in [J_0, J_0 + B)$ and every $i, i' \in [I_0, I_0+B]$, include every boundary point of the solution set of each of the following equations:
\[f_{i,j}^-(r) = f_{i',j}^-(r) \qquad f_{i,j}^-(r) = f_{i',j}^+(r) \qquad f_{i,j}^+(r) = f_{i',j}^-(r) \qquad f_{i,j}^+(r) = f_{i',j}^+(r).\]
Also include the boundary points of the sets of values for which one of $f_{i,j}^+(r)$ or $f_{i,j}^-(r)$ is equal to $j$ or $j+1$, or for which the corresponding door is nonempty.
Second, for every $i \in [I_0, I_0+B)$ and every $j,j' \in [J_0, J_0 + B]$, include every boundary point of the solution set of each of the following equations:
\[g_{i,j}^-(r) = g_{i,j'}^-(r) \qquad g_{i,j}^-(r) = g_{i,j'}^+(r) \qquad g_{i,j}^+(r) = g_{i,j'}^-(r) \qquad g_{i,j}^+(r) = g_{i,j'}^+(r).\]
Again include the boundary points of the endpoint-incidence and door-nonemptiness sets.  These univariate semi-algebraic sets have constant complexity, and hence have only $O(1)$ boundary points.
We refer to $\mathcal{C}_\beta$ as the set of \emph{local critical values for $\beta$}.
Observe that we have $\vert \mathcal{C}_\beta\vert = O(B^3)$.

For a given open interval $(r^-, r^+)$, we say that a block $\beta$ is \emph{$(r^-, r^+)$-good} if $\mathcal{C}_\beta \cap (r^-, r^+) = \emptyset$, and otherwise it is \emph{$(r^-, r^+)$-bad}.

We would like to efficiently find an open interval $(r^-, r^+)$ containing $r^*$ such that all but at most a $1/B$ fraction of the blocks are $(r^-, r^+)$-good. Unlike in Section~\ref{sec:discreteoptimize}, we do not have an efficient algorithm to (for example) deterministically select the $k$th smallest element in the additive union of the sets $\mathcal{C}_\beta$. We will instead use a randomized procedure.

Choose a set $\mathcal{S}$ of $B^{c_1}$ blocks uniformly at random, for a sufficiently large constant $c_1$, and explicitly construct and sort the set $\mathcal{C}_{\mathcal{S}} = \bigcup_{\beta \in \mathcal{S}} \mathcal{C}_\beta$.  We binary search this list using the decision algorithm.  When a sampled critical value $c$ is encountered, we evaluate the decision problem both at $c$ and infinitesimally below $c$.  If the two answers differ, then $c=r^*$ and we return $c$ immediately.  Otherwise, the binary search finds consecutive values $r^-<r^+$ in the sorted list such that
\[
    r^*\in(r^-,r^+)
    \qquad\text{and}\qquad
    \mathcal C_{\mathcal S}\cap(r^-,r^+)=\emptyset.
\]
This requires $O(\log\log n)$ calls to the decision algorithm in Theorem~\ref{thm:continuousdecision}, so the total time is
\[
    O((n^2/\log^2 n)\log^3\log n).
\]
To bound the number of bad blocks, let $N$ be the total number of blocks, so
$N=\Theta(n^2/B^2)$.  For any interval that intersects the local critical
sets of more than $N/B$ blocks, the probability that a sample of $B^{c_1}$
blocks misses all of them is at most $(1-1/B)^{B^{c_1}}$.  There are only
polynomially many intervals with endpoints at local critical values, so a
union bound shows that, for sufficiently large $c_1$, at most a $1/B$
fraction of the blocks are $(r^-,r^+)$-bad with probability
$1-n^{-\omega(1)}$.  Fix this open interval for the rest of the section; the
case in which the procedure already returned $r^*$ needs no further work.

It remains to actually classify blocks as either good or bad. Consider a vertical slab
$\sigma=[I_0,I_0+B]\times[0,n]$.  For every $i,i'\in [I_0,I_0+B]$
and every choice of signs $a,a'\in\{-,+\}$, consider the semi-algebraic set
of all pairs $(\eta,\eta')\in\R^{2d}$ for which the solution set of the equation
\[
        f^a(p_i,\eta,\eta',r)=f^{a'}(p_{i'},\eta,\eta',r)
\]
has a boundary point $r\in(r^-,r^+)$. We also include the analogous semi-algebraic sets
expressing that the endpoint-incidence or nonemptiness set of a vertical door
has a boundary point in $(r^-,r^+)$.  These are constant-complexity semi-algebraic sets, and there are
$B^{O(1)}$ of them.
Their arrangement has $B^{O(d)}$ cells, and in $B^{O(d)}$ time we can construct
a point-location structure for the arrangement and store the exact membership
of every cell, including lower-dimensional cells, in the union of the sets above.  Thus, given a pair $(q_j,q_{j+1})$, in $O(\log B)$ time
we can decide whether the vertical contribution to $\mathcal C_\beta$ contains
a value in $(r^-,r^+)$ for the block $\beta=\sigma\cap\tau(j)$, where
$\tau(j)=[0,n]\times[\lfloor j/B\rfloor B,\lfloor j/B\rfloor B+B]$.

We build the symmetric structures for horizontal slabs.  Namely, for a
horizontal slab $\tau=[0,n]\times[J_0,J_0+B]$, the query point is
$(p_i,p_{i+1})\in\R^{2d}$, and the semi-algebraic sets are obtained from the
equations involving the functions $g^-$ and $g^+$, together with the endpoint
and door-appearance boundary events.  A query determines whether the horizontal
contribution to $\mathcal C_\beta$ contains a value in $(r^-,r^+)$ for the
block $\beta=\sigma(i)\cap\tau$, where
$\sigma(i)=[\lfloor i/B\rfloor B,\lfloor i/B\rfloor B+B]\times[0,n]$.

Initialize all blocks as good.  For every vertical slab $\sigma$ and every
$j\in[0,n)$, we query the vertical point-location structure for $\sigma$ with
$(q_j,q_{j+1})$; if the query is positive, we mark the block
$\sigma\cap\tau(j)$ as bad.  Similarly, for every horizontal slab $\tau$ and
every $i\in[0,n)$, we query the horizontal structure for $\tau$ with
$(p_i,p_{i+1})$; if the query is positive, we mark the block
$\sigma(i)\cap\tau$ as bad.  By construction, after these two passes, a block
is marked bad exactly when $\mathcal C_\beta\cap(r^-,r^+)\ne\emptyset$.
The total number of queries is $O(n^2/B)$, and each query takes $O(\log B)$
time.  The construction time for all the point-location structures is
$O((n/B)B^{O(d)}) = o(n^2/B)$.  Hence all blocks can be classified in
$O((n^2/B)\log B)$ time.

The following lemma summarizes the algorithm so far:

\begin{lemma} \label{lem:prune}
    There is a Las Vegas randomized algorithm running in time $O((n^2/\log^2 n)\log^{O(1)}\log n)$ that either returns $r^*$ or finds an open interval $(r^-, r^+)$ such that $r^* \in (r^-, r^+)$ and exactly classifies every block as $(r^-,r^+)$-good or $(r^-,r^+)$-bad.  With probability $1 - n^{-\omega(1)}$, at most a $1/B$ fraction of the blocks are $(r^-, r^+)$-bad.
\end{lemma}

\paragraph{Local Reachability within Blocks.}
Consider a block $\beta = [I_0, I_0 + B] \times [J_0, J_0 + B]$ and a parameter $r$. Let $\partial_L \beta = \{I_0\} \times [J_0, J_0+B]$ be the left boundary of $\beta$. Similarly, let $\partial_B \beta = [I_0, I_0+B] \times \{J_0\}$, $\partial_R \beta = \{I_0+B\} \times [J_0, J_0+B]$, and $\partial_T \beta = [I_0, I_0+B] \times \{J_0+B\}$ be the bottom, right, and top boundaries respectively. We say that a point $p \in \beta$ is \emph{$(\beta, r)$-locally reachable} if both of the following hold:
\begin{itemize}
    \item There is an $xy$-monotone path through $\mathcal{F}_r$ from some point in $(\partial_L \beta \cup \partial_B \beta) \cap \mathcal{F}_r$ to $p$.
    \item There is an $xy$-monotone path through $\mathcal{F}_r$ from $p$ to some point in $(\partial_T \beta \cup \partial_R \beta) \cap \mathcal{F}_r$.
\end{itemize}
In other words, for the chosen parameter $r$, we require that there is some $xy$-monotone path in the free space from the bottom/left boundary of $\beta$ to the top/right boundary of $\beta$ that goes through $p$. (Observe that the path must be contained within $\beta$.) Let $\mathcal{L}_r^{(\beta)}$ be the set of all $(\beta, r)$-locally reachable points in $\beta$.

The \emph{$r$-local reach at a grid edge $e$ in $\beta$} refers to the intersection of $\mathcal{L}_r^{(\beta)}$ with $e$. If nonempty, it is a vertical or horizontal line segment. Let $\text{v-local-reach}^-_{i,j,r,\beta}$ and $\text{v-local-reach}^+_{i,j,r,\beta}$ denote the min and max $y$-coordinates respectively of the $r$-local reach at the vertical grid edge $\{i\} \times [j, j+1]$.\footnote{It is necessary to specify the block $\beta$ because some vertical edges are shared between two blocks. Unlike $r$-reachability, we \emph{do not} necessarily have that $\text{v-local-reach}^+_{i,j,r,\beta} = \text{v-door}^+_{i,j,r}$.} As with the $r$-doors and $r$-reaches, if the $r$-local reach is empty then we instead set $\text{v-local-reach}^-_{i,j,r,\beta} = \text{v-local-reach}^+_{i,j,r,\beta} =$ ``-''. We define $\text{h-local-reach}^-_{i,j,r,\beta}$ and $\text{h-local-reach}^+_{i,j,r,\beta}$ analogously.

Let $\sigma = [I_0, I_0 + B] \times [0,n]$ be the vertical slab such that $\beta \subset \sigma$. Let $\text{v-local-reach-ID}^{-(\sigma)}_{i,j,r}$ be $z - I_0$ if $\text{v-local-reach}^-_{i,j,r,\beta}$ is equal to $\text{v-door}^-_{z,j,r}$ for some $z \in [I_0, I_0+B]$, or ``-'' if the local reach at the grid edge $\{i\} \times [j, j+1]$ is empty. (One of the two cases must hold; there is no ``!'' case as in the definition of $\text{v-reach-ID}$ from Section~\ref{sec:encodingreaches}.)
We do not specify $\beta$ in the encoding, because the choice of $\sigma$ and $j$ determines~$\beta$. Define $\text{h-local-reach-ID}^{-(\tau)}_{i,j,r}$ analogously, using the horizontal slab $\tau = [0,n] \times [J_0, J_0 + B]$ such that $\beta \subset \tau$. Observe that each encoding can be stored in $O(\log B)$ bits.

Fix an arbitrary value $\tilde{r} \in (r^-, r^+)$. Our goal now is to construct, in $O((n^2/\log^2 n)\log^{O(1)}\log n)$ preprocessing time, a data structure for the following operation:

\begin{operation}\label{op:localreachIDs}\rm
(Reporting the lower Local Reach IDs on the output boundary of a block with respect to $\tilde{r}$)
\begin{tabbing}
\ \ \ \ \ \ \ \= Output:\ \ \=\+\kill
Input: 
\> mini-block $\mu = [i_0, i_0 + b] \times [j_0, j_0 + b]$ contained within a vertical slab $\sigma$ and a horizontal slab $\tau$,\\
\> \quad and adjacent to the top or right boundary of the block $\beta=\sigma\cap\tau$.\\
Output:
\> $\langle\text{v-local-reach-ID}^{-(\sigma)}_{i_0+b, j_0,\tilde{r}}, \ldots, \text{v-local-reach-ID}^{-(\sigma)}_{i_0+b, j_0 + b - 1,\tilde{r}}\rangle$, if the right side of $\mu$ lies on $\partial_R\beta$, and\\
\> $\langle\text{h-local-reach-ID}^{-(\tau)}_{i_0, j_0+b,\tilde{r}}, \ldots, \text{h-local-reach-ID}^{-(\tau)}_{i_0 + b - 1, j_0+b,\tilde{r}}\rangle$, if the top side of $\mu$ lies on $\partial_T\beta$.
\end{tabbing}
\end{operation}

On $\partial_T\beta\cup\partial_R\beta$, every point reachable from the entry
boundary is locally reachable and can continue along its door to the upper endpoint on the right side, or the right endpoint on the top side.  Thus only the lower endpoint of each local reach needs to be reported.  Each output vector
has $O(b\log B)=O(\varepsilon\log n)$ bits and fits in a word.

To implement Operation~\ref{op:localreachIDs}, we first run the preprocessing
from Section~\ref{sec:continuousdecision} with the parameter
$\tilde r$.  Thus we have the slab door encodings, the mini-block door
encodings, and the tables for Operations~\ref{op:doors}--\ref{op:dp} and their
horizontal analogs.
Now process each block $\beta=[I_0,I_0+B]\times[J_0,J_0+B]$, where
$\beta=\sigma\cap\tau$.  First perform a forward scan, exactly as in
Section~\ref{sec:DPfordecision}, except that the scan is local to $\beta$:
all points in $(\partial_L\beta\cup\partial_B\beta)\cap\mathcal F_{\tilde r}$
are initialized as reachable.  We propagate through the mini-blocks in
increasing order, using Operations~\ref{op:doors}--\ref{op:dp}, and retain
the packed lower endpoints on $\partial_R\beta$ and $\partial_T\beta$.  As in
Section~\ref{sec:DPfordecision}, the conversion operations
from Section~\ref{sec:encodingreaches} translate between the slab-relative
encodings and the ranks used by the universal mini-block tables.
The applicable output vector or vectors in Operation~\ref{op:localreachIDs}
are stored with each boundary mini-block.  The scan takes $O((B/b)^2)$ time
per block.

Putting everything together, we have the following lemma:

\begin{lemma} \label{lem:makeop6}
    There is an algorithm running in time $O((n^2/\log^2 n)\log^{O(1)}\log n)$ that outputs a data structure implementing Operation~\ref{op:localreachIDs} with $O(1)$ query time.
\end{lemma}

\subsection{Preprocessing Good Blocks}
\label{sec:preprocessgood}

We would like to efficiently construct a set of data structures that will allow us to execute the decision version of the continuous Fr\'echet DP within any good block in time $O(B\log B)$, for any query parameter $r \in (r^-, r^+)$. In the final subsection, we will leverage these data structures to construct a decision oracle for the continuous Fr\'echet distance problem with query time significantly faster than $O(n^2/\log^2 n)$. This improved oracle will allow us to solve the optimization problem in time $O((n^2/\log^2 n)\log^{O(1)}\log n)$, via Lemma~\ref{lem:reduce}.

Consider a block $\beta = [I_0, I_0+B] \times [J_0, J_0 + B]$ and a parameter $r$. Let $\partial_T \mathcal{L}_r^{(\beta)} = \mathcal{L}_r^{(\beta)} \cap \partial_T \beta$, and let $\partial_R \mathcal{L}_r^{(\beta)} = \mathcal{L}_r^{(\beta)} \cap \partial_R \beta$; these are the intersections of the locally reachable points with the top and right boundaries of $\beta$, respectively. Observe that $\partial_T \mathcal{L}_r^{(\beta)}$ and $\partial_R \mathcal{L}_r^{(\beta)}$ each consist of at most $O(B)$ intervals. For any subset $\mathcal{S} \subset \mathcal{F}_r \cap \beta$, let $\text{span}_r^{(\beta)}(\mathcal{S})$ be the set of all points in $\beta$ that are reachable from at least one point in $\mathcal{S}$ via an $xy$-monotone path through $\mathcal{F}_r$. This definition implies that $\text{span}_r^{(\beta)}(\mathcal{S}) \subseteq \mathcal{F}_r \cap \beta$. We also define $\partial_T \text{span}_r^{(\beta)}(\mathcal{S}) = \partial_T \beta \cap \text{span}_r^{(\beta)}(\mathcal{S})$ and $\partial_R \text{span}_r^{(\beta)}(\mathcal{S}) = \partial_R \beta \cap \text{span}_r^{(\beta)}(\mathcal{S})$ to be the intersections of the span with the top and right boundaries of $\beta$ respectively. We have the following useful characterization for the span of any line segment contained within one grid edge on the bottom or left boundary of $\beta$. (This is closely related to Alt and Godau's Lemma~10~\cite{AltG95}, Corollary~4.2 of Ben Avraham et al.~\cite{BenAvrahamKS15}, and Proposition~4.4 of Bringmann et al.~\cite{BringmannKN21}.)

\begin{lemma} \label{lem:intervalspan}
    Let $\beta = [I_0, I_0+B] \times [J_0, J_0+B]$ be any block, let $r$ be any parameter, and let $X$ be any line segment contained in $\mathcal F_r$ and in a grid edge $e = [i, i+1] \times \{J_0\}$ for some $I_0 \leq i < I_0+B$, or contained in $\mathcal F_r$ and in a grid edge $e = \{I_0\} \times [j, j+1]$ for some $J_0 \leq j < J_0+B$. Then there are two (possibly empty) line segments $\mathcal{I}_r^{(\beta)}(X) \subseteq [I_0, I_0+B] \times \{J_0+B\}$ and $\mathcal{J}_r^{(\beta)}(X) \subseteq \{I_0 + B\} \times [J_0, J_0+B]$ such that $\partial_T \text{span}_r^{(\beta)}(X) = \partial_T\mathcal{L}_r^{(\beta)} \cap \mathcal{I}_r^{(\beta)}(X)$ and $\partial_R \text{span}_r^{(\beta)}(X) = \partial_R\mathcal{L}_r^{(\beta)} \cap \mathcal{J}_r^{(\beta)}(X)$.
\end{lemma}

\begin{proof}
We prove the claim when $X$ lies on the bottom boundary of $\beta$; the
left-boundary case is symmetric.  Order the entry boundary by walking down
$\partial_L\beta$ and then right along $\partial_B\beta$.  If two sources and
two top-boundary targets occur in opposite orders, any monotone paths between
the corresponding pairs must intersect.  The analogous statement holds on
the right boundary, ordered from top to bottom.

Let $Y_T=\partial_T\text{span}_r^{(\beta)}(X)$, and let
$\mathcal I_r^{(\beta)}(X)$ be the minimal top-boundary segment containing
$Y_T$, or the empty segment if $Y_T=\emptyset$.  Clearly
$Y_T\subseteq\partial_T\mathcal L_r^{(\beta)}\cap
\mathcal I_r^{(\beta)}(X)$.  Conversely, take a point $y$ in the latter set.
If $y\notin Y_T$, closedness and minimality give points $y^-,y^+\in Y_T$ to
its left and right.  Let $Q$ be a monotone path from an entry-boundary point $s$ to
$y$.  If $s\in X$, we are done.  If $s$ precedes $X$, then $Q$ intersects a
path from $X$ to $y^-$; if $s$ succeeds $X$, it intersects a path from $X$
to $y^+$.  Switching paths at an intersection shows in either case that $X$
reaches $y$.  Thus
$Y_T=\partial_T\mathcal L_r^{(\beta)}\cap\mathcal I_r^{(\beta)}(X)$.
The same argument on the right boundary gives the segment
$\mathcal J_r^{(\beta)}(X)$.
\end{proof}

We define \emph{$\text{top-span}_r^{(\beta)}(X)$} and
\emph{$\text{right-span}_r^{(\beta)}(X)$} to be the minimal line segments
satisfying the corresponding equalities in Lemma~\ref{lem:intervalspan}; this
makes the two choices unique.


\paragraph{Encoding Scheme and Data Structures.}
Fix a block $\beta = [I_0, I_0+B] \times [J_0, J_0 + B]$ and the parameter
$\tilde r\in(r^-,r^+)$ from earlier.  We construct data structures reporting the endpoints
of $\text{top-span}_{\tilde r}^{(\beta)}(X)$ and
$\text{right-span}_{\tilde r}^{(\beta)}(X)$ for a segment $X$ in one bottom
grid edge of $\beta$.  The left-boundary case is symmetric.

Suppose that $X \subseteq e\cap\mathcal F_{\tilde r}$ for
$e = [i, i+1] \times \{J_0\}$, and let
$\upsilon = [i_0, i_0+b] \times [J_0, J_0+B]$ be the vertical mini-slab
containing $e$.  We split $\text{top-span}_{\tilde r}^{(\beta)}(X)$ into the
minimal segments containing, respectively,
\[
\partial_T\text{span}_{\tilde r}^{(\beta)}(X)
 \cap ([i_0,i_0+b]\times\{J_0+B\})
\quad\text{and}\quad
\partial_T\text{span}_{\tilde r}^{(\beta)}(X)
 \cap ([i_0+b,I_0+B]\times\{J_0+B\}).
\]
Denote them by $\text{top-span}_{\tilde r}^{(\upsilon)}(X)$ and
$\text{top-span}_{\tilde r}^{(\beta\backslash\upsilon)}(X)$.  No reachable
top-boundary point lies to the left of $\upsilon$, so these two segments
determine $\text{top-span}_{\tilde r}^{(\beta)}(X)$ in constant time.

Let $\text{top-span}_{\tilde{r}}^{-(\upsilon)}(X)$ and $\text{top-span}_{\tilde{r}}^{+(\upsilon)}(X)$ denote the min and max $x$-coordinates of $\text{top-span}_{\tilde{r}}^{(\upsilon)}(X)$ respectively. If $\text{top-span}_{\tilde{r}}^{(\upsilon)}(X)$ is empty, then set $\text{top-span}_{\tilde{r}}^{-(\upsilon)}(X) = \text{top-span}_{\tilde{r}}^{+(\upsilon)}(X) =$ ``-''. Similarly, let $\text{top-span}_{\tilde{r}}^{-(\beta\backslash\upsilon)}(X)$ and $\text{top-span}_{\tilde{r}}^{+(\beta\backslash\upsilon)}(X)$ be the min and max $x$-coordinates of $\text{top-span}_{\tilde{r}}^{(\beta\backslash\upsilon)}(X)$ respectively, and also let $\text{right-span}_{\tilde{r}}^{-(\beta)}(X)$ and $\text{right-span}_{\tilde{r}}^{+(\beta)}(X)$ be the min and max $y$-coordinates of $\text{right-span}_{\tilde{r}}^{(\beta)}(X)$ respectively. (Again, if $\text{top-span}_{\tilde{r}}^{(\beta\backslash\upsilon)}(X)$ or $\text{right-span}_{\tilde{r}}^{(\beta)}(X)$ is empty, then we set the corresponding values to be ``-''.)

We define the following encoding scheme for $\text{top-span}_{\tilde{r}}^{(\upsilon)}(X)$:
\begin{itemize}
    \item Let $\text{top-span-ID}_{\tilde{r}}^{-(\upsilon)}(X)$ be ``*'' if $\text{top-span}_{\tilde{r}}^{-(\upsilon)}(X)$ is the $x$-coordinate of the upper-right corner of $\upsilon$; otherwise, let it be $(i-I_0, j - J_0)$ if this coordinate is equal to $\text{h-door}^-_{i,j,\tilde{r}}$ for some $(i,j) \in \beta$, or ``!'' if it is equal to $X^-$, or ``-'' if $\text{top-span}_{\tilde{r}}^{(\upsilon)}(X)$ is empty.
    \item Let $\text{top-span-ID}_{\tilde{r}}^{+(\upsilon)}(X)$ be $i-I_0$ if $\text{top-span}_{\tilde{r}}^{+(\upsilon)}(X)$ is equal to $\text{h-door}^+_{i,J_0+B,\tilde{r}}$ for some $I_0 \leq i < I_0+B$, or ``-'' if $\text{top-span}_{\tilde{r}}^{(\upsilon)}(X)$ is empty.  Indeed, every reachable point in a top-boundary door can move to its right endpoint.
\end{itemize}
Define a similar set of encodings for $\text{top-span}_{\tilde{r}}^{(\beta\backslash\upsilon)}(X)$ and $\text{right-span}_{\tilde{r}}^{(\beta)}(X)$:
\begin{itemize}
    \item Let $\text{top-span-ID}_{\tilde{r}}^{-(\beta\backslash\upsilon)}(X)$ be ``*'' if $\text{top-span}_{\tilde{r}}^{-(\beta\backslash\upsilon)}(X)$ is the $x$-coordinate of the upper-right corner of $\upsilon$, or ``**'' if it is the $x$-coordinate of the upper-right corner of $\beta$ but not of $\upsilon$; otherwise, let it be $(i-I_0, j - J_0)$ if this coordinate is equal to $\text{h-door}^-_{i,j,\tilde{r}}$ for some $(i,j) \in \beta$, or ``-'' if $\text{top-span}_{\tilde{r}}^{(\beta\backslash\upsilon)}(X)$ is empty. If ``*'' and a door-endpoint ID both apply, we use ``*''.
    \item Let $\text{top-span-ID}_{\tilde{r}}^{+(\beta\backslash\upsilon)}(X)$ be $i-I_0$ if $\text{top-span}_{\tilde{r}}^{+(\beta\backslash\upsilon)}(X)$ is equal to $\text{h-door}^+_{i,J_0+B,\tilde{r}}$ for some $I_0 \leq i < I_0+B$, or ``-'' if $\text{top-span}_{\tilde{r}}^{(\beta\backslash\upsilon)}(X)$ is empty.
    \item Let $\text{right-span-ID}_{\tilde{r}}^{-(\beta)}(X)$ be ``**'' if $\text{right-span}_{\tilde{r}}^{-(\beta)}(X)$ is the $y$-coordinate of the upper-right corner of $\beta$; otherwise, let it be $(i-I_0, j - J_0)$ if this coordinate is equal to $\text{v-door}^-_{i,j,\tilde{r}}$ for some $(i,j) \in \beta$, or ``-'' if $\text{right-span}_{\tilde{r}}^{(\beta)}(X)$ is empty.
    \item Let $\text{right-span-ID}_{\tilde{r}}^{+(\beta)}(X)$ be $j-J_0$ if $\text{right-span}_{\tilde{r}}^{+(\beta)}(X)$ is equal to $\text{v-door}^+_{I_0+B,j,\tilde{r}}$ for some $J_0 \leq j < J_0+B$, or ``-'' if $\text{right-span}_{\tilde{r}}^{(\beta)}(X)$ is empty.
\end{itemize}
We will need one more set of encodings, this time for the line segment $X$ itself. Let $X^-$ and $X^+$ be the min and max $x$-coordinates of $X$ respectively, and recall that we have $X \subseteq e$ for some grid edge $e = [i,i+1] \times \{J_0\}$ on the bottom boundary of $\beta$. Let $\tau = [0,n] \times [J_0, J_0+B]$ be the horizontal slab such that $\beta \subset \tau$.
\begin{itemize}
    \item Let $\text{X-ID}^{-(\tau)}_i$ be the rank of $X^-$ in the sorted list $\hDoors^{(\tau)}_i$, or its between-position if it is not in the list, using a special symbol below or above the list.
    \item Define $\text{X-ID}^{+(\tau)}_i$ analogously for $X^+$.
\end{itemize}
Observe that all encodings can be stored in $O(\log B)$ bits.

We construct data structures supporting the following operations in constant query time. Their input spaces are too large for direct truth tables.

\begin{operation}\label{op:xtopnegminislab}\rm
(Reporting $\text{top-span-ID}_{\tilde{r}}^{-(\upsilon)}(X)$ for a query line segment $X$.)
\begin{tabbing}
\ \ \ \ \ \ \ \= Output:\ \ \=\+\kill
Input: 
\> grid edge $e = [i, i+1] \times \{J_0\}$, and \\
\> $\text{X-ID}^{-(\tau)}_i$ and $\text{X-ID}^{+(\tau)}_i$ for a line segment $X\subseteq e\cap\mathcal F_{\tilde r}$.\\
\> \quad ($e$ is on the bottom boundary of a vertical mini-slab $\upsilon = [i_0, i_0+b] \times [J_0, J_0+B]$, \\
\> \quad which is contained within a block $\beta$, and both are contained within a horizontal slab $\tau$).\\
Output:
\> $\text{top-span-ID}_{\tilde{r}}^{-(\upsilon)}(X)$.
\end{tabbing}
\end{operation}

\begin{operation}\label{op:xtopposminislab}\rm
(Reporting $\text{top-span-ID}_{\tilde{r}}^{+(\upsilon)}(X)$ for a query line segment $X$.)
\begin{tabbing}
\ \ \ \ \ \ \ \= Output:\ \ \=\+\kill
Input: 
\> Same as for Operation~\ref{op:xtopnegminislab}.\\
Output:
\> $\text{top-span-ID}_{\tilde{r}}^{+(\upsilon)}(X)$.
\end{tabbing}
\end{operation}

\begin{operation}\label{op:topnegbackslash}\rm
(Reporting $\text{top-span-ID}_{\tilde{r}}^{-(\beta\backslash\upsilon)}(X)$ for a query line segment $X$.)
\begin{tabbing}
\ \ \ \ \ \ \ \= Output:\ \ \=\+\kill
Input: 
\> Same as for Operation~\ref{op:xtopnegminislab}.\\
Output:
\> $\text{top-span-ID}_{\tilde{r}}^{-(\beta\backslash\upsilon)}(X)$.
\end{tabbing}
\end{operation}

\begin{operation}\label{op:topposbackslash}\rm
(Reporting $\text{top-span-ID}_{\tilde{r}}^{+(\beta\backslash\upsilon)}(X)$ for a query line segment $X$.)
\begin{tabbing}
\ \ \ \ \ \ \ \= Output:\ \ \=\+\kill
Input: 
\> Same as for Operation~\ref{op:xtopnegminislab}.\\
Output:
\> $\text{top-span-ID}_{\tilde{r}}^{+(\beta\backslash\upsilon)}(X)$.
\end{tabbing}
\end{operation}

\begin{operation}\label{op:rightneg}\rm
(Reporting $\text{right-span-ID}_{\tilde{r}}^{-(\beta)}(X)$ for a query line segment $X$.)
\begin{tabbing}
\ \ \ \ \ \ \ \= Output:\ \ \=\+\kill
Input: 
\> Same as for Operation~\ref{op:xtopnegminislab}.\\
Output:
\> $\text{right-span-ID}_{\tilde{r}}^{-(\beta)}(X)$.
\end{tabbing}
\end{operation}

\begin{operation}\label{op:rightpos}\rm
(Reporting $\text{right-span-ID}_{\tilde{r}}^{+(\beta)}(X)$ for a query line segment $X$.)
\begin{tabbing}
\ \ \ \ \ \ \ \= Output:\ \ \=\+\kill
Input: 
\> Same as for Operation~\ref{op:xtopnegminislab}.\\
Output:
\> $\text{right-span-ID}_{\tilde{r}}^{+(\beta)}(X)$.
\end{tabbing}
\end{operation}

\noindent
The choice of edge $e$ determines $\upsilon$, $\beta$, and $\tau$, so these do not need to be given explicitly as inputs.

The data structures for these operations will be constructed at the fixed value $\tilde r$.
The following lemma shows that, for a good block, the same encoded answers can
be decoded at every parameter in $(r^-,r^+)$. So for each good block, the tables will in fact apply for every relevant query value given to our decision algorithm.

\begin{lemma}\label{lem:goodstructuresallr}
Let $\beta$ be an $(r^-,r^+)$-good block.  Throughout $(r^-,r^+)$, every door in $\beta$ has the same empty/nonempty status, and all door-endpoint orderings used by the dynamic program are fixed.  The data structures for Operation~\ref{op:localreachIDs} and Operations~\ref{op:xtopnegminislab}--\ref{op:rightpos}, constructed with the parameter $\tilde r$, also implement the corresponding operations for every $r\in(r^-,r^+)$.  More precisely, for Operation~\ref{op:localreachIDs}, the returned IDs can be decoded by evaluating their indicated door endpoints at $r$.  For each of the other six operations and every valid query segment $X\subseteq e\cap\mathcal F_r$, if the two $X$-IDs are computed at $r$, the returned ID can likewise be decoded at $r$, with ``!'' interpreted as the corresponding endpoint of $X$, ``*'' as the shared corner used in the definition of the relevant span, and ``**'' as the upper-right corner of $\beta$.  The same statements hold for the symmetric operations for a query segment on the left boundary of $\beta$.
\end{lemma}

\begin{proof}
Because $\beta$ is good, $(r^-,r^+)$ contains no value at which a door changes
status or incidence, or at which two relevant door endpoints exchange order.
Thus all forward and reverse dynamic programs make the same comparisons and
return the same IDs throughout the interval, retaining equality classes and using fixed labels to break
persistent ties.  Once the two $X$-IDs are fixed, the other tables use no
coordinates of $X$.  Evaluating the indicated endpoints at $r$ therefore
decodes their answers.  The left-boundary case is symmetric.
\end{proof}

\subsection{Implementing Operations~\ref{op:xtopnegminislab} and~\ref{op:xtopposminislab}}

Observe that for a given query line segment $X$ on the bottom boundary of a vertical mini-slab $\upsilon$, $\text{top-span-ID}_{\tilde{r}}^{-(\upsilon)}(X)$ and $\text{top-span-ID}_{\tilde{r}}^{+(\upsilon)}(X)$ are entirely determined by $\text{X-ID}^{-(\tau)}_i$ and $\text{X-ID}^{+(\tau)}_i$, along with the combinatorial description of the door orderings in $\upsilon$. This is because any $xy$-monotone path from some point in $X$ to the top boundary of $\upsilon$ must be entirely contained within $\upsilon$. We start by describing how to efficiently find the combinatorial description for each vertical mini-slab $\upsilon$.

Fix a horizontal slab $\tau=[0,n]\times[J_0,J_0+B]$.  The \emph{horizontal
description ID} of some vertical mini-slab $\upsilon = [i_0, i_0 + b] \times [J_0, J_0+B]$ in $\tau$ is the packed vector
\[
    \langle \hDoorsRep_{i_0}^{(\tau)},\ldots,
    \hDoorsRep_{i_0+b-1}^{(\tau)}\rangle,
\]
where each ID is understood with respect to the fixed parameter $\tilde{r}$.
This vector gives, for each $[i, i+1] \in [i_0, i_0+b]$, the ordering
of all horizontal door endpoints in $\upsilon$ whose min and max $x$-coordinates lie in $[i, i+1]$.  To compute these vectors, we
build the same point-location structure for the horizontal slab $\tau$ as in
Section~\ref{sec:continuousdecision}, query it with each pair
$(p_i,p_{i+1})$, and pack consecutive groups of $b$ answers.  Over all
horizontal slabs, this takes
$O((n/B)(B^{O(d)}+n\log B)+n^2/B)$ time, which is $o(n^2/\log^2 n)$ assuming $B = \log^c n$ for a sufficiently large constant $c$.

Next we need a compact representation of the relative ordering of the vertical doors in $\upsilon$.
For fixed $\tau$, the points $q_{J_0},\ldots,q_{J_0+B}$ are fixed, and the
vertical door orderings of $\upsilon$ depend only on the $b+1$ points
$p_{i_0},\ldots,p_{i_0+b}$.  The surfaces at which this description changes
are the equalities
\[
    f^a(\xi_i,q_j,q_{j+1},\tilde r)
    = f^{a'}(\xi_{i'},q_j,q_{j+1},\tilde r),
\]
over all $i,i'\in[0,b]$, all $j\in[J_0,J_0+B)$, and all
$a,a'\in\{-,+\}$, together with the equalities representing the endpoint and door-appearance
events for these vertical doors.  This is an arrangement of $B^{O(1)}$ constant-degree
surfaces in $O(db)$ dimensions, and hence (because $d = O(1)$) it has $B^{O(b)}=n^{O(\eps)}$
cells by the Milnor--Thom theorem~\cite{AgarwalS00a,Milnor64,Thom65}.  So there are only $B^{O(b)}$ possible descriptions for the ordering of all vertical doors within a vertical mini-slab inside the fixed horizontal slab $\tau$.

We still need to actually assign these descriptions to the vertical
mini-slabs efficiently. First apply the preprocessing described in Section~\ref{sec:decisionencodingdoors} (with the parameter $\tilde{r}$) to obtain, for each vertical slab $\sigma$ and each choice of $j_0$ divisible by $b$, the packed vector $\langle\vDoorsRep_{j_0}^{(\sigma)},\ldots,\vDoorsRep_{j_0+b-1}^{(\sigma)}\rangle$.\footnote{Again, each ID is understood with respect to $\tilde{r}$.}
Now divide each vertical mini-slab $\upsilon$ into the $B/b$ mini-blocks
\[
    [i_0,i_0+b]\times[j_0,j_0+b],
    \qquad j_0=J_0,J_0+b,\ldots,J_0+B-b.
\]
For each such mini-block, we can use Operation~\ref{op:doors}, along with the packed vectors we just computed, to get the universal ID
for its vertical door ordering in $O(1)$ time.

We combine the $B/b$ universal IDs constructively as follows.  Fix a horizontal
slab $\tau$.  After processing the first $t$ mini-blocks, a ``state'' is the
vertical description restricted to the first $tb$ rows of $\upsilon$.  At the
next stage, for every actual vertical mini-slab in $\tau$, form the word-sized
pair consisting of its current state ID and the universal door ID of its next
mini-block.  We sort these pairs using standard packed radix sorting~\cite{AlbersH97,Chan12}, assign the same new ID to equal pairs,
and retain one explicit representative for every new ID.  Thus the IDs and
transitions are obtained directly from the descriptions that actually occur.  The
bound above implies that there are at most $B^{O(b)}$ distinct states, so each
state ID has $O(b\log B)=O(\eps\log n)$ bits.  Over all $B/b$ stages and all
horizontal slabs, the total number of pairs processed is $O((n/b)^2)$, and the
total construction time is $O((n/b)^2\log^{O(1)}B)$.  The final state ID is the
desired vertical description ID of $\upsilon$.

\paragraph{Lookup tables for top spans.}
For each vertical mini-slab $\upsilon$, the tuple consisting of (i) the horizontal slab $\tau$ such that $\upsilon \subset \tau$, (ii) its horizontal description ID, and (iii) its vertical description ID
determines all comparisons among door endpoints that can affect
reachability from the bottom of $\upsilon$ to the top of $\upsilon$.  It
remains only to account for the query segment $X$.  For each fixed
horizontal slab $\tau$, we build a table indexed by the horizontal description ID, the
vertical description ID, the offset $i-i_0\in[0,b)$ of the bottom edge
containing $X$, and the two values
$\text{X-ID}^{-(\tau)}_i,\text{X-ID}^{+(\tau)}_i$.  For every possible input
to this table, we choose one representative ordering consistent with the two
descriptions and with the two positions of $X^-$ and $X^+$, and simulate the
Alt--Godau dynamic program inside the $b\times B$ vertical mini-slab, with
$X$ as the only initially reachable segment on the bottom boundary.  The
simulation uses only order comparisons between reach endpoints, $X^-$,
$X^+$, and door endpoints, so all inputs with the same
combinatorial description have the same answer.  By induction in
the free-space DP, every endpoint produced by the simulation is either
$X^-$, $X^+$, a door endpoint, or the shared corner of the mini-slab, so the two endpoints of
$\text{top-span}_{\tilde r}^{(\upsilon)}(X)$ are encoded exactly by the IDs
defined above.  The table outputs these two IDs (using local coordinates
inside $\upsilon$, which are converted to the coordinates relative to
$\beta$ in $O(1)$ time).  Thus, as in the mini-block DP of
Section~\ref{sec:DPfordecision}, the ranks supplied with respect to the
horizontal slab are translated to the local ranks of the representative
mini-slab before the table is applied, and the returned local ranks are
translated back to the IDs relative to $\beta$.

For each horizontal slab $\tau$, the number of possible horizontal description IDs is $B^{O(b)}$, the number of
possible vertical description IDs is $B^{O(b)}$, and the number of choices
for the offset and the two $X$-IDs is $B^{O(1)}$.  Thus the table for one
horizontal slab has size $B^{O(b)}=n^{O(\eps)}$, and the time to construct the tables is also $n^{O(\varepsilon)}$.  Over all
horizontal slabs the table-construction time is $n^{1+O(\eps)}$.

\paragraph{Queries and running time.}
A query for Operations~\ref{op:xtopnegminislab} and~\ref{op:xtopposminislab} now reads the stored horizontal description ID and
vertical description ID of the relevant vertical mini-slab $\upsilon$, combines them with the offset of the
edge containing $X$ and the two $X$-IDs, and performs one table lookup.  The
total preprocessing time is
\[
    O((n/b)^2\log^{O(1)}B+(n^2/B)\log B+n^{1+O(\eps)})
    = O((n^2/\log^2 n)\log^{O(1)}\log n),
\]
assuming $B = \log^c n$ for a sufficiently large constant $c$, and assuming $\varepsilon$ is a sufficiently small constant.  The query time is $O(1)$. Putting everything together, we have:

\begin{lemma}\label{lem:minislabtopspan}
    There is an algorithm running in time $O((n^2/\log^2 n)\log^{O(1)}\log n)$ that outputs a data structure implementing Operations~\ref{op:xtopnegminislab} and~\ref{op:xtopposminislab} with $O(1)$ query time.
\end{lemma}

\subsection{Implementing Operations~\ref{op:topposbackslash} and~\ref{op:rightneg}}\label{sec:op1011}

We begin with three combinatorial observations. Throughout this subsection, all reachability is with respect to $\mathcal F_{\tilde r}$.  For a rectangle $R$, order its \emph{entry boundary} $\partial_LR\cup\partial_BR$ by walking from the upper-left corner down the left side and then from left to right along the bottom side. The endpoint monotonicity below is the continuous counterpart of Corollary~4.3 in~\cite{BenAvrahamKS15}; the remaining observations use the same crossing-path argument.

\begin{lemma}\label{lem:sourcedominance}
Let $R$ be a rectangle contained in a block $\beta$, and let $a_1,a_2\in(\partial_LR\cup\partial_BR)\cap\mathcal F_{\tilde r}$, where $a_1$ precedes $a_2$ in the entry-boundary order.  For $i=1,2$, let $y_i$ be the bottommost point of $\partial_RR$ reachable from $a_i$ inside $R$, and let $z_i$ be the rightmost point of $\partial_TR$ reachable from $a_i$ inside $R$.  Whenever both $y_1$ and $y_2$ are defined, $y_2$ is not above $y_1$.  Whenever both $z_1$ and $z_2$ are defined, $z_2$ is not to the left of $z_1$.
\end{lemma}

\begin{proof}
Suppose first that $y_2$ is above $y_1$, and let $P_i$ be an $xy$-monotone path from $a_i$ to $y_i$.  The endpoints of $P_1$ and $P_2$ occur in opposite orders on the entry boundary and on $\partial_RR$, so the two paths intersect.  Following $P_2$ up to an intersection point and then $P_1$ to $y_1$ gives an $xy$-monotone path from $a_2$ to $y_1$, contradicting the choice of $y_2$.  The proof for the top boundary is identical: if $z_2$ were to the left of $z_1$, the two paths would cross, implying that $a_2$ can also reach $z_1$.
\end{proof}

For $S\in\{T,R\}$, let $\mathcal L_{\tilde r,S}^{(\beta)}$ be the set of
points $p\in\beta\cap\mathcal F_{\tilde r}$ such that $p$ is reachable
through $\mathcal F_{\tilde r}\cap\beta$ from
$\partial_L\beta\cup\partial_B\beta$ and can reach $\partial_S\beta$
through $\mathcal F_{\tilde r}\cap\beta$.  We call these points
\emph{$S$-alive}.  Thus
\[
    \mathcal L_{\tilde r}^{(\beta)}
    =\mathcal L_{\tilde r,T}^{(\beta)}
     \cup\mathcal L_{\tilde r,R}^{(\beta)}.
\]
As with the local reach, the $S$-alive points on any grid edge form an
interval.

\begin{lemma}\label{lem:rightendpoint}
Let $e$ be a bottom grid edge of $\beta$, let $S\in\{T,R\}$, and let $u$
be the right endpoint of the nonempty interval
$\mathcal L_{\tilde r,S}^{(\beta)}\cap e$.  Let $X\subseteq e\cap
\mathcal F_{\tilde r}$ intersect this interval.  If $S=T$, the rightmost point
of $\partial_T\beta$ reachable from $X$ is the rightmost point reachable from
$u$.  If $S=R$, the bottommost point of $\partial_R\beta$ reachable from $X$
is the bottommost point reachable from $u$.
\end{lemma}

\begin{proof}
Some point of $X$ can first move right along the free interval to $u$, so every
point reachable from $u$ is reachable from $X$.  Conversely, every point of
$X$ that reaches $\partial_S\beta$ is $S$-alive and precedes $u$ on the entry
boundary.  Lemma~\ref{lem:sourcedominance}, applied to this point and $u$,
shows that the appropriate extremum from $u$ is at least as good.  The claim
follows.
\end{proof}

For $S\in\{T,R\}$, a vertical segment $\gamma$, and a point $Y\in\gamma$, let
$\operatorname{succ}_{\gamma,S}(Y)$ be the lowest point of
$\mathcal L_{\tilde r,S}^{(\beta)}\cap\gamma$ that is not below $Y$, or ``-''
if no such point exists.

\begin{lemma}\label{lem:alivesuccessor}
Fix $S\in\{T,R\}$.  Let $\upsilon$ be a vertical mini-slab in $\beta$, let
$\gamma=\partial_R\upsilon$, and let
$s\in\partial_L\upsilon\cup\partial_B\upsilon$.  Let $Y$ be the bottommost
point of $\gamma$ reachable from $s$ inside $\upsilon$.  If some $S$-alive
point of $\gamma$ is reachable from $s$, then
$Y'=\operatorname{succ}_{\gamma,S}(Y)$ is reachable from $s$.  Moreover, for
$S=R$, $Y'$ preserves the bottommost point reachable on $\partial_R\beta$;
for $S=T$, it preserves the rightmost point reachable on $\partial_T\beta$.
\end{lemma}

\begin{proof}
Let $W$ be an $S$-alive point of $\gamma$ reachable from $s$.  By definition,
$Y'$ lies between $Y$ and $W$.  Let $Q$ be the final portion inside
$\upsilon$ of a path witnessing that $Y'$ is reachable from
$\partial_L\beta\cup\partial_B\beta$, and let
$a\in\partial_L\upsilon\cup\partial_B\upsilon$ be its first point.  If
$Y'=Y$, $Y'=W$, or $a=s$, then $Y'$ is immediately reachable from $s$.  If
$a$ precedes $s$, the path $Q$ crosses a path from $s$ to $W$; if $a$
succeeds $s$, it crosses a path from $s$ to $Y$.  Switching paths at an
intersection again shows that $s$ reaches $Y'$.

Every reachable $S$-alive point on $\gamma$ lies on or above $Y'$.  Applying
Lemma~\ref{lem:sourcedominance} to the portion of $\beta$ to the right of
$\gamma$ shows that $Y'$ gives the desired extremum no worse than any such
higher point.
\end{proof}

\paragraph{Overview.}
For each bottom grid edge of a block $\beta$, we will retain only the
right endpoint of its $T$-alive interval and the right endpoint of its
$R$-alive interval.  If a query segment intersects one of these intervals,
then some point of the query can first move right to this endpoint.  By
Lemma~\ref{lem:sourcedominance}, formalized for this setting in
Lemma~\ref{lem:rightendpoint}, the two endpoints therefore suffice for
computing, respectively, the rightmost reachable point on $\partial_T\beta$
and the bottommost reachable point on $\partial_R\beta$.

We will simulate the continuous Fr\'echet DP from all these endpoints
simultaneously, sweeping through the vertical mini-slabs from left to right.
Universal tables will propagate each endpoint across one mini-slab and report
the corresponding top-boundary contribution.  At each separator, we will use
Lemma~\ref{lem:alivesuccessor} to replace the bottommost reached point by the
lowest $S$-alive point through which the sweep can continue.  The sweeps will
be performed in packed form and will take
$O((B^2/b^2)\log^{O(1)}B)$ time per block.

\paragraph{Boundary encodings and universal tables.}
We use the horizontal and vertical description IDs from the preceding
subsection.  A point on a vertical mini-slab boundary is encoded by its grid
edge and by the rank, or between-position, of its coordinate in the
corresponding slab door ordering.  We use the analogous encoding on horizontal
boundaries and use ``-'' for an absent point.  Each boundary ID has
$O(\log B)$ bits, so $b$ boundary IDs fit in one word.  Unlike in Section~\ref{sec:DPfordecision}, the $b$ points in a packed vector may lie on arbitrary rows. We batch all conversions at one separator: sort the $O(B)$ point records together with the $O(B)$ row door-description IDs, copy the appropriate description to each point by a packed scan, and restore the original order. A conversion table receives the containing slab, the packed row descriptions, the points' ranks, and the mini-slab offset, and translates the ranks in groups of $b$. The tables have $n^{1+O(\eps)}$ total preprocessing cost, just as in Section~\ref{sec:encodingreaches}; the batch costs $O((B/b)\log^{O(1)}B)$ time. We use these batched conversions in both sweeps below. We construct three types of tables.
\begin{enumerate}
    \item \textbf{Bottom-to-Right Table.} This table is indexed by the horizontal and vertical description IDs of a vertical mini-slab $\upsilon$, together with a packed vector containing one bottom boundary ID for each of its $b$ bottom grid edges.  For every input point, the table returns the vertical boundary ID of the bottommost point of $\partial_R\upsilon$ reachable from it, or ``-'' if no such point exists.  It also returns one bit indicating whether the input point reaches the upper-right corner of $\upsilon$; this bit is needed because the corner belongs to $\text{top-span}_{\tilde r}^{(\beta\backslash\upsilon)}(X)$.
    \item \textbf{Left-to-Right Table.} This table is indexed by the same two description IDs and a packed vector of $b$ points on $\partial_L\upsilon$.  For every input point, it returns the vertical boundary ID of the bottommost reachable point on $\partial_R\upsilon$ and the top boundary ID of the rightmost reachable point on $\partial_T\upsilon$.  Empty inputs and outputs are represented by ``-''. If an output keeps the input coordinate, the table returns the pass-through symbol ``!'' from Section~\ref{sec:encodingreaches}; otherwise it returns the defining local door endpoint. Keeping the original input ID then permits exact conversion back to a block-relative ID.
    \item \textbf{Verification Table.} This table is additionally given a packed vector of
    candidate points on $\partial_R\upsilon$ and reports which candidates are
    reachable from their corresponding bottom or left input points. For a left input and candidate on the same row, we also supply their comparison sign in the block-relative ordering before conversion. This distinguishes their order even when both have the same between-position in the local ordering.  Thus a failed
    candidate can be replaced by ``-'' in one table lookup.
\end{enumerate}

For each table entry, we choose a representative ordering consistent with the two description IDs and the supplied boundary IDs, and simulate the usual Alt--Godau dynamic program inside the $b\times B$ mini-slab.  The simulation uses only comparisons among the input points, candidate points, and door endpoints, so the output depends only on the table index.  The descriptions and the packed input and candidate vectors have $B^{O(b)}$ possible values.  So the size and construction time of these tables, over all horizontal slabs, is $n^{1+O(\eps)}$.

\paragraph{Packed sweeps.}
Fix a block $\beta=[I_0,I_0+B]\times[J_0,J_0+B]$, put $m=B/b$, and define
\[
    \upsilon_t=[I_0+tb,I_0+(t+1)b]\times[J_0,J_0+B],
    \qquad \gamma_t=\partial_R\upsilon_t,
    \qquad t=0,\ldots,m-1.
\]
Using the same packed scans, we combine the forward scan from
Operation~\ref{op:localreachIDs} with two reverse scans, initialized on
$\partial_T\beta$ and $\partial_R\beta$, respectively.  This gives the $T$-alive and $R$-alive
intervals on every mini-block boundary in $O(m^2)$ time.  We retain them on
the block bottom and on every separator $\gamma_t$.

For every $s$ and $S$, pack the right endpoints of the $S$-alive intervals on
the $b$ bottom grid edges of $\upsilon_s$, using ``-'' for an empty interval.
We run one conceptual sweep from each such packed vector, but process all
sweeps together from left to right.

At $\upsilon_t$, we start the sweeps with $s=t$ by applying the
bottom-to-right table, and advance every earlier sweep by applying the
left-to-right table.  For the $T$-alive sweeps, we retain the rightmost top
candidate encountered so far, including the upper-right corner reported by the
bottom-to-right table.  For the $R$-alive sweeps, we retain only the current
right-side point.

Before moving to $\upsilon_{t+1}$, collect the right-side points produced by
all sweeps with starting index at most $t$.  For either value of $S$, there are
at most $b(t+1)\le B$ such points.  Sort them together with the endpoints of
the $S$-alive intervals on $\gamma_t$ and scan the packed list to replace each
point $Y$ by $\operatorname{succ}_{\gamma_t,S}(Y)$.  We apply the verification
table to these candidates and replace every unreachable candidate by ``-''.
By Lemma~\ref{lem:alivesuccessor}, if the candidate is unreachable, then no
reachable $S$-alive point exists on the separator; otherwise the candidate
preserves the relevant suffix extremum.  The surviving points are used as the
inputs to the next mini-slab.
The packed sort and scan take $O(m\log^{O(1)}B)$ time per separator and per
choice of $S$~\cite{AlbersH97,Chan12}.

After the last mini-slab, the $R$-alive sweeps give the bottommost reachable
points on $\partial_R\beta$, and the accumulated $T$-alive answers give the
rightmost reachable points on $\partial_T\beta$ to the right of the starting
mini-slab.  Using the batched rank conversions above, we convert their local vertical-mini-slab
answers to the block-relative output IDs and store them by starting mini-slab
and bottom-edge offset.

\paragraph{Queries and running time.}
Let $X$ lie on a bottom grid edge of $\upsilon_s$.  To answer
Operation~\ref{op:topposbackslash}, test whether $X$ intersects the $T$-alive
interval on this edge.  If not, return ``-''.  Otherwise,
Lemma~\ref{lem:rightendpoint} shows that the answer stored for the right
endpoint of this interval is precisely the rightmost point of
$\partial_T\beta$ to the right of $\upsilon_s$ reachable from $X$.
Operation~\ref{op:rightneg} is identical, using the $R$-alive interval and its
stored bottommost right-boundary point.  The two intersection tests use the
stored interval IDs and the two $X$-IDs, and hence take constant time by table
lookup.

There are $O(m^2)$ table applications in all the sweeps for one block.  The
side-specific scans take $O(m^2)$ time, and the successor batches take
$O(m^2\log^{O(1)}B)$ time.  Thus the preprocessing time per block is
$O((B^2/b^2)\log^{O(1)}B)$, and over all blocks it is
\[
    O\left((n/B)^2\frac{B^2}{b^2}\log^{O(1)}B\right)
    =O((n^2/\log^2 n)\log^{O(1)}\log n).
\]
The universal-table construction time $n^{1+O(\eps)}$ is smaller for
sufficiently small $\eps$, and both operations have $O(1)$ query time.  This
gives the following lemma.

\begin{lemma}\label{lem:topposrightneg}
There is an algorithm running in time $O((n^2/\log^2 n)\log^{O(1)}\log n)$ that outputs a data structure implementing Operations~\ref{op:topposbackslash} and~\ref{op:rightneg} with $O(1)$ query time.
\end{lemma}

\subsection{Implementing Operations~\ref{op:topnegbackslash} and~\ref{op:rightpos}}

We now implement the two remaining operations.  The difficulty, compared with
Operations~\ref{op:topposbackslash} and~\ref{op:rightneg}, is that the answer is
not determined by a single endpoint of the alive interval on the input edge. Throughout this subsection, all reachability is with respect to
$\mathcal F_{\tilde r}$, and we use the entry-boundary order from the preceding
subsection. We start with some counterparts of the combinatorial lemmas given in Section \ref{sec:op1011}.

\begin{lemma}\label{lem:upperdominance}
Let $R$ be a rectangle, and let
$a_1,a_2\in(\partial_LR\cup\partial_BR)\cap\mathcal F_{\tilde r}$, where
$a_1$ precedes $a_2$.  For $i=1,2$, let $y_i$ be the uppermost point of
$\partial_RR$ reachable from $a_i$ inside $R$, and let $z_i$ be the
leftmost point of $\partial_TR$ reachable from $a_i$ inside $R$.  Whenever
the corresponding points are defined, $y_1$ is not below $y_2$, and $z_1$
is not to the right of $z_2$.
\end{lemma}

\begin{proof}
If $y_1$ were below $y_2$, paths from $a_1$ to $y_1$ and from $a_2$ to
$y_2$ would have their endpoints in opposite orders and hence would
intersect.  Switching paths at an intersection would show that $a_1$ reaches
$y_2$, contradicting the definition of $y_1$.  The proof for the top
boundary is identical.
\end{proof}

Recall the $T$-alive and $R$-alive sets
$\mathcal L_{\tilde r,T}^{(\beta)}$ and
$\mathcal L_{\tilde r,R}^{(\beta)}$.  For $S\in\{T,R\}$, a vertical
mini-slab boundary $\gamma$, and $Y\in\gamma$, let
$\operatorname{up}_{\gamma,S}(Y)$ be the top endpoint of the highest
maximal $S$-alive interval on $\gamma$ whose bottom endpoint is not above
$Y$; use ``-'' if no such interval exists.

\begin{lemma}\label{lem:upperalive}
Let $\upsilon$ be a vertical mini-slab in $\beta$, let
$\gamma=\partial_R\upsilon$, and let
$s\in\partial_L\upsilon\cup\partial_B\upsilon$.  Suppose that $Y$ is the
uppermost point of $\gamma$ reachable from $s$ inside $\upsilon$.  Fix
$S\in\{T,R\}$.  If $s$ reaches an $S$-alive point of $\gamma$, then
$Y'=\operatorname{up}_{\gamma,S}(Y)$ is reachable from $s$.  Moreover,
replacing all such reachable $S$-alive points by $Y'$ does not change the
uppermost reachable point on $\partial_R\beta$ when $S=R$, or the leftmost
reachable point on the portion of $\partial_T\beta$ weakly to the right of
$\gamma$ when $S=T$.
\end{lemma}

\begin{proof}
Let $W$ be an $S$-alive point of $\gamma$ reachable from $s$.  If $Y$ lies
in the interval ending at $Y'$, then $s$ reaches $Y'$ by moving upward from
$Y$.  Otherwise, $W$ is at or below $Y'$, which in turn is below $Y$.

Consider the final portion inside $\upsilon$ of a path witnessing that $Y'$
is $S$-alive, and let $a$ be its first point on
$\partial_L\upsilon\cup\partial_B\upsilon$.  The claim is immediate if
$a=s$.  If $a$ precedes $s$, this path crosses a path from $s$ to $Y$; if
$a$ succeeds $s$, it crosses a path from $s$ to $W$.  Switching paths at an
intersection shows that $s$ reaches $Y'$.

Every reachable $S$-alive point of $\gamma$ is at or below $Y'$.  Applying
Lemma~\ref{lem:upperdominance} to the part of $\beta$ to the right of
$\gamma$ proves the second claim.
\end{proof}

Let $\upsilon$ be a vertical mini-slab and
$\gamma=\partial_R\upsilon$.  For $x\in\partial_B\upsilon$, let $Y(x)$ be
the uppermost point of $\gamma$ reachable from $x$ inside $\upsilon$.  For
$S\in\{T,R\}$, define
\[
    \Psi_S(x)=\operatorname{up}_{\gamma,S}(Y(x))
\]
if $Y(x)$ is defined and this point is reachable from $x$; otherwise set
$\Psi_S(x)=$``-''.

\begin{lemma}\label{lem:criticalvalues}
As $x$ moves from left to right along $\partial_B\upsilon$, the nonempty
values of $\Psi_S(x)$ move monotonically from top to bottom along $\gamma$.
There are $O(B)$ maximal ranges on which $\Psi_S$ is nonempty and constant.

For fixed horizontal and vertical descriptions of $\upsilon$ and fixed
encoded $S$-alive intervals on $\gamma$, $\Psi_S$ is constant on every
bottom-boundary-ID class.
\end{lemma}

\begin{proof}
Lemma~\ref{lem:upperdominance} implies that $Y(x)$ moves from top to bottom,
and $\operatorname{up}_{\gamma,S}$ is monotone in the same order.
Lemma~\ref{lem:upperalive} implies that $\Psi_S(x)$ is nonempty precisely
when $x$ reaches an $S$-alive point of $\gamma$.

On each bottom grid edge, these source points form a prefix of the free
interval: an earlier point can move right to a later point and then follow
its path.  Thus empty ranges introduce only $O(b)$ breaks.  A nonempty value
can change only when it passes from one $S$-alive interval of $\gamma$ to
another, and there are $O(B)$ such intervals.

With the encoded $S$-alive intervals fixed as well, all comparisons made by
the Alt--Godau dynamic program have fixed outcomes, proving the boundary-ID
claim.
\end{proof}

\paragraph{Overview.}
For every block $\beta$, we will compute separate suffix structures for the
$T$-alive and $R$-alive intervals on each vertical mini-slab boundary.  The
first structure will store leftmost reachable points on $\partial_T\beta$,
and the second will store uppermost reachable points on $\partial_R\beta$.
Lemma~\ref{lem:upperdominance} will give the required monotonicity in the
entry-boundary order.  Lemma~\ref{lem:upperalive} will then let us replace
all relevant crossings of a separator by one reachable interval endpoint,
so the two structures can be built by right-to-left sweeps.

For each vertical mini-slab $\upsilon$, we will represent $\Psi_T$ and
$\Psi_R$ on $\partial_B\upsilon$ by their maximal constant ranges.
Lemma~\ref{lem:criticalvalues} will ensure that there are only $O(B)$ such
ranges and that their endpoints can be represented by the existing boundary
IDs.  We will attach the appropriate suffix answer to every range.  A query
will reduce to finding its leftmost $T$-alive or $R$-alive source point and
performing one constant-time predecessor query in the corresponding range
representation.

\paragraph{Preprocessing.}
Fix a block
$\beta=[I_0,I_0+B]\times[J_0,J_0+B]$, put $m=B/b$, and let
\[
    \upsilon_t=[I_0+tb,I_0+(t+1)b]\times[J_0,J_0+B],
    \qquad t=0,\ldots,m-1,
\]
and let $\gamma_t=\partial_R\upsilon_t$.  We use the same boundary IDs and the batched rank conversions from Section~\ref{sec:op1011}.
Thus a list of $O(B)$ points or records occupies $O(m)$ packed words.

We compute the $T$-alive and $R$-alive intervals by one forward scan of
$\beta$ and the two side-specific reverse scans from the preceding
subsection.  We retain these intervals on the bottom boundary and on every
$\gamma_t$.

We build a universal table that takes the horizontal and vertical
descriptions of a vertical mini-slab and a packed vector of $b$ points on its
left side.  For each point, it returns the uppermost reachable point on the
right side and the leftmost reachable point on the top side, including the
top-left corner.  As before, a table entry is obtained by choosing a
representative ordering and simulating the Alt--Godau dynamic program.  There
are $B^{O(b)}$ inputs, so the total table-construction time over all
horizontal slabs is $n^{1+O(\eps)}$.

We also use two upper-successor structures on each $\gamma_t$, one for the
$T$-alive intervals and one for the $R$-alive intervals.  Given a packed list
$Q$ of query points on
$\gamma_t$, the structure for $S\in\{T,R\}$ sorts $Q$ together with the bottom
endpoints of the $S$-alive intervals.  A scan finds, for every $Y\in Q$, the
highest such interval whose bottom endpoint is not above $Y$, and returns its
top endpoint, namely $\operatorname{up}_{\gamma_t,S}(Y)$.  The packed sorting
and scan take $O((m + \vert Q\vert/b)\log^{O(1)}B)$ time~\cite{AlbersH97,Chan12}.  We batch all queries at a separator, and invoke
this structure only when an $S$-alive crossing is guaranteed to exist.
Lemma~\ref{lem:upperalive} therefore guarantees that the returned endpoint is
reachable and preserves the suffix answer.

\paragraph{Suffix sweep.}
For every $S$-alive interval on $\gamma_t$, we store the suffix answer of its
top endpoint.  For $S=T$, this is the leftmost reachable point of
$\partial_T\beta$ weakly to the right of $\gamma_t$; for $S=R$, it is the
uppermost reachable point of $\partial_R\beta$.  We compute the two sets of
answers separately.  On $\gamma_{m-1}=\partial_R\beta$, the $R$-answer is
the top endpoint itself, and the $T$-answer is the upper-right corner of
$\beta$.  Indeed, every $T$-alive point of $\partial_R\beta$ reaches
$\partial_T\beta$ at this corner.

Now process
$t=m-2,m-3,\ldots,0$.  For a fixed $S\in\{T,R\}$, pack the top endpoints of
the $S$-alive intervals on $\gamma_t$ in groups of $b$, and apply the new
table for $\upsilon_{t+1}$.  For every interval this yields a
leftmost top candidate in $\upsilon_{t+1}$ and an uppermost reachable point
$Y$ on $\gamma_{t+1}$.  For $S=T$, a nonempty top candidate is already the
leftmost suffix answer, since every point reached after $\gamma_{t+1}$ lies
weakly farther right.  If this candidate is empty, $T$-aliveness of the input
forces every witness path to cross $\gamma_{t+1}$ at a $T$-alive point; we then
apply the $T$-upper-successor structure, and Lemma~\ref{lem:upperalive}
identifies the interval whose stored $T$-answer should be copied.

For $S=R$, every witness path must cross $\gamma_{t+1}$ at an $R$-alive
point, so we always apply the $R$-upper-successor structure and copy the
stored $R$-answer of the returned interval.  Thus every successor query is
made under the premise of Lemma~\ref{lem:upperalive}.  All such queries at
one separator are batched.  Whenever an answer is copied from
$\gamma_{t+1}$, the batched conversions retain its block-relative door-endpoint ID. A corner tag is first interpreted at its old separator and then re-encoded at the new separator, so ``*'' always denotes the corner of the separator currently in use. Since there are
$O(B)$ intervals on a separator, each stage
costs $O(m\log^{O(1)}B)$ time, and both sweeps cost
$O(m^2\log^{O(1)}B)$ time per block.  We store each answer using the existing
output encoding for its separator; in particular, a shared mini-slab corner
is encoded by ``*'', whereas the upper-right corner of $\beta$ is encoded by
``**'' unless it is also the shared corner.

\paragraph{Breakpoint structures.}
We build one further universal table.  Its input consists of the horizontal
and vertical descriptions of a vertical mini-slab, a choice
$S\in\{T,R\}$, and a packed group of $b$ consecutive $S$-alive intervals on
its right side.  For each listed interval the input also identifies the
interval immediately above it; only one such interval lies outside the
group, and a sentinel is used above the topmost interval.  The table
enumerates the bottom boundary IDs, simulates the
definition of $\Psi_S$, including its reachability test, and returns every
maximal exact range that maps to the top endpoint of a listed interval.

Lemma~\ref{lem:criticalvalues} bounds the output by $O(b)$ records: there are
only $b$ possible nonempty values, and the $b$ bottom grid edges introduce
only $O(b)$ additional breaks.  We retain all maximal ranges, since an
undefined range can separate two occurrences of the same value.  The input
has $O(b\log B)$ bits, and enumerating the boundary IDs costs only a
polynomial factor in $B$.  Thus this table also has $B^{O(b)}$ entries and
can be constructed over all horizontal slabs in $n^{1+O(\eps)}$ time.

For a fixed $\upsilon_s$, we apply the table separately to the $T$-alive and
$R$-alive intervals on $\gamma_s$, in consecutive groups of $b$.  For either
side, the $O(m)$ calls produce $O(B)$ ranges.  We attach to every range
the suffix answer stored for its alive interval, sort the ranges by their
first bottom boundary ID, and store them in a small-universe predecessor
structure.  After a predecessor query, one comparison with the last ID of
the returned range tests membership.  Hence an ID in a gap, where
$\Psi_S=$``-'', is correctly reported as empty.

Each list occupies $O(m)$ packed words, so its construction takes
$O(m\log^{O(1)}B)$ time and its query time is constant. For example, use an $O(b)$-ary tree whose keys at each node fit in one word and are searched by a precomputed table. Its height is $O(\log_b B)=O(1)$, and it can be built from the sorted packed list in $O(B/b)$ time.  The inclusive
$T$-suffix invariant also handles the shared corner: sources that reach the
top-right corner of $\upsilon_s$ occur in the ordinary $\Psi_T$-ranges, whose
attached answer is encoded by ``*''.

\paragraph{Queries and running time.}
Finally consider a query segment
$X\subseteq([i,i+1]\times\{J_0\})\cap\mathcal F_{\tilde r}$ on the bottom
boundary of $\beta$, and let $\upsilon_s$ be the vertical mini-slab containing
this edge.  A constant-time comparison table intersects $X$ separately with
the stored $T$-alive and $R$-alive intervals on this bottom edge.  Let $c_S$
be the leftmost point of the intersection for $S\in\{T,R\}$, or ``-'' if the
intersection is empty.  Because $X$ lies in one free interval, $c_S$ can first
move right along $X$ to every later $S$-alive point in the intersection.
Consequently the extremum obtained from all of $X$ is the one obtained from
$c_S$.

If $c_R\ne$``-'', we query the $R$-range structure for $\upsilon_s$ with
the bottom boundary ID of $c_R$.  Operation~\ref{op:rightpos} returns the
attached suffix answer if the returned exact range contains $c_R$, and
returns ``-'' otherwise.  Operation~\ref{op:topnegbackslash} is identical,
using $c_T$ and the $T$-range structure.  In either case an empty
intersection immediately gives ``-''.  Constant-time tables perform the
comparisons and convert the temporary boundary IDs to the required output
IDs.

The side-specific scans and the construction of the separator descriptions
take $O(m^2)$ time per block.  The suffix sweeps take
$O(m^2\log^{O(1)}B)$ time.  For each of the $m$ starting mini-slabs, the two
breakpoint structures take $O(m\log^{O(1)}B)$ time, for another
$O(m^2\log^{O(1)}B)$ total.
Thus the preprocessing time per block for Operations~\ref{op:topnegbackslash}
and~\ref{op:rightpos} is
\[
    O\left(\frac{B^2}{b^2}\log^{O(1)}B\right)
    =O\left(\frac{B^2}{b^2}\log^{O(1)}\log n\right),
\]
and over all blocks this is $O\left((n/B)^2\cdot (B^2/b^2)\log^{O(1)}B\right)
    =O((n^2/\log^2 n)\log^{O(1)}\log n).$
The universal-table construction time $n^{1+O(\eps)}$ is smaller for
sufficiently small $\eps$.  Both operations have constant query time, giving
the following lemma.

\begin{lemma}\label{lem:topnegrightpos}
There is an algorithm running in time
$O((n^2/\log^2 n)\log^{O(1)}\log n)$ that outputs a data structure
implementing Operations~\ref{op:topnegbackslash} and~\ref{op:rightpos} with
$O(1)$ query time.
\end{lemma}

\subsection{Putting Everything Together}

We now describe how to use the data structures from Section~\ref{sec:preprocessgood} to perform the continuous Fr\'echet DP more efficiently. First apply the algorithm in Lemma~\ref{lem:prune}.  If it returns $r^*$, we store this value and answer every decision query by testing whether $r\ge r^*$; the optimization problem is already solved in this case.  Otherwise it identifies an open interval $(r^-, r^+)$ containing $r^*$ and classifies every block exactly as good or bad with respect to this interval.  With probability $1 - n^{-\omega(1)}$, at most a $1/B$ fraction of the blocks are bad. Lemmas~\ref{lem:makeop6},~\ref{lem:minislabtopspan},~\ref{lem:topposrightneg}, and~\ref{lem:topnegrightpos} now give Operation~\ref{op:localreachIDs}, Operations~\ref{op:xtopnegminislab}--\ref{op:rightpos}, and symmetric versions of the latter operations for query intervals on the left boundary of a block. The total preprocessing time is $O((n^2/\log^2 n)\log^{O(1)}\log n)$.

Suppose that we are given a query value $r$.  If $r\le r^-$, then we can immediately
return ``no'', and if $r\ge r^+$, then we can immediately return ``yes'', since we know that $r^*\in(r^-,r^+)$.  It remains to handle the case that $r\in(r^-,r^+)$.

First run the standard continuous Fr\'echet DP over the blocks, in
increasing order of their lower-left corners. Initialize the reaches on the bottom and left boundaries of the full diagram in $O(n)$ time, as in Section~\ref{sec:DPfordecision}.  When the current block
$\beta=[I_0,I_0+B]\times[J_0,J_0+B]$ is bad, we do not try to use any of the
precomputed data structures: we run the ordinary Alt--Godau DP inside
$\beta$ in $O(B^2)$ time, using the door endpoints corresponding to the query value $r$, and record the resulting reaches on $\partial_T\beta$ and $\partial_R\beta$.

Now suppose that $\beta$ is good.  The input to the block consists of at most
$B$ reachable intervals on $\partial_B\beta$ and at most $B$ reachable intervals
on $\partial_L\beta$.  Consider first a nonempty reachable interval
$X\subseteq [i,i+1]\times\{J_0\}$ on the bottom boundary.  We compute
$\text{X-ID}^{-(\tau)}_i$ and $\text{X-ID}^{+(\tau)}_i$ by comparing the two
endpoints of $X$ with the sorted list of horizontal door endpoints on the same
edge, at the query value $r$; this takes $O(\log B)$ time by binary search.  The binary search uses the door ordering stored for the good block; by Lemma~\ref{lem:goodstructuresallr}, this ordering is the same at $r$ as at the representative value $\tilde r$ used in preprocessing.

Using these two $X$-IDs, Lemma~\ref{lem:goodstructuresallr} allows us to query Operations~\ref{op:xtopnegminislab} and~\ref{op:xtopposminislab} to obtain the two endpoint IDs of
$\text{top-span}_r^{(\upsilon)}(X)$, where $\upsilon$ is the vertical mini-slab
whose bottom boundary contains $X$.  We also query
Operations~\ref{op:topnegbackslash} and~\ref{op:topposbackslash} to obtain the
two endpoint IDs of $\text{top-span}_r^{(\beta\backslash\upsilon)}(X)$, and
query Operations~\ref{op:rightneg} and~\ref{op:rightpos} to obtain the two
endpoint IDs of $\text{right-span}_r^{(\beta)}(X)$.  By
Lemma~\ref{lem:goodstructuresallr}, we decode all returned IDs at $r$ and
combine the two top spans in constant time.  We perform the symmetric
procedure for each nonempty reachable interval on $\partial_L\beta$.  In this
way we obtain $O(B)$ candidate intervals on $\partial_T\beta$ and $O(B)$
candidate intervals on $\partial_R\beta$.

It remains to intersect these intervals with the locally reachable portions of $\partial_R\beta$ and $\partial_T\beta$. We obtain the lower local-reach IDs for all grid edges on
$\partial_T\beta$ and $\partial_R\beta$ by reading the entries stored for
Operation~\ref{op:localreachIDs} in the $O(B/b)$ mini-blocks adjacent to the top
and right sides of $\beta$.  By Lemma~\ref{lem:goodstructuresallr}, we decode these IDs by evaluating the indicated door
endpoints at $r$; the other endpoint of each nonempty local reach is the upper endpoint of its vertical door, or the right endpoint of its horizontal door, at $r$.  By Lemma~\ref{lem:intervalspan}, for every
input segment $X$, the portion of $\partial_T\beta$ reachable from $X$ is
exactly the intersection of $\partial_T\mathcal L_r^{(\beta)}$ with
$\text{top-span}_r^{(\beta)}(X)$, and analogously for the right boundary.  Thus
the output reaches of $\beta$ are obtained by taking the union of the
$O(B)$ candidate span intervals and intersecting this union with the $O(B)$
local-reach intervals on the top and right boundaries.  Sorting the endpoints
and sweeping along the two sides of $\beta$ takes $O(B\log B)$ time.  This
produces the exact reaches on $\partial_T\beta$ and $\partial_R\beta$.

After processing all blocks, we simply answer ``yes'' if and only if $(n,n)$ is reachable.  A good block costs $O(B\log B)$ time, and
there are $O(n^2/B^2)$ blocks.  With high probability, at most a $1/B$ fraction
of the blocks are bad, and each bad block costs $O(B^2)$ time.  With high probability, the total time per query is
\[
    O\left(n+\frac{n^2}{B^2}\cdot B\log B+
    \frac{1}{B}\cdot\frac{n^2}{B^2}\cdot B^2\right)
    =O\left(\frac{n^2\log B}{B}\right).
\]
Choosing the constant $c$ in $B=\log^c n$ sufficiently large, the query time is upper bounded by $O(n^2/\log^C n)$ for any fixed constant $C$ of our choice.  The DP computation above is exact for every outcome of the random sample; the high-probability event is used only to bound the number of bad blocks.

The following lemma summarizes the construction:

\begin{lemma}\label{lem:fastdecisionoracle}
    For every constant $C$, there is a Las Vegas randomized preprocessing algorithm that always outputs a correct decision oracle for the continuous Fr\'echet distance problem.  The preprocessing time is
    $O((n^2/\log^2 n)\log^{O(1)}\log n)$, and with high probability the oracle has query time $O(n^2/\log^C n)$.
\end{lemma}

Finally we apply the reduction in Lemma~\ref{lem:reduce}
with $T_D=O(n^2/\log^C n)$ and $\delta=1/n^C$ for an arbitrarily large constant $C$.  The decision oracle above accesses the query parameter only through comparisons needed to locate door endpoints in the symbolic orderings, to compare two such endpoints, to test whether a door appears or disappears, or to test whether a door endpoint hits a grid-edge endpoint; these are exactly comparisons with critical values of the type used in Lemma~\ref{lem:reduce}.  The table lookups themselves are purely symbolic, and Lemma~\ref{lem:goodstructuresallr} justifies decoding the tables built at $\tilde r$ at any query value $r\in(r^-,r^+)$.
Consequently, we obtain a Las Vegas randomized algorithm that always returns the exact continuous Fr\'echet distance, and whose running time is
$O(n^2/\log^{C/2-3/2} n)$ with high probability, after preprocessing in $O((n^2/\log^2n)\log^{O(1)}\log n)$ time. 
We conclude:

\begin{theorem} \label{thm:continuousoptimization}
The continuous Fr\'echet distance problem in any constant dimension can be solved by a Las Vegas randomized algorithm that always returns the correct distance and whose running time is $O((n^2/\log^2n)\log^{O(1)}\log n)$
with high probability.
\end{theorem}

\section*{Acknowledgments}
All ideas in this paper are solely due to the human authors, and the human authors are fully responsible for the contents of this paper. After writing a full draft of the paper, the authors used ChatGPT 5.5 Pro to improve the presentation of some of the proofs, which were then edited again by the authors.

{\small
\bibliographystyle{alphaurl}
\bibliography{refs}

@inproceedings{ChanH25a,
  author       = {Timothy M. Chan and
                  Zhengcheng Huang},
  _editor       = {Oswin Aichholzer and
                  Haitao Wang},
  title        = {Faster Algorithms for Reverse Shortest Path in Unit-Disk Graphs and
                  Related Geometric Optimization Problems: Improving the Shrink-And-Bifurcate
                  Technique},
  booktitle    = {Proc. 41st International Symposium on Computational Geometry (SoCG)},
  _series       = {LIPIcs},
  _volume       = {332},
  pages        = {32:1--32:14},
  _publisher    = {Schloss Dagstuhl - Leibniz-Zentrum f{\"{u}}r Informatik},
  year         = {2025},
  _url          = {https://doi.org/10.4230/LIPIcs.SoCG.2025.32},
  doi          = {10.4230/LIPICS.SOCG.2025.32},
  timestamp    = {Fri, 21 Nov 2025 14:06:54 +0100},
  biburl       = {https://dblp.org/rec/conf/compgeom/ChanH25a.bib},
  bibsource    = {dblp computer science bibliography, https://dblp.org}
}

@article{BuchinBMM17,
  author       = {Kevin Buchin and
                  Maike Buchin and
                  Wouter Meulemans and
                  Wolfgang Mulzer},
  title        = {Four {S}oviets Walk the Dog: Improved Bounds for Computing the {F}r{\'{e}}chet
                  Distance},
  journal      = {Discret. Comput. Geom.},
  volume       = {58},
  number       = {1},
  pages        = {180--216},
  year         = {2017},
  _url          = {https://doi.org/10.1007/s00454-017-9878-7},
  doi          = {10.1007/S00454-017-9878-7},
  timestamp    = {Thu, 12 Mar 2020 17:20:46 +0100},
  biburl       = {https://dblp.org/rec/journals/dcg/BuchinBMM17.bib},
  bibsource    = {dblp computer science bibliography, https://dblp.org},
note 	= {Preliminary version in SODA 2014}
}

@inproceedings{ChengH25,
  author       = {Siu{-}Wing Cheng and
                  Haoqiang Huang},
  _editor       = {Yossi Azar and
                  Debmalya Panigrahi},
  title        = {Fr{\'{e}}chet Distance in Subquadratic Time},
  booktitle    = {Proc. 36th Annual {ACM-SIAM} Symposium on Discrete Algorithms (SODA)},
  pages        = {5100--5113},
  _publisher    = {{SIAM}},
  year         = {2025},
  _url          = {https://doi.org/10.1137/1.9781611978322.173},
  doi          = {10.1137/1.9781611978322.173},
  timestamp    = {Tue, 28 Jan 2025 14:38:41 +0100},
  biburl       = {https://dblp.org/rec/conf/soda/ChengH25.bib},
  bibsource    = {dblp computer science bibliography, https://dblp.org}
}

@article{AgarwalAKS14,
  author       = {Pankaj K. Agarwal and
                  Rinat {Ben Avraham} and
                  Haim Kaplan and
                  Micha Sharir},
  title        = {Computing the Discrete {F}r{\'{e}}chet Distance in Subquadratic
                  Time},
  journal      = {{SIAM} J. Comput.},
  volume       = {43},
  number       = {2},
  pages        = {429--449},
  year         = {2014},
  _url          = {https://doi.org/10.1137/130920526},
  doi          = {10.1137/130920526},
  timestamp    = {Sat, 27 May 2017 14:22:58 +0200},
  biburl       = {https://dblp.org/rec/journals/siamcomp/AgarwalAKS14.bib},
  bibsource    = {dblp computer science bibliography, https://dblp.org},
note 	= {Preliminary version in SODA 2013}
}

@inproceedings{BuchinOS19,
  author       = {Kevin Buchin and
                  Tim Ophelders and
                  Bettina Speckmann},
  _editor       = {Timothy M. Chan},
  title        = {{SETH} Says: Weak {F}r{\'{e}}chet Distance is Faster, but only
                  if it is Continuous and in One Dimension},
  booktitle    = {Proc. 30th Annual {ACM-SIAM} Symposium on Discrete
                  Algorithms (SODA)},
  pages        = {2887--2901},
  _publisher    = {{SIAM}},
  year         = {2019},
  _url          = {https://doi.org/10.1137/1.9781611975482.179},
  doi          = {10.1137/1.9781611975482.179},
  timestamp    = {Tue, 07 May 2024 20:08:16 +0200},
  biburl       = {https://dblp.org/rec/conf/soda/BuchinOS19.bib},
  bibsource    = {dblp computer science bibliography, https://dblp.org}
}

@article{BringmannM16,
  author       = {Karl Bringmann and
                  Wolfgang Mulzer},
  title        = {Approximability of the discrete {F}r{\'{e}}chet distance},
  journal      = {J. Comput. Geom.},
  volume       = {7},
  number       = {2},
  pages        = {46--76},
  year         = {2016},
  url          = {https://doi.org/10.20382/jocg.v7i2a4},
  doi          = {10.20382/JOCG.V7I2A4},
  bibsource    = {dblp computer science bibliography, https://dblp.org},
note = {Preliminary version in SoCG 2015}
}

@inproceedings{Bringmann14,
  author       = {Karl Bringmann},
  title        = {Why Walking the Dog Takes Time: {F}rechet Distance Has No Strongly Subquadratic
                  Algorithms Unless {SETH} Fails},
  booktitle    = {Proc. 55th {IEEE} Annual Symposium on Foundations of Computer Science (FOCS)},
  pages        = {661--670},
  _publisher    = {{IEEE} Computer Society},
  year         = {2014},
  _url          = {https://doi.org/10.1109/FOCS.2014.76},
  doi          = {10.1109/FOCS.2014.76},
  timestamp    = {Tue, 08 Jul 2025 16:37:45 +0200},
  biburl       = {https://dblp.org/rec/conf/focs/Bringmann14.bib},
  bibsource    = {dblp computer science bibliography, https://dblp.org}
}

@inproceedings{BringmannDNP22,
  author       = {Karl Bringmann and
                  Anne Driemel and
                  Andr{\'{e}} Nusser and
                  Ioannis Psarros},
  _editor       = {Joseph (Seffi) Naor and
                  Niv Buchbinder},
  title        = {Tight Bounds for Approximate Near Neighbor Searching for Time Series
                  under the {F}r{\'{e}}chet Distance},
  booktitle    = {Proc. 33rd {ACM-SIAM} Symposium on Discrete Algorithms (SODA)},
  pages        = {517--550},
  _publisher    = {{SIAM}},
  year         = {2022},
  _url          = {https://doi.org/10.1137/1.9781611977073.25},
  doi          = {10.1137/1.9781611977073.25},
  timestamp    = {Tue, 07 May 2024 20:08:16 +0200},
  biburl       = {https://dblp.org/rec/conf/soda/BringmannDNP22.bib},
  bibsource    = {dblp computer science bibliography, https://dblp.org}
}

@inproceedings{ChengHZ25,
  author       = {Siu{-}Wing Cheng and
                  Haoqiang Huang and
                  Shuo Zhang},
  _editor       = {Michal Kouck{\'{y}} and
                  Nikhil Bansal},
  title        = {Constant Approximation of {F}r{\'{e}}chet Distance in Strongly
                  Subquadratic Time},
  booktitle    = {Proc. 57th Annual {ACM} Symposium on Theory of Computing (STOC)},
  pages        = {2329--2340},
  _publisher    = {{ACM}},
  year         = {2025},
  _url          = {https://doi.org/10.1145/3717823.3718157},
  doi          = {10.1145/3717823.3718157},
  timestamp    = {Sun, 02 Nov 2025 21:27:35 +0100},
  biburl       = {https://dblp.org/rec/conf/stoc/ChengHZ25.bib},
  bibsource    = {dblp computer science bibliography, https://dblp.org}
}

@article{ChanR18,
  author       = {Timothy M. Chan and
                  Zahed Rahmati},
  title        = {An improved approximation algorithm for the discrete {F}r{\'{e}}chet
                  distance},
  journal      = {Inf. Process. Lett.},
  volume       = {138},
  pages        = {72--74},
  year         = {2018},
  _url          = {https://doi.org/10.1016/j.ipl.2018.06.011},
  doi          = {10.1016/J.IPL.2018.06.011},
  timestamp    = {Sat, 11 Aug 2018 00:57:39 +0200},
  biburl       = {https://dblp.org/rec/journals/ipl/ChanR18.bib},
  bibsource    = {dblp computer science bibliography, https://dblp.org}
}

@article{Chan20,
  author       = {Timothy M. Chan},
  title        = {More Logarithmic-factor Speedups for {3SUM}, (median, +)-convolution,
                  and Some Geometric {3SUM}-hard Problems},
  journal      = {{ACM} Trans. Algorithms},
  volume       = {16},
  number       = {1},
  pages        = {7:1--7:23},
  year         = {2020},
  _url          = {https://doi.org/10.1145/3363541},
  doi          = {10.1145/3363541},
  timestamp    = {Sat, 08 Jan 2022 02:22:06 +0100},
  biburl       = {https://dblp.org/rec/journals/talg/Chan20.bib},
  bibsource    = {dblp computer science bibliography, https://dblp.org},
note = {Preliminary version in SODA 2018}
}

@article{Chan99,
  author       = {Timothy M. Chan},
  title        = {Geometric Applications of a Randomized Optimization Technique},
  journal      = {Discret. Comput. Geom.},
  volume       = {22},
  number       = {4},
  pages        = {547--567},
  year         = {1999},
  _url          = {https://doi.org/10.1007/PL00009478},
  doi          = {10.1007/PL00009478},
  timestamp    = {Thu, 12 Mar 2020 17:21:21 +0100},
  biburl       = {https://dblp.org/rec/journals/dcg/Chan99.bib},
  bibsource    = {dblp computer science bibliography, https://dblp.org}
}

@article{Chan10,
  author       = {Timothy M. Chan},
  title        = {More Algorithms for All-Pairs Shortest Paths in Weighted Graphs},
  journal      = {{SIAM} J. Comput.},
  volume       = {39},
  number       = {5},
  pages        = {2075--2089},
  year         = {2010},
  _url          = {https://doi.org/10.1137/08071990X},
  doi          = {10.1137/08071990X},
  timestamp    = {Wed, 14 Nov 2018 10:45:05 +0100},
  biburl       = {https://dblp.org/rec/journals/siamcomp/Chan10.bib},
  bibsource    = {dblp computer science bibliography, https://dblp.org},
note = {Preliminary version in STOC 2007}
}

@article{Chan12,
  author       = {Timothy M. Chan},
  title        = {All-pairs shortest paths for unweighted undirected graphs in {$o(mn)$}
                  time},
  journal      = {{ACM} Trans. Algorithms},
  volume       = {8},
  number       = {4},
  pages        = {34:1--34:17},
  year         = {2012},
  url          = {https://doi.org/10.1145/2344422.2344424},
  doi          = {10.1145/2344422.2344424},
  bibsource    = {dblp computer science bibliography, https://dblp.org},
note = {Preliminary version in SODA 2006}
}

@inproceedings{AbboudHWW16,
  author       = {Amir Abboud and
                  Thomas Dueholm Hansen and
                  Virginia {Vassilevska Williams} and
                  Ryan Williams},
  _editor       = {Daniel Wichs and
                  Yishay Mansour},
  title        = {Simulating branching programs with edit distance and friends: or:
                  a polylog shaved is a lower bound made},
  booktitle    = {Proc. 48th Annual ACM Symposium on Theory
                  of Computing (STOC)},
  pages        = {375--388},
  _publisher    = {{ACM}},
  year         = {2016},
  _url          = {https://doi.org/10.1145/2897518.2897653},
  doi          = {10.1145/2897518.2897653},
  timestamp    = {Mon, 30 Mar 2026 15:37:35 +0200},
  biburl       = {https://dblp.org/rec/conf/stoc/AbboudHWW16.bib},
  bibsource    = {dblp computer science bibliography, https://dblp.org}
}

@article{AltG95,
  author       = {Helmut Alt and
                  Michael Godau},
  title        = {Computing the {F}r{\'{e}}chet distance between two polygonal curves},
  journal      = {Int. J. Comput. Geom. Appl.},
  volume       = {5},
  pages        = {75--91},
  year         = {1995},
  url          = {https://doi.org/10.1142/S0218195995000064},
  doi          = {10.1142/S0218195995000064},
  bibsource    = {dblp computer science bibliography, https://dblp.org},
note = {Preliminary version in SoCG 1992}
}

@inproceedings{HorstKOS23,
  author       = {Thijs van der Horst and
                  Marc J. van Kreveld and
                  Tim Ophelders and
                  Bettina Speckmann},
  editor       = {Nikhil Bansal and
                  Viswanath Nagarajan},
  title        = {A Subquadratic {$n^\varepsilon$}-approximation
                  for the Continuous {F}r{\'{e}}chet Distance},
  booktitle    = {Proc. 34th {ACM-SIAM} Symposium on Discrete Algorithms (SODA)},
  pages        = {1759--1776},
  _publisher    = {{SIAM}},
  year         = {2023},
  _url          = {https://doi.org/10.1137/1.9781611977554.ch67},
  doi          = {10.1137/1.9781611977554.CH67},
  timestamp    = {Thu, 01 May 2025 20:26:30 +0200},
  biburl       = {https://dblp.org/rec/conf/soda/HorstKOS23.bib},
  bibsource    = {dblp computer science bibliography, https://dblp.org}
}

@inproceedings{BuchinOS17,
  author       = {Kevin Buchin and
                  Tim Ophelders and
                  Bettina Speckmann},
  _editor       = {Philip N. Klein},
  title        = {Computing the {F}r{\'{e}}chet Distance between Real-Valued Surfaces},
  booktitle    = {Proc. 28th Annual {ACM-SIAM} Symposium on Discrete
                  Algorithms (SODA)},
  pages        = {2443--2455},
  _publisher    = {{SIAM}},
  year         = {2017},
  _url          = {https://doi.org/10.1137/1.9781611974782.162},
  doi          = {10.1137/1.9781611974782.162},
  timestamp    = {Tue, 07 May 2024 20:08:16 +0200},
  biburl       = {https://dblp.org/rec/conf/soda/BuchinOS17.bib},
  bibsource    = {dblp computer science bibliography, https://dblp.org}
}

@inproceedings{BringmannKN19,
  author       = {Karl Bringmann and
                  Marvin K{\"{u}}nnemann and
                  Andr{\'{e}} Nusser},
  _editor       = {Gill Barequet and
                  Yusu Wang},
  title        = {Walking the Dog Fast in Practice: Algorithm Engineering of the {F}r{\'{e}}chet
                  Distance},
  booktitle    = {Proc. 35th International Symposium on Computational Geometry (SoCG)},
  _series       = {LIPIcs},
  volume       = {129},
  pages        = {17:1--17:21},
  _publisher    = {Schloss Dagstuhl - Leibniz-Zentrum f{\"{u}}r Informatik},
  year         = {2019},
  _url          = {https://doi.org/10.4230/LIPIcs.SoCG.2019.17},
  doi          = {10.4230/LIPICS.SOCG.2019.17},
  timestamp    = {Fri, 21 Nov 2025 14:06:54 +0100},
  biburl       = {https://dblp.org/rec/conf/compgeom/BringmannKN19.bib},
  bibsource    = {dblp computer science bibliography, https://dblp.org}
}

@inproceedings{Har-PeledRR25,
  author       = {Sariel Har{-}Peled and
                  Benjamin Raichel and
                  Eliot W. Robson},
  _editor       = {Oswin Aichholzer and
                  Haitao Wang},
  title        = {The {F}r{\'{e}}chet Distance Unleashed: Approximating a Dog with
                  a Frog},
  booktitle    = {Proc. 41st International Symposium on Computational Geometry (SoCG)},
  _series       = {LIPIcs},
  volume       = {332},
  pages        = {54:1--54:13},
  _publisher    = {Schloss Dagstuhl - Leibniz-Zentrum f{\"{u}}r Informatik},
  year         = {2025},
  _url          = {https://doi.org/10.4230/LIPIcs.SoCG.2025.54},
  doi          = {10.4230/LIPICS.SOCG.2025.54},
  timestamp    = {Fri, 21 Nov 2025 14:06:54 +0100},
  biburl       = {https://dblp.org/rec/conf/compgeom/Har-PeledRR25.bib},
  bibsource    = {dblp computer science bibliography, https://dblp.org}
}

@article{Har-PeledR14,
  author       = {Sariel Har{-}Peled and
                  Benjamin Raichel},
  title        = {The {F}r{\'{e}}chet distance revisited and extended},
  journal      = {{ACM} Trans. Algorithms},
  volume       = {10},
  number       = {1},
  pages        = {3:1--3:22},
  year         = {2014},
  _url          = {https://doi.org/10.1145/2532646},
  doi          = {10.1145/2532646},
  timestamp    = {Tue, 06 Nov 2018 12:51:20 +0100},
  biburl       = {https://dblp.org/rec/journals/talg/Har-PeledR14.bib},
  bibsource    = {dblp computer science bibliography, https://dblp.org},
note = {Preliminary version in SoCG 2011}
}

@inproceedings{ConradiHR26,
  author       = {Jacobus Conradi and
                  Ivor van der Hoog and
                  Eva Rotenberg},
  _editor       = {Hee{-}Kap Ahn and
                  Michael Hoffmann and
                  Amir Nayyeri},
  title        = {On Computing the (Exact) {F}r{\'{e}}chet Distance with a Frog},
  booktitle    = {Proc. 42nd International Symposium on Computational Geometry (SoCG)},
  _series       = {LIPIcs},
  _volume       = {367},
  pages        = {35:1--35:20},
  publisher    = {Schloss Dagstuhl - Leibniz-Zentrum f{\"{u}}r Informatik},
  year         = {2026},
  _url          = {https://doi.org/10.4230/LIPIcs.SoCG.2026.35},
  doi          = {10.4230/LIPICS.SOCG.2026.35},
  timestamp    = {Thu, 28 May 2026 16:24:28 +0200},
  biburl       = {https://dblp.org/rec/conf/compgeom/ConradiHR26.bib},
  bibsource    = {dblp computer science bibliography, https://dblp.org}
}

@inproceedings{ConradiHHO26,
  author       = {Jacobus Conradi and
                  Ivor van der Hoog and
                  Thijs van der Horst and
                  Tim Ophelders},
  _editor       = {Sepehr Assadi and
                  Eva Rotenberg},
  title        = {Computing the {F}r{\'{e}}chet Distance When Just One Curve is {$c$-packed}:
                  {A} Simple Almost-Tight Algorithm},
  booktitle    = {Proc. 9th SIAM Symposium on Simplicity in Algorithms (SOSA)},
  pages        = {43--55},
  _publisher    = {{SIAM}},
  year         = {2026},
  _url          = {https://doi.org/10.1137/1.9781611978964.4},
  doi          = {10.1137/1.9781611978964.4},
  timestamp    = {Tue, 07 Apr 2026 20:23:41 +0200},
  biburl       = {https://dblp.org/rec/conf/sosa/ConradiHHO26.bib},
  bibsource    = {dblp computer science bibliography, https://dblp.org}
}

@inproceedings{Blank26,
  author       = {Lotte Blank},
  _editor       = {Hee{-}Kap Ahn and
                  Michael Hoffmann and
                  Amir Nayyeri},
  title        = {Fr{\'{e}}chet Distance in the Imbalanced Case},
  booktitle    = {Proc. 42nd International Symposium on Computational Geometry (SoCG)},
  _series       = {LIPIcs},
  _volume       = {367},
  pages        = {17:1--17:17},
  _publisher    = {Schloss Dagstuhl - Leibniz-Zentrum f{\"{u}}r Informatik},
  year         = {2026},
  _url          = {https://doi.org/10.4230/LIPIcs.SoCG.2026.17},
  doi          = {10.4230/LIPICS.SOCG.2026.17},
  timestamp    = {Thu, 28 May 2026 16:24:28 +0200},
  biburl       = {https://dblp.org/rec/conf/compgeom/Blank26.bib},
  bibsource    = {dblp computer science bibliography, https://dblp.org}
}

@inproceedings{ColombeF21,
  author       = {Connor Colombe and
                  Kyle Fox},
  _editor       = {Kevin Buchin and
                  {\'{E}}ric Colin de Verdi{\`{e}}re},
  title        = {Approximating the (Continuous) {F}r{\'{e}}chet Distance},
  booktitle    = {Proc. 37th International Symposium on Computational Geometry (SoCG)},
  _series       = {LIPIcs},
  _volume       = {189},
  pages        = {26:1--26:14},
  _publisher    = {Schloss Dagstuhl - Leibniz-Zentrum f{\"{u}}r Informatik},
  year         = {2021},
  _url          = {https://doi.org/10.4230/LIPIcs.SoCG.2021.26},
  doi          = {10.4230/LIPICS.SOCG.2021.26},
  timestamp    = {Fri, 21 Nov 2025 14:06:54 +0100},
  biburl       = {https://dblp.org/rec/conf/compgeom/ColombeF21.bib},
  bibsource    = {dblp computer science bibliography, https://dblp.org}
}

@article{BringmannKN21,
  author       = {Karl Bringmann and
                  Marvin K{\"{u}}nnemann and
                  Andr{\'{e}} Nusser},
  title        = {Discrete {F}r{\'{e}}chet Distance under Translation: Conditional
                  Hardness and an Improved Algorithm},
  journal      = {{ACM} Trans. Algorithms},
  volume       = {17},
  number       = {3},
  pages        = {25:1--25:42},
  year         = {2021},
  _url          = {https://doi.org/10.1145/3460656},
  doi          = {10.1145/3460656},
  timestamp    = {Thu, 12 Aug 2021 08:58:06 +0200},
  biburl       = {https://dblp.org/rec/journals/talg/BringmannKN21.bib},
  bibsource    = {dblp computer science bibliography, https://dblp.org},
  note = {Preliminary version in SODA 2019}
}

@article{DriemelHW12,
  author       = {Anne Driemel and
                  Sariel Har{-}Peled and
                  Carola Wenk},
  title        = {Approximating the {F}r{\'{e}}chet Distance for Realistic Curves
                  in Near Linear Time},
  journal      = {Discret. Comput. Geom.},
  volume       = {48},
  number       = {1},
  pages        = {94--127},
  year         = {2012},
  _url          = {https://doi.org/10.1007/s00454-012-9402-z},
  doi          = {10.1007/S00454-012-9402-Z},
  timestamp    = {Tue, 07 May 2024 20:26:22 +0200},
  biburl       = {https://dblp.org/rec/journals/dcg/DriemelHW12.bib},
  bibsource    = {dblp computer science bibliography, https://dblp.org},
  note = {Preliminary version in SoCG 2010}
}

@article{DriemelH13,
  author       = {Anne Driemel and
                  Sariel Har{-}Peled},
  title        = {Jaywalking Your Dog: Computing the {F}r{\'{e}}chet Distance with
                  Shortcuts},
  journal      = {{SIAM} J. Comput.},
  volume       = {42},
  number       = {5},
  pages        = {1830--1866},
  year         = {2013},
  _url          = {https://doi.org/10.1137/120865112},
  doi          = {10.1137/120865112},
  timestamp    = {Tue, 07 May 2024 20:21:44 +0200},
  biburl       = {https://dblp.org/rec/journals/siamcomp/DriemelH13.bib},
  bibsource    = {dblp computer science bibliography, https://dblp.org},
  note = {Preliminary version in SODA 2012}
}

@article{Har-PeledNSS16,
  author       = {Sariel Har{-}Peled and
                  Amir Nayyeri and
                  Mohammad R. Salavatipour and
                  Anastasios Sidiropoulos},
  title        = {How to Walk Your Dog in the Mountains with No Magic Leash},
  journal      = {Discret. Comput. Geom.},
  volume       = {55},
  number       = {1},
  pages        = {39--73},
  year         = {2016},
  _url          = {https://doi.org/10.1007/s00454-015-9737-3},
  doi          = {10.1007/S00454-015-9737-3},
  timestamp    = {Thu, 12 Mar 2020 17:20:54 +0100},
  biburl       = {https://dblp.org/rec/journals/dcg/Har-PeledNSS16.bib},
  bibsource    = {dblp computer science bibliography, https://dblp.org},
  note = {Preliminary version in SoCG 2012}
}

@inproceedings{HorstKOS25,
  author       = {Thijs van der Horst and
                  Marc J. van Kreveld and
                  Tim Ophelders and
                  Bettina Speckmann},
  _editor       = {Anne Benoit and
                  Haim Kaplan and
                  Sebastian Wild and
                  Grzegorz Herman},
  title        = {The Geodesic {F}r{\'{e}}chet Distance Between Two Curves Bounding
                  a Simple Polygon},
  booktitle    = {Proc. 33rd Annual European Symposium on Algorithms (ESA)},
  _series       = {LIPIcs},
  volume       = {351},
  pages        = {35:1--35:16},
  _publisher    = {Schloss Dagstuhl - Leibniz-Zentrum f{\"{u}}r Informatik},
  year         = {2025},
  _url          = {https://doi.org/10.4230/LIPIcs.ESA.2025.35},
  doi          = {10.4230/LIPICS.ESA.2025.35},
  timestamp    = {Fri, 21 Nov 2025 23:44:11 +0100},
  biburl       = {https://dblp.org/rec/conf/esa/HorstKOS25.bib},
  bibsource    = {dblp computer science bibliography, https://dblp.org}
}

@article{AvrahamFKKS15,
  author       = {Rinat {Ben Avraham} and
                  Omrit Filtser and
                  Haim Kaplan and
                  Matthew J. Katz and
                  Micha Sharir},
  title        = {The Discrete and Semicontinuous {F}r{\'{e}}chet Distance with Shortcuts
                  via Approximate Distance Counting and Selection},
  journal      = {{ACM} Trans. Algorithms},
  volume       = {11},
  number       = {4},
  pages        = {29:1--29:29},
  year         = {2015},
  _url          = {https://doi.org/10.1145/2700222},
  doi          = {10.1145/2700222},
  timestamp    = {Sun, 02 Oct 2022 15:49:59 +0200},
  biburl       = {https://dblp.org/rec/journals/talg/AvrahamFKKS15.bib},
  bibsource    = {dblp computer science bibliography, https://dblp.org},
  note = {Preliminary version in SoCG 2014}
}

@article{FiltserGMP24,
  author       = {Omrit Filtser and
                  Mayank Goswami and
                  Joseph S. B. Mitchell and
                  Valentin Polishchuk},
  title        = {On Flipping the {F}r{\'{e}}chet Distance},
  journal      = {Algorithmica},
  volume       = {86},
  number       = {12},
  pages        = {3629--3652},
  year         = {2024},
  _url          = {https://doi.org/10.1007/s00453-024-01267-8},
  doi          = {10.1007/S00453-024-01267-8},
  timestamp    = {Wed, 04 Dec 2024 22:33:38 +0100},
  biburl       = {https://dblp.org/rec/journals/algorithmica/FiltserGMP24.bib},
  bibsource    = {dblp computer science bibliography, https://dblp.org},
  note = {Preliminary version in ITCS 2023}
}

@inproceedings{AfshaniD18,
  author       = {Peyman Afshani and
                  Anne Driemel},
  _editor       = {Artur Czumaj},
  title        = {On the complexity of range searching among curves},
  booktitle    = {Proc. 29th Annual {ACM-SIAM} Symposium on Discrete
                  Algorithms (SODA)},
  pages        = {898--917},
  publisher    = {{SIAM}},
  year         = {2018},
  _url          = {https://doi.org/10.1137/1.9781611975031.58},
  doi          = {10.1137/1.9781611975031.58},
  timestamp    = {Tue, 07 May 2024 20:08:16 +0200},
  biburl       = {https://dblp.org/rec/conf/soda/AfshaniD18.bib},
  bibsource    = {dblp computer science bibliography, https://dblp.org}
}

@article{MasekP80,
  author       = {William J. Masek and
                  Mike Paterson},
  title        = {A Faster Algorithm Computing String Edit Distances},
  journal      = {J. Comput. Syst. Sci.},
  volume       = {20},
  number       = {1},
  pages        = {18--31},
  year         = {1980},
  _url          = {https://doi.org/10.1016/0022-0000(80)90002-1},
  doi          = {10.1016/0022-0000(80)90002-1},
  timestamp    = {Tue, 16 Feb 2021 14:04:42 +0100},
  biburl       = {https://dblp.org/rec/journals/jcss/MasekP80.bib},
  bibsource    = {dblp computer science bibliography, https://dblp.org}
}

@article{GronlundP18,
  author       = {Allan Gr{\o}nlund and
                  Seth Pettie},
  title        = {Threesomes, Degenerates, and Love Triangles},
  journal      = {J. {ACM}},
  volume       = {65},
  number       = {4},
  pages        = {22:1--22:25},
  year         = {2018},
  _url          = {https://doi.org/10.1145/3185378},
  doi          = {10.1145/3185378},
  timestamp    = {Wed, 21 Nov 2018 12:44:29 +0100},
  biburl       = {https://dblp.org/rec/journals/jacm/GronlundP18.bib},
  bibsource    = {dblp computer science bibliography, https://dblp.org},
  note = {Preliminary version in FOCS 2014}
}

@article{BilleF08,
  author       = {Philip Bille and
                  Martin Farach{-}Colton},
  title        = {Fast and compact regular expression matching},
  journal      = {Theor. Comput. Sci.},
  volume       = {409},
  number       = {3},
  pages        = {486--496},
  year         = {2008},
  _url          = {https://doi.org/10.1016/j.tcs.2008.08.042},
  doi          = {10.1016/J.TCS.2008.08.042},
  timestamp    = {Fri, 09 Apr 2021 18:27:58 +0200},
  biburl       = {https://dblp.org/rec/journals/tcs/BilleF08.bib},
  bibsource    = {dblp computer science bibliography, https://dblp.org}
}

@article{Grabowski16,
  author       = {Szymon Grabowski},
  title        = {New tabulation and sparse dynamic programming based techniques for
                  sequence similarity problems},
  journal      = {Discret. Appl. Math.},
  volume       = {212},
  pages        = {96--103},
  year         = {2016},
  _url          = {https://doi.org/10.1016/j.dam.2015.10.040},
  doi          = {10.1016/J.DAM.2015.10.040},
  timestamp    = {Thu, 11 Feb 2021 23:26:28 +0100},
  biburl       = {https://dblp.org/rec/journals/dam/Grabowski16.bib},
  bibsource    = {dblp computer science bibliography, https://dblp.org}
}

@article{BackursI18,
  author       = {Arturs Backurs and
                  Piotr Indyk},
  title        = {Edit Distance Cannot Be Computed in Strongly Subquadratic Time (Unless
                  {SETH} is False)},
  journal      = {{SIAM} J. Comput.},
  volume       = {47},
  number       = {3},
  pages        = {1087--1097},
  year         = {2018},
  _url          = {https://doi.org/10.1137/15M1053128},
  doi          = {10.1137/15M1053128},
  timestamp    = {Sun, 19 Jan 2025 14:37:01 +0100},
  biburl       = {https://dblp.org/rec/journals/siamcomp/BackursI18.bib},
  bibsource    = {dblp computer science bibliography, https://dblp.org},
  note = {Preliminary version in STOC 2015}
}

@inproceedings{AbboudBW15,
  author       = {Amir Abboud and
                  Arturs Backurs and
                  Virginia {Vassilevska Williams}},
  _editor       = {Venkatesan Guruswami},
  title        = {Tight Hardness Results for {LCS} and Other Sequence Similarity Measures},
  booktitle    = {Proc. 56th Annual {IEEE} Symposium on Foundations of Computer Science (FOCS)},
  pages        = {59--78},
  _publisher    = {{IEEE} Computer Society},
  year         = {2015},
  _url          = {https://doi.org/10.1109/FOCS.2015.14},
  doi          = {10.1109/FOCS.2015.14},
  timestamp    = {Tue, 08 Jul 2025 16:38:00 +0200},
  biburl       = {https://dblp.org/rec/conf/focs/AbboudBW15.bib},
  bibsource    = {dblp computer science bibliography, https://dblp.org}
}

@inproceedings{BringmannK15,
  author       = {Karl Bringmann and
                  Marvin K{\"{u}}nnemann},
  _editor       = {Venkatesan Guruswami},
  title        = {Quadratic Conditional Lower Bounds for String Problems and Dynamic
                  Time Warping},
  booktitle    = {Proc. 56th Annual {IEEE} Symposium on Foundations of Computer Science (FOCS)},
  pages        = {79--97},
  _publisher    = {{IEEE} Computer Society},
  year         = {2015},
  _url          = {https://doi.org/10.1109/FOCS.2015.15},
  doi          = {10.1109/FOCS.2015.15},
  timestamp    = {Tue, 08 Jul 2025 16:38:00 +0200},
  biburl       = {https://dblp.org/rec/conf/focs/BringmannK15.bib},
  bibsource    = {dblp computer science bibliography, https://dblp.org}
}

@article{FederM95,
  author       = {Tom{\'{a}}s Feder and
                  Rajeev Motwani},
  title        = {Clique Partitions, Graph Compression and Speeding-Up Algorithms},
  journal      = {J. Comput. Syst. Sci.},
  volume       = {51},
  number       = {2},
  pages        = {261--272},
  year         = {1995},
  _url          = {https://doi.org/10.1006/jcss.1995.1065},
  doi          = {10.1006/JCSS.1995.1065},
  timestamp    = {Tue, 16 Feb 2021 14:04:13 +0100},
  biburl       = {https://dblp.org/rec/journals/jcss/FederM95.bib},
  bibsource    = {dblp computer science bibliography, https://dblp.org},
  note = {Preliminary version in STOC 1991}
}

@article{Fredman76,
  author       = {Michael L. Fredman},
  title        = {New Bounds on the Complexity of the Shortest Path Problem},
  journal      = {{SIAM} J. Comput.},
  volume       = {5},
  number       = {1},
  pages        = {83--89},
  year         = {1976},
  _url          = {https://doi.org/10.1137/0205006},
  doi          = {10.1137/0205006},
  timestamp    = {Sat, 27 May 2017 14:22:58 +0200},
  biburl       = {https://dblp.org/rec/journals/siamcomp/Fredman76.bib},
  bibsource    = {dblp computer science bibliography, https://dblp.org},
  note = {Preliminary version in FOCS 1975}
}

@inproceedings{Chan15,
  author       = {Timothy M. Chan},
  _editor       = {Piotr Indyk},
  title        = {Speeding up the {F}our {R}ussians Algorithm by About One More Logarithmic
                  Factor},
  booktitle    = {Proc. 26th Annual {ACM-SIAM} Symposium on Discrete
                  Algorithms (SODA)},
  pages        = {212--217},
  _publisher    = {{SIAM}},
  year         = {2015},
  _url          = {https://doi.org/10.1137/1.9781611973730.16},
  doi          = {10.1137/1.9781611973730.16},
  timestamp    = {Tue, 02 Feb 2021 17:07:43 +0100},
  biburl       = {https://dblp.org/rec/conf/soda/Chan15.bib},
  bibsource    = {dblp computer science bibliography, https://dblp.org}
}

@article{Yu18,
  author       = {Huacheng Yu},
  title        = {An improved combinatorial algorithm for {B}oolean matrix multiplication},
  journal      = {Inf. Comput.},
  volume       = {261},
  pages        = {240--247},
  year         = {2018},
  url          = {https://doi.org/10.1016/j.ic.2018.02.006},
  doi          = {10.1016/J.IC.2018.02.006},
  bibsource    = {dblp computer science bibliography, https://dblp.org},
  note = {Preliminary version in ICALP 2015}
}

@article{ArlazarovDKF70,
author = {V. Z. Arlazarov and E. A. Dinic and M. A. Kronrod and I. A. Faradzhev},
title = {On economical construction of the transitive closure of a directed graph},
journal = {Sov. Math. Dokl.},
volume = {11},
number = {5},
pages = {1209--1210},
year = {1970}
}

@article{BarbaCILOS19,
  author       = {Luis Barba and
                  Jean Cardinal and
                  John Iacono and
                  Stefan Langerman and
                  Aur{\'{e}}lien Ooms and
                  Noam Solomon},
  title        = {Subquadratic Algorithms for Algebraic {3SUM}},
  journal      = {Discret. Comput. Geom.},
  volume       = {61},
  number       = {4},
  pages        = {698--734},
  year         = {2019},
  _url          = {https://doi.org/10.1007/s00454-018-0040-y},
  doi          = {10.1007/S00454-018-0040-Y},
  timestamp    = {Tue, 14 Oct 2025 19:43:25 +0200},
  biburl       = {https://dblp.org/rec/journals/dcg/BarbaCILOS19.bib},
  bibsource    = {dblp computer science bibliography, https://dblp.org},
  note = {Preliminary version in SoCG 2017}
}

@article{EzraHKS20,
  author       = {Esther Ezra and
                  Sariel Har{-}Peled and
                  Haim Kaplan and
                  Micha Sharir},
  title        = {Decomposing Arrangements of Hyperplanes: {VC}-Dimension, Combinatorial
                  Dimension, and Point Location},
  journal      = {Discret. Comput. Geom.},
  volume       = {64},
  number       = {1},
  pages        = {109--173},
  year         = {2020},
  _url          = {https://doi.org/10.1007/s00454-019-00141-7},
  doi          = {10.1007/S00454-019-00141-7},
  timestamp    = {Tue, 07 May 2024 20:26:22 +0200},
  biburl       = {https://dblp.org/rec/journals/dcg/EzraHKS20.bib},
  bibsource    = {dblp computer science bibliography, https://dblp.org}
}

@article{Meiser93,
  author       = {Stefan Meiser},
  title        = {Point Location in Arrangements of Hyperplanes},
  journal      = {Inf. Comput.},
  volume       = {106},
  number       = {2},
  pages        = {286--303},
  year         = {1993},
  _url          = {https://doi.org/10.1006/inco.1993.1057},
  doi          = {10.1006/INCO.1993.1057},
  timestamp    = {Fri, 12 Feb 2021 22:16:57 +0100},
  biburl       = {https://dblp.org/rec/journals/iandc/Meiser93.bib},
  bibsource    = {dblp computer science bibliography, https://dblp.org}
}

@article{AgarwalASS93,
  author       = {Pankaj K. Agarwal and
                  Boris Aronov and
                  Micha Sharir and
                  Subhash Suri},
  title        = {Selecting Distances in the Plane},
  journal      = {Algorithmica},
  volume       = {9},
  number       = {5},
  pages        = {495--514},
  year         = {1993},
  url          = {https://doi.org/10.1007/BF01187037},
  doi          = {10.1007/BF01187037},
  bibsource    = {dblp computer science bibliography, https://dblp.org}
}

@article{ChanZ24,
  author       = {Timothy M. Chan and
                  Da Wei Zheng},
  title        = {Hopcroft's Problem, Log$^*$ Shaving, Two-dimensional Fractional
                  Cascading, and Decision Trees},
  journal      = {{ACM} Trans. Algorithms},
  volume       = {20},
  number       = {3},
  pages        = {24},
  year         = {2024},
  _url          = {https://doi.org/10.1145/3591357},
  doi          = {10.1145/3591357},
  timestamp    = {Fri, 19 Jul 2024 23:17:24 +0200},
  biburl       = {https://dblp.org/rec/journals/talg/ChanZ24.bib},
  bibsource    = {dblp computer science bibliography, https://dblp.org},
  note = {Preliminary version in SODA 2022}
}

@incollection{AgarwalE98,
author = {Pankaj K. Agarwal and Jeff Erickson},
title = {Geometric range searching and its relatives},
editor = {Bernard Chazelle and Jacob E. Goodman and Richard Pollack},
booktitle = {Advances in Discrete and Computational Geometry},
publisher = {AMS Press},
year = {1998},
pages = {1--56},
url = {https://jeffe.cs.illinois.edu/pubs/pdf/survey.pdf}
}

@incollection{AgarwalS00a,
  author       = {Pankaj K. Agarwal and
                  Micha Sharir},
  editor       = {J{\"{o}}rg{-}R{\"{u}}diger Sack and
                  Jorge Urrutia},
  title        = {Arrangements and Their Applications},
  booktitle    = {Handbook of Computational Geometry},
  pages        = {49--119},
  publisher    = {North Holland / Elsevier},
  year         = {2000},
  _url          = {https://doi.org/10.1016/b978-044482537-7/50003-6},
  doi          = {10.1016/B978-044482537-7/50003-6},
  timestamp    = {Fri, 28 Jun 2019 09:10:30 +0200},
  biburl       = {https://dblp.org/rec/books/el/00/AgarwalS00a.bib},
  bibsource    = {dblp computer science bibliography, https://dblp.org}
}

@article{ChazelleEGS91,
  author       = {Bernard Chazelle and
                  Herbert Edelsbrunner and
                  Leonidas J. Guibas and
                  Micha Sharir},
  title        = {A Singly Exponential Stratification Scheme for Real Semi-Algebraic
                  Varieties and its Applications},
  journal      = {Theor. Comput. Sci.},
  volume       = {84},
  number       = {1},
  pages        = {77--105},
  year         = {1991},
  _url          = {https://doi.org/10.1016/0304-3975(91)90261-Y},
  doi          = {10.1016/0304-3975(91)90261-Y},
  timestamp    = {Wed, 17 Feb 2021 22:00:06 +0100},
  biburl       = {https://dblp.org/rec/journals/tcs/ChazelleEGS91.bib},
  bibsource    = {dblp computer science bibliography, https://dblp.org}
}

@article{Milnor64,
author = {John W. Milnor},
title = {On the {B}etti numbers of real algebraic varieties},
journal = {Proc. Amer. Math. Soc.},
article = {15},
year = {1964},
pages = {275--280}
}

@incollection{Thom65,
author = {Ren\'e Thom},
title = {Sur l'homologie des vari\'et\'es alg\'ebriques re\'elles},
booktitle = {Differential and Combinatorial Topology},
editor = {S. S. Cairns},
publisher = {Princeton Univ. Press},
year = {1965}
}

@article{Megiddo83,
  author       = {Nimrod Megiddo},
  title        = {Applying Parallel Computation Algorithms in the Design of Serial Algorithms},
  journal      = {J. {ACM}},
  volume       = {30},
  number       = {4},
  pages        = {852--865},
  year         = {1983},
  _url          = {https://doi.org/10.1145/2157.322410},
  doi          = {10.1145/2157.322410},
  timestamp    = {Sun, 02 Jun 2019 20:59:12 +0200},
  biburl       = {https://dblp.org/rec/journals/jacm/Megiddo83.bib},
  bibsource    = {dblp computer science bibliography, https://dblp.org}
}

@article{BentleyO79,
  author       = {Jon Louis Bentley and
                  Thomas Ottmann},
  title        = {Algorithms for Reporting and Counting Geometric Intersections},
  journal      = {{IEEE} Trans. Computers},
  volume       = {28},
  number       = {9},
  pages        = {643--647},
  year         = {1979},
  _url          = {https://doi.org/10.1109/TC.1979.1675432},
  doi          = {10.1109/TC.1979.1675432},
  timestamp    = {Sun, 02 Jun 2019 21:05:11 +0200},
  biburl       = {https://dblp.org/rec/journals/tc/BentleyO79.bib},
  bibsource    = {dblp computer science bibliography, https://dblp.org}
}

@article{AlbersH97,
  author       = {Susanne Albers and
                  Torben Hagerup},
  title        = {Improved Parallel Integer Sorting without Concurrent Writing},
  journal      = {Inf. Comput.},
  volume       = {136},
  number       = {1},
  pages        = {25--51},
  year         = {1997},
  _url          = {https://doi.org/10.1006/inco.1997.2632},
  doi          = {10.1006/INCO.1997.2632},
  timestamp    = {Fri, 12 Feb 2021 22:16:36 +0100},
  biburl       = {https://dblp.org/rec/journals/iandc/AlbersH97.bib},
  bibsource    = {dblp computer science bibliography, https://dblp.org}
}

@article{BenAvrahamKS15,
  author = {Rinat {Ben Avraham} and Haim Kaplan and Micha Sharir},
  title = {A Faster Algorithm for the Discrete {F}r{\'{e}}chet Distance under Translation},
  journal = {CoRR},
  volume = {abs/1501.03724},
  year = {2015},
  url = {https://arxiv.org/abs/1501.03724},
  eprint = {1501.03724},
  archivePrefix = {arXiv}
}
}

\end{document}